\documentclass[
a4paper,twocolumn,11pt,accepted=2023-12-03]{quantumarticle}
\pdfoutput=1
\usepackage[utf8]{inputenc}
\usepackage[english]{babel}
\usepackage[T1]{fontenc}

\usepackage[pdftex]{graphicx}
\usepackage{amsmath,amsfonts,amsbsy,amssymb,amsthm,mathtools}
\usepackage{mathrsfs}
\usepackage{epstopdf}
\usepackage{ulem,color}
\usepackage{enumitem}
\usepackage{siunitx}
\usepackage{amssymb,amsmath}
\usepackage{ulem,color}

\usepackage{pgfplots}
\usepackage{pstool}
\usepackage{tikzsymbols}
\usetikzlibrary{plotmarks,spy}
\usepackage[font=small,labelfont=bf,justification=justified,format=plain]{caption}

\newtheorem{prop}{Proposition}
\newtheorem{lemma}{Lemma}

\theoremstyle{definition}
\newtheorem{definition}{Definition}
\newtheorem{conjecture}{Conjecture}

\newcommand{\bbsout}[1]{}

\newcommand{\bsout}[1]{}
\usepackage{siunitx}
\usepackage{hyperref}
\usepackage[T1]{fontenc}
\usepackage{mwe}
\newcommand{\bra}[1]{\langle #1 |}
\newcommand{\ket}[1]{| #1 \rangle}
\newcommand{\ketbra}[2]{\left|{#1}\right\>\mkern-7mu\left\<{#2}\right|}
\newcommand{\tr}[1]{\textrm{Tr}\left[ #1 \right]}

\newcommand{\mg}[1]{\textcolor{blue}{[MG: #1]}}

\renewcommand{\>}{\rangle}
\newcommand{\<}{\langle}
\newcommand{\N}{{\mathbb{N}}} %natural numbers
\newcommand{\C}{{\mathbb{C}}} %complex numbers
\newcommand{\R}{{\mathbb{R}}} %natural numbers

\renewcommand{\th}{{\mathrm{th}}}
\newcommand{\res}{{\mathrm{res}}}
\newcommand{\env}{{\mathrm{env}}}
\newcommand{\dts}{{\mathrm{dts}}}

\newcommand{\fwidth}{0.75\columnwidth}
\newcommand{\fheight}{0.463535\columnwidth}

\renewcommand{\H}{\mathcal{H}}

\newcommand{\ANU}{Centre for Quantum Computation and Communication Technology, Department of Quantum Science,
Research School of Physics and Engineering, Australian National University, Canberra ACT, Australia 2601.\looseness=-1}
\newcommand{\NTU}{Nanyang Quantum Hub, School of Physical and Mathematical Sciences, Nanyang Technological University, Singapore 639673.\looseness=-1}
\newcommand{\CQT}{Centre for Quantum Technologies,
National University of Singapore, 3 Science Drive 2, Singapore 117543.\looseness=-1}

\newcommand{\IMRE}{A*STAR Quantum Innovation Centre (Q.InC), Institute of Materials Research and Engineering (IMRE), Agency for Science Technology and Research (A*STAR), 2 Fusionopolis Way, 08-03 Innovis, Singapore 138634.\looseness=-1}

\newcommand{\Majulab}{CNRS-UNS-NUS-NTU International Joint Research Unit, UMI 3654, Singapore 117543.\looseness=-1}
\newcommand{\QINC}{A*STAR Quantum Innovation Centre (Q.InC), Institute of High Performance Computing (IHPC), Agency for Science, Technology and Research (A*STAR), Singapore.\looseness=-1}

\pgfplotsset{compat=newest,height=8cm,width=8cm}

\tikzstyle{allfigs}=[line width=1pt,samples=100,domain=0:20]

\pgfplotsset{
  table/search path={data,data/finiteT,data/vacRes},
  myaxis/.style=
  {    width=\fwidth,
    height=\fheight,
    at={(0,0)},
scale only axis,
xmin=0,
xmax=20,
xlabel={$|\alpha_\text{max}|^2$},
ymin=0,
ymax=6,
ylabel={Mutual information},
legend pos=north west,
legend style={legend cell align=left,align=left,
  fill=white,fill opacity=0.8,
  draw=black,text opacity=1}}}

\pgfplotsset{
  table/search path={data,data/finiteT,data/vacRes},
  mywigaxis/.style=
  {
legend pos=north east,
legend style={legend cell align=left,align=left,
  fill=white,fill opacity=0.8,
  draw=black,text opacity=1}}}

\usepackage{tikz}
\usetikzlibrary{quantikz}
\usepackage{subcaption}
\newcommand{\at}[1]{{\textcolor{orange}{[#1]}}}

\renewcommand{\at}[1]{}
\renewcommand{\mg}[1]{}

\begin{document}

\normalem

\newlength\figHeight 
\newlength\figWidth 

% \title{Optimal encoding of classical information on passive linear thermal operations}
\title{Quantum-optimal information encoding using noisy passive linear optics}

\author{Andrew Tanggara}
\email{andrew.tanggara@gmail.com}
\affiliation{\CQT}
\affiliation{\NTU}

\author{Ranjith Nair}
\affiliation{\NTU}

\author{Syed Assad}
\affiliation{\ANU}
\affiliation{\IMRE}

\author{Varun Narasimhachar}
\affiliation{\QINC}
\affiliation{\NTU}

\author{Spyros Tserkis}
\affiliation{\ANU}

\author{Jayne Thompson}
\affiliation{\QINC}

\author{Ping Koy Lam}
\affiliation{\ANU}
\affiliation{\IMRE}

\author{Mile Gu}
\email{mgu@quantumcomplexity.org}
\affiliation{\NTU}
\affiliation{\CQT}
\affiliation{\Majulab}

% \date{\today}
% \date{April 29, 2023}
\maketitle

\begin{abstract}
The amount of information that a noisy channel can transmit has been one of the primary subjects of interest in information theory.
In this work we consider a practically-motivated family of optical quantum channels that can be implemented without an external energy source.
% We extend this study to the quantum regime for a practically-motivated family of optical quantum channels we call the thermal channels, simply consisting of a mixing interaction with the environment in a thermal state at a certain temperature and a phase-shift operation, hence requiring no external energy source.
We optimize the Holevo information over procedures that encode information in attenuations and phase-shifts applied by these channels on a resource state of finite energy.
% We optimize the Holevo information over procedures that encode information in attenuations and phase-shifts applied to a given finite-energy input state.
% Information is encoded in the thermal channel's attenuation and phase-shift which are distributed according to some encoding procedure and it is obtained by sending a quantum system in a certain state into the channel then measuring its output.
% As the measure of information capacity of a quantum channel, we investigate conditions to maximize the Holevo information of the thermal channel which quantifies its information capacity.
It is shown that for any given input state and environment temperature, the maximum Holevo information can be achieved by an encoding procedure that uniformly distributes the channel’s phase-shift parameter.
Moreover for large families of input states, any maximizing encoding scheme has a finite number of channel attenuation values, simplifying the codewords to a finite number of rings around the origin in the output phase space.
% Moreover for large families of input states, any maximizing encoding has a finite support over the channel attenuation, nicely forming a finite number of rings around the origin in the phase space.
The above results and numerical evidence suggests that this property holds for all resource states.
Our results are directly applicable to the quantum reading of an optical memory in the presence of environmental thermal noise.
\end{abstract}

\section{Introduction}

For a classical channel between a sender (Alice) and receiver (Bob), the rate at which information can be noiselessly transmitted from Alice and Bob for a given probability  distribution of the input symbols equals the Shannon mutual information between the input and output \cite{shannon_1948mathematical}.
Maximizing the mutual information over all input probability distributions gives the capacity of the channel, which can be achieved using encoding and decoding over long codes.
Quantum-mechanically, the analog of an input distribution is an ensemble of quantum states indexed by Alice's input alphabet, and the Holevo information of the ensemble \cite{Holevo1973} replaces the mutual information in the classical case.
If Alice encodes her messages in suitable sequences of states and Bob can make measurement on all received output states at once, communication at a rate equals to the Holevo information is possible with arbitrarily small error probability \cite{Holevo1998,Schumacher1997}. 
Optimizing over the input probability distribution thus gives the ultimate capacity of encoding information in the given set of quantum states.

% The quantity of information that is contained in a sequence of symbols has been the center of information theory since its conception by Claude Shannon \cite{shannon_1948mathematical}.
% For an ideal channel connecting a sender (Alice) and a receiver (Bob), this quantity is given by the Shannon entropy of the ensemble of symbols that Alice sends.
% In the case of a quantum channel where Alice encodes her sequence of symbols in a quantum state, this quantity is the von Neumann entropy of the corresponding quantum state.
% Given a particular quantum encoding scheme, the maximum amount of information that Alice can transmit to Bob is known as the Holevo information \cite{Holevo1973}, which can be achieved by some ideal encoding and decoding scheme \cite{Holevo1998,Schumacher1997}.
% This quantity depends on the ensemble of quantum states that Alice uses to encode her message, but assumes that Bob can perform any arbitrary measurement, including a collective measurement.

\mg{Should reorder by motivating the big question first, e.g. how much information can we store encode onto a state without use of external energy? After convincing readers why this question is interesting, we can then continue about how in this paper we can address the problem using Thermal operations.}
\at{OK, did this somewhat, but not sure if it is good enough.}
% \at{Perhaps, motivation would be smth like why do we want to use thermal encoding? as opposed to other type of encoding? why is this important? applications? quantum reading? \cite{Guha2013,Guha2011,Wilde2012} consider the case of $T=0$, i.e: mixing with vacuum, quantum reading uses the same "thermal channel" but with $T=0$, i.e. vacuum state environment. So, channel capacity for the case of $T=0$ has been studied more whereas the case of $T>0$ is less known. Meanwhile, the set of operations that made up the more general case of $T\geq 0$, which made out of mixing operations with thermal state and phase-rotation operations, has been studied in the context of thermodynamics resource theory \cite{Narasimhachar2019} and algorithmic cooling \cite{serafini2020gaussian}.}

Optical quantum channels are an important class of quantum channels for which the Holevo information has been very well studied in the context of different tasks, such as communication \cite{Giovannetti2014}, optical memory reading \cite{Hsu2006,Pirandola2011,Pirandola2011_2,Guha2011,Wilde2012,Guha2013}, algorithmic cooling \cite{serafini2020gaussian}, as well as in pure-loss quantum optical channels \cite{giovannetti2004classical}.
% When one further consider the interplay between the energy cost of performing these tasks and the Holevo information, where the Holevo information increases with the energy constraint on the corresponding ensemble of quantum systems, sources of energy throughout the protocol emerges as a resource.
Given its many applications, we focus on optical quantum channels that separates out the sources of energy, thus rendering energy as a resource in performing the task at hand.
This notion of energy as resource has been studied in \cite{Narasimhachar2019} for thermodynamic tasks.
% separates the \emph{resource} components in the protocol acting as a source of energy from the \emph{free} components, which does not act as an energy source.
% Hence, the difficulty of generating the quantum states of the output ensemble can be reduced to the difficulty of generating a fixed resource state from which the ensemble can be generated using free operations.

\begin{figure}[t]
    \centering
    \includegraphics[width=1.2\columnwidth]{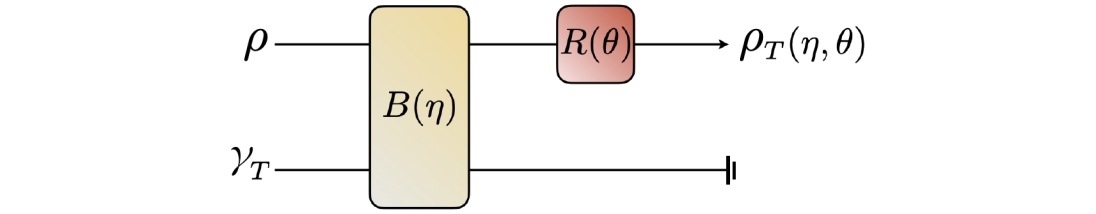}
    \caption{\textbf{Thermal channel circuit representation.} Input "resource" state $\rho$ is mixed with an "environment" in a thermal state $\gamma_T$ by a beamsplitter $B_\eta$ of transmittance $\eta$ and then undergoes a $\theta$ phase-rotation operation $R_\theta$, giving an output codeword state $\rho_T(\eta,\theta)$.}
    \label{fig:thermal_channel}
\end{figure}

% For an optical quantum channel, these feasibly-implementable free components corresponds to the passive linear-optical operations, consisting of phase-shift operations and operations that mix a quantum system in some state with the environment at a certain temperature.
% This energy-cost consideration motivates our focus in this work on the so-called
In what follows, we consider encoding information by optimally applying passive linear-optical operations on a given finite-energy resource state $\rho$.
Such operations, which we call \emph{thermal channels} (illustrated in Fig~\ref{fig:thermal_channel}) are optical quantum channels that can be constructed simply by mixing the given input state with the environment at a certain temperature, followed by a phase-shift operation.
% These channels act on a finitely \emph{peak-energy} constrained resource quantum system which allows one to encode information in the channel attenuation and phase parameters that govern the mixing operation and phase-shift operation, respectively.
The average energy of every state of such output ensembles is bounded above.
Such peak energy constraints have physical meaning as the greatest amount of energy that can be tolerated by a given device or channel. 
In addition, technological limitations may impose such constraints -- for example, it is very difficult to generate Fock states of large occupation number.
Even for coherent-state sources, such constraints can be severe in some applications. 
For example, a satellite-based laser communication system may be highly constrained in its energy output by payload limitations and energy scarcity in space.
Note that a peak energy constraint on an ensemble is more restrictive than an overall average energy constraint, which has been the subject of earlier work.
Indeed, the capacity-achieving ensembles for communication on a large class of bosonic Gaussian channels are circular Gaussian distributions of coherent states in phase space \cite{Giovannetti2014,giovannetti2004classical,holevo_2019quantum}, which are are clearly not peak-constrained.
% Note that this is a peak-energy constraint on individual input system as opposed on an energy constraint on the averaged state over ensemble of states (which is the focus of \cite{Giovannetti2014}).

In addition, we also note that by adopting the thermal channel model we obtain a finite-temperature generalization of existing results that assumed a zero-temperature (i.e. vacuum) environment. 
These include past studies of tasks such as the quantum reading of optical memory \cite{Wilde2012,Guha2013} (illustrated in Fig.~\ref{fig:QCD}) and communication over the pure-loss channel \cite{giovannetti2004classical}.

\begin{figure}[t]
    \centering
    % \begin{quantikz}
    %     \lstick{$\ket{\alpha}$} &\qw& \gate{\gamma_T,\eta_1,\theta_1} \gategroup[wires=6,steps=1,style={inner sep=4pt, fill=blue!20}, background]{Memory cells} &\qw& \gate[6][2cm]{\text{Receiver}}\\
    %     \lstick{$\ket{\alpha}$} &\qw& \gate{\gamma_T,\eta_2,\theta_2} &\qw& \\
    %     \lstick{$\ket{\alpha}$} &\qw& \gate{\gamma_T,\eta_3,\theta_3} &\qw& \\
    %     \lstick{$\ket{\alpha}$} &\qw& \gate{\gamma_T,\eta_4,\theta_4} &\qw& \\
    %     \lstick{$\ket{\alpha}$} &\qw& \gate{\gamma_T,\eta_5,\theta_5} &\qw& \\
    %     \lstick{$\ket{\alpha}$} &\qw& \gate{\gamma_T,\eta_6,\theta_6} &\qw& 
    % \end{quantikz}
    \includegraphics[width=\columnwidth]{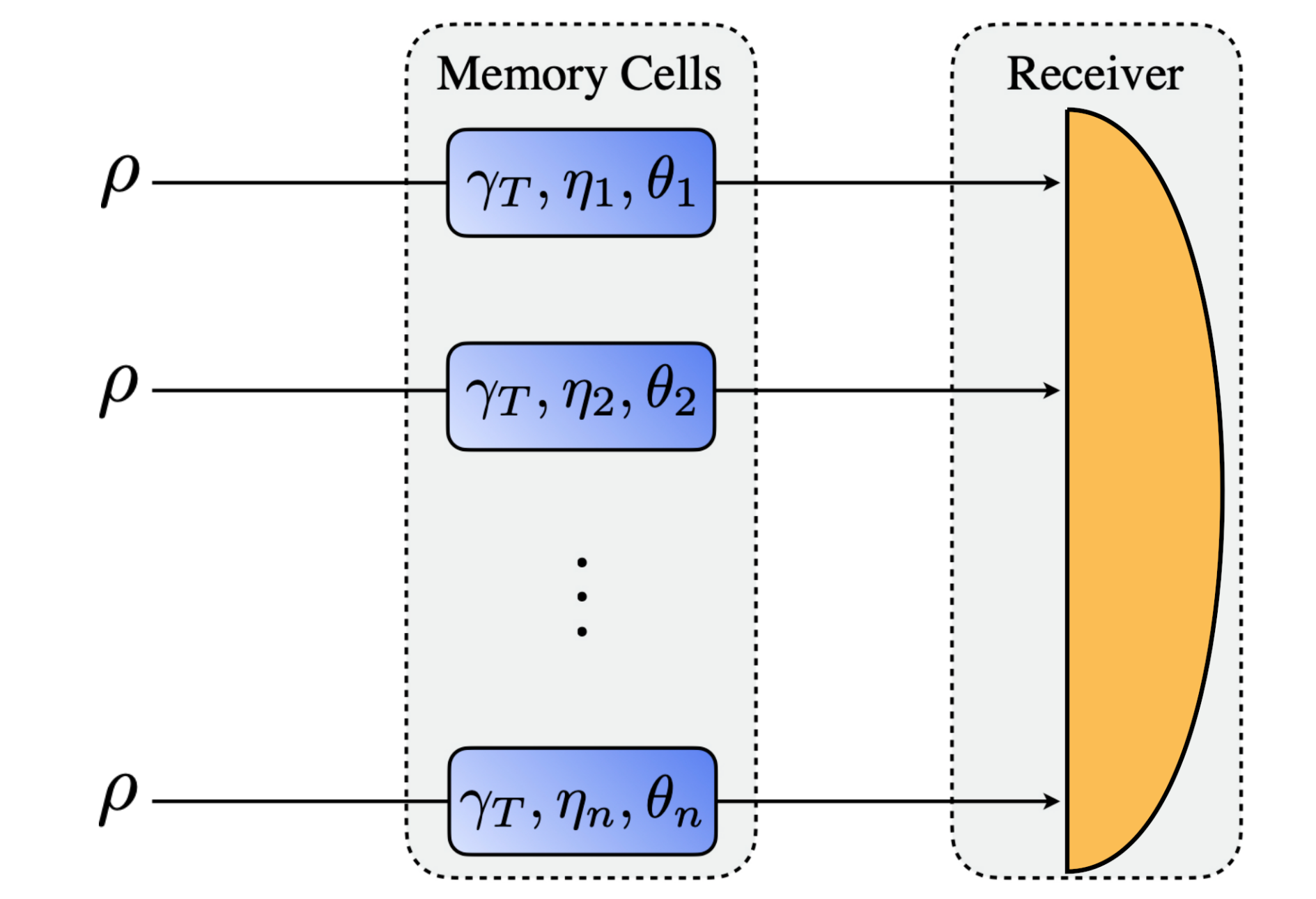}
    \caption{\textbf{Thermal channel memory cell reading.} Each memory cell is a thermal channel with temperature $T$ parameterized by attenuation $\eta_j$ and phase $\theta_j$, transforming an input resource state $\rho$ to an output state to be measured by the receiver.}
    \label{fig:QCD}
\end{figure}

In this work, we investigate the information capacity of such a thermal encoding protocol as quantified by its Holevo information.
% optimal encoding scheme, namely the probability distribution over the attenuation and phase of a thermal channel that achieves the Holevo information for a given arbitrary input state and an arbitrary environment temperature.
First we show that given an arbitrary input state and an arbitrary environment temperature, the encoding that maximizes the Holevo information is characterized by an independent distribution of the attenuation coefficients and the phase shift values, with the latter being uniformly distributed -- we call such encodings circularly symmetric encodings.
We then derive necessary and sufficient information-theoretic conditions that characterize the optimal distribution of the attenuation coefficient for such circularly symmetric encodings.
For the case of a coherent-state resource at zero temperature and a thermal state at any temperature, we show analytically that the optimal encoding involves only a finite number of attenuation values.
% Moreover, we believe that any optimal circularly symmetric encoding scheme contains only a finite number of channel attenuation points.
% We support this by showing that the cardinality of the support of distributions over channel attenuation is finite for any optimal circularly symmetric encoding in the case of coherent state input at zero-temperature and in the case of thermal state input.
In addition, numerical results based on the aforementioned information-theoretic optimality conditions also show that this property holds in various cases where the output states are displaced thermal states.
Based on this combination of analytical and numerical evidence, we conjecture that circular symmetric encodings with a finite number of attenuation values maximize the Holevo information for any resource state.
% Moreover in cases where the resource state is a coherent state or a thermal state, we show that the cardinality of the support of the distribution over channel attenuation is finite for any optimal encoding, thus forming finite number of rings around the origin in the phase space.
% We then proceed by deriving properties of the optimal encoding in the case of different families of Gaussian input states: coherent states, displaced phase-squeezed states, displaced thermal states, and the vacuum state.
% Assuming a finite energy of the input states, we analyze how its capacity changes as the input energy changes as well as how it changes with environment temperature.
% In relation to previous works on continuous-variable channels, we also explore behaviors of the thermal channel in the contexts of the lossy channel scenario \cite{giovannetti2004classical, Guha2013} and for different Gaussian input states in comparison to the capacity of an optimal (Gaussian or non-Gaussian) input state that maximizes the channel capacity \cite{Wilde2012,Giovannetti2014}.

This paper is organized as follows:
In section~\ref{sec_framework}, we lay down the formal definition of a thermal channel and then present the general properties of an optimal encoding.
In Section~\ref{sec_0temp_environment}, we focus on the case of a zero-temperature environment and obtain optimal encodings when the resource is a coherent state or a displaced thermal state.
In Section~\ref{sec_mixed_state}, we discuss a nonzero-temperature environment with a coherent-state or thermal-state resource.
We close with Section~\ref{sec:discussion}, which briefly discusses the implications of our results.
The technical proofs are relegated to the appendices.

\mg{Extensive coverage of the classical case is likely unnecessary, and certain too diverts readers from key ideas in the introduction. Probably doesn't make much sense till after framework}
\mg{Too much exposition on possibly related work before we get to main result. Recommending moving most of this to discussion where it'd make more sense. Last paragraph of introduction should be describing what we did.}\at{OK. moved all of these to relevant sections.}

% Hence one natural question that arises from this is: What is the maximum amount of information can Alice encode into these states?
% More particularly: Given these "resource" quantum states, which distribution over the choice of parameters of her encoding operation should Alice use to achieve the Holevo's information?

% A particular class of operations that does not contain source of energy, the thermal operations, has been studied in the context of thermodynamics \cite{Narasimhachar2019}, algorithmic cooling \cite{serafini2020gaussian}, and more relevantly on its information capacity in encoding information in a two-level (qubit) system~\cite{Narasimhachar2019_2}.

% her information by performing certain choice of thermal operations on her resource state, where a thermal operation is defined as a unitary interaction between a "resource" state, acting as input, and a "free" Gaussian thermal state, which does not change their total energy.

% In this case, the thermal operations consists of a phase shift on the resource state, and a mixing between the resource state and a thermal state.
% We will give a more rigorous definition of these thermal operations later.

\section{Thermal encodings: definitions and optimization}\label{sec_framework}

\subsection{Thermal operations}\label{sec_thermal_operation}

\mg{Need to be careful. We are talking only about passive thermal operations, and are not all thermal operations passive operations? Thermal operations are also generally defined different e.g. energy conserving interaction with a Gibbs state at temperature T. What we mean here is Gaussian Thermal Operations? Need to briefly describe and motivate this class.}
\at{OK, I've made up the name "thermal channel".}
% In this subsection, we will lay down definitions related to thermal encoding before proceeding to the discussion about its properties.
In this section, we first define \emph{thermal channels} whose output states constitute the potential codewords for the channel.
Consider a given input \emph{resource state} with energy $E<\infty$ represented by density operator $\rho$ in Hilbert space $\H_A$ and thermal state with temperature $T$ as the \emph{environment state} represented by density operator $\gamma_T = \frac{e^{-\beta H}}{\mathrm{Tr}(e^{-\beta H})}$ in Hilbert space $\H_E$ for a Hamiltonian $H = \hat{a}_E^\dag\hat{a}_E$ (where $\hat{a}_E^\dag$ and $\hat{a}_E$ are the creation and annihilation operators on $\H_E$, respectively) and inverse temperature $\beta=\frac{1}{kT}$.
Where it is convenient to express a thermal state in terms of its mean photon number $N=\<\hat{a}_E^\dag\hat{a}_E\>$, we write $\rho_\th(N) = \sum_{n=0}^\infty \frac{N^n}{(N+1)^{n+1}} \ketbra{n}{n}$.
We consider only phase-rotation operations and mixing operations with a thermal state.
These operations do not require an external energy source and have been studied in~\cite{Narasimhachar2019,serafini2020gaussian}.
The thermal channel consists of exactly each of these two operations on the resource state: an operation mixing $\rho$ and $\gamma_T$ followed by a noiseless phase-rotation operation.

\begin{definition}[Thermal channel]
    A \textit{thermal channel} with attenuation $\eta\in[0,1]$, phase rotation $\theta\in[-\pi,\pi)$, and environment temperature $T\geq0$ is a quantum channel $\mathcal{T}_{\eta,\theta}^{(T)}$ whose action on a given input state $\rho\in\H_A$ is defined as
    \begin{align*}
        \mathcal{T}_{\eta,\theta}^{(T)}(\rho) = \mathrm{Tr}_E\left(  R_\theta\; B_\eta (\rho\otimes\gamma_T) B_\eta^\dag\; R_\theta^\dag \right)
    \end{align*}
    where $B_\eta = e^{\varphi_\eta(\hat{a}_A^\dag\hat{a}_E - \hat{a}_E^\dag\hat{a}_A)}$ is a beamsplitter operation with transmissivity $\eta = \cos^2\varphi_\eta$ and $R_\theta = e^{-i\theta \hat{a}_A^\dag\hat{a}_A}$ is a phase-rotation of $\theta$ on $\H_A$.
    For a circuit representation, see Fig.~\ref{fig:thermal_channel}.
\end{definition}

Given thermal channel with temperature $T$, attenuation $\eta$, and phase $\theta$, its output is called a codeword and is denoted by $\rho_T(\eta,\theta) = \mathcal{T}_{\eta,\theta}^{(T)}(\rho)$ or simply $\rho(\eta,\theta)$ when temperature $T$ is clear from the context. 
A \emph{thermal encoding} is defined by a cumulative distribution function $F(\eta,\theta)$ over attenuation $\eta$ and phase $\theta$ of the thermal channel which generates the ensemble of codewords $\{\rho(\eta,\theta), dF(\eta,\theta)\}_{\eta,\theta}$.
For a thermal encoding $F$ with environment temperature $T$ acting on resource state $\rho$, the amount of classical information that can be encoded onto the output quantum state is given by the Holevo information~\cite{Holevo1973,Holevo1998,Schumacher1997}:
\begin{align}
  \chi_{te}[F] &= S(\rho_\text{ave}) - \int dF(\eta,\theta) S(\rho_T(\eta,\theta))\;,
\end{align}
where $S$ is the von Neumann entropy and $\rho_\text{ave} = \int dF(\eta,\theta) \rho_T(\eta,\theta)$ is the averaged state of the ensemble.
All information quantities throughout this paper are in bits, hence all $\log$ functions are base $2$.
Given a resource state $\rho$ and temperature $T$, our task is to find the thermal encoding $F^*$ that maximises the Holevo information $\chi_{te}[F]$.
We note that $\chi_{te}$ is an adaptation of the thermal information capacity proposed in~\cite{Narasimhachar2019_2} which quantifies the optimal Holevo information over distribution of passive thermal operations.

\subsection{Characterizations of optimal thermal encodings}

\mg{Is $\phi$ a state or a density matrix? Notation requires clarification}\at{OK, made this clear above}

In this section, we state the properties of thermal encodings $F$ that maximize information capacity of the thermal channel $\chi_{te}[F]$ given resource state $\rho$ and environment temperature $T$.
Interestingly, it turns out that these properties share many features with the optimal encodings for classical additive white Gaussian noise (AWGN) channels where the signals satisfy a \emph{peak} energy constraint $E$ (see discussions in Section~\ref{sec:discussion}).
Before we proceed to state these properties, we define circularly symmetric encodings which play an important role in the optimality conditions.

\begin{definition}\label{def:rings}
    A thermal encoding $F'$ such that 
    % $dF'(\eta,\theta) = dF(\eta) \frac{d\theta}{2\pi}$ 
    $F'(\eta,\theta) = F(\eta) \frac{\theta+\pi}{2\pi}$ is a \textit{circularly symmetric encoding}.
    Given a resource state $\rho$, for each $\eta$, we call $\rho(\eta) = \int \frac{d\theta}{2\pi} \rho(\eta,\theta)$ a \textit{ring state}.
\end{definition}

A ring state $\rho(\eta)$ is a mixture of codewords $\{\rho(\eta,\theta)\}_\theta$ in which the phase $\theta$ is distributed uniformly around the origin of the phase space, thus forming a ring.
Hence we can write the average output state as a mixture of ring states
\begin{align*}
    \rho_\text{ave}= \int dF(\eta) \rho(\eta) \;.
\end{align*}
Both $\rho(\eta)$ and $\rho_\mathrm{ave}$ are diagonal in the Fock basis with their $n$-th diagonal entry being $P_n[\rho(\eta)] = \left\langle n |\rho(\eta)|n\right\rangle$ and
\begin{align}
P_n[F] &= \left\langle n |\rho_\text{ave}|n\right\rangle = \int dF(\eta) \,P_n[\rho(\eta)] \;, \label{eq:19}
\end{align}
respectively.
For a proof, see Lemma \ref{lem:avg_phase_diagonal_fock_basis_state}.
% As a consequence, any output state with circularly symmetric $F$ has mutually independent $\eta$ and $\theta$.
This family of encodings is central to the first optimality condition.

\begin{prop}\label{prop_1}
    Given a resource state $\rho$ and channel temperature $T$, for any encoding $F$, there exists a circularly symmetric encoding $F'$ such that $\chi_{te}[F'] \geq \chi_{te}[F]$.
\end{prop}
% \begin{proof}
%     See Appendix~\ref{prop_1_proof}
% \end{proof}

\mg{Above proposition lacks details. Are we talking about phase space? What is a circularly symmetric ring? While we can leave proofs do the appendix, we should state Propositions formally. Also if this result is already proven by Guha then I would recommend not putting it as its own proposition, that makes it look like we are claiming this as an original result. If it not, the the phrasing below makes it sound like it is not original} \at{OK, I think this should be clearer now.}
This implies that for any optimal encoding $F^*= \arg\max_F \chi_{te}[F]$ there exists a circularly symmetric encoding $F'$ that attains that same optimal capacity (i.e. $\chi_{te}[F']=\chi_{te}[F^*]$).
This property is mentioned and shown in a plot by Guha et al.~\cite{Guha2011} comparing capacity of different quantum reading settings, but was missing a formal proof.
We complete this missing piece with a proof in Appendix~\ref{prop_1_proof}.
Thanks to Proposition 1, we can focus our search for optimal encodings to circularly symmetric encodings, which are completely characterized by their cumulative distribution function (cdf) over the attenuation coefficent $\eta$ alone.
To avoid notational complexity, we will hereafter use $F(\eta)$ to also denote the cdf of $\eta$ alone when dealing with circulary symmetric encodings.

% Because our codewords form rings, the optimal encodings are independent of phase $\theta$ and can be identified solely by its transmissivity $\eta$. 
\mg{Really need more details here. Should have an independent definition that describes how our encodings are parameterized.} \at{OK, I got rid of this confusing sentence.}
% Given attenuation $\eta$ from a circularly symmetric $F^*$, we get a ring state $\rho(\eta)=\int \frac{d\theta}{2\pi} \rho(\eta,\theta)$.
% An encoding scheme can therefore be parameterized by a (cumulative) probability distribution function $F(\eta)$ where the sum of the average states $\rho_\text{ave}= \int dF(\eta) \rho(\eta) $ is also diagonal in the Fock basis with entries

\mg{Wording again makes it sound like we are restating a known result. [3] is purely classical, which should be clarified. Better presentation is state our definition first, and then mention that for the case where it is `classical', the definition agrees with that of Smith?}\at{might need some help on how to think about this.}
We now discuss the second property of the optimal thermal encoding, a necessary and sufficient condition for an optimal encoding $F^*$. 
We first note that the Holevo information of a circularly symmetric encoding $F$ can be written as an average of a function that depends on $\eta$ that indicates information content of the ring at $\eta$ by expressing $\chi_{te}$ in terms of relative entropy $S(\cdot||\cdot)$, i.e.
\begin{align}
\begin{split}
    &\chi_{te}[F] = \int dF(\eta,\theta) S(\rho(\eta,\theta) || \rho_\mathrm{ave}) \\
    &= \int dF(\eta,\theta) \Big( - \tr{\rho(\eta,\theta) \log \rho_\mathrm{ave}} + \\&\quad \tr{\rho(\eta,\theta) \log\rho(\eta,\theta)} \Big) \\
    &= \int dF(\eta) \Big( - \tr{\rho(\eta) \log \rho_\mathrm{ave}} -  S(\rho(\eta,0)) \Big) \;,
    % &= \int dF(\eta) -\sum_n P_n[\rho(\eta)] \log P_n[F] - S(\rho(\eta,0)) \\
    % &= \int dF(\eta) i[\eta,F] \;, \label{eq:18}
\end{split}
\end{align}
where we have used Fubini's theorem (see, e.g. Chapter 2.3 of \cite{stein_shakarchi_2009real}) to swap the sum and integral in $S(\rho_\text{ave})$ and used the fact that entropy $S(\rho(\eta,\theta))$ is independent of the phase $\theta$, allowing us to put $\rho(\eta,0)$ in place of $\rho(\eta,\theta)$.
Now, define the \textit{marginal information density} of circularly symmetric encoding $F$ at $\eta$ as 
\begin{gather}
    i[\eta,F] = -\sum_n P_n[\rho(\eta)]\log P_n[F] - S(\rho(\eta,0)) \label{eq:17} \\
    \textup{so,}\quad \chi_{te}[F] = \int dF(\eta) i[\eta,F] \;. \label{eq:18_main}
\end{gather}
Since the relative entropy $S(\rho(\eta,\theta) || \rho_\mathrm{ave})$ is always non-negative, the marginal information density $i[\eta,F]$ is also non-negative, indicating that it is an information measure of ring $\rho(\eta)$ given circularly symmetric encoding $F$.
% (one can show that $i[\eta,F]$ is non-negative by noting that the entropy of output state with an arbitrary phase $\theta$, $S(\rho(\eta,\theta))$ is at most the entropy of the phase-averaged output state $S(\rho(\eta))$ and then get the relative entropy of $\rho(\eta)$ from $\rho_\mathrm{ave}$).
% where the first term is the cross entropy between $P_n[\rho(\eta)]$ and $P_n[F]$.
% In congruence with this interpretation, the Holevo information for encoding $F$ can be written as the average of the marginal information density,
We note that the marginal information density also plays an integral part in the results for the capacity of the classical amplitude-constrained additive white-Gaussian-noise (AWGN) channel by Smith~\cite{Smith1971} and the direct-detection photon channel by Shamai~\cite{Shamai1990}, particularly on the properties of the optimal input alphabet distribution.
The set of attenuation-coefficients $\eta$ where $F$ is increasing is of particular importance for this second property.
This set has been called the \emph{points of increase (POI)} (e.g.~\cite{smith_1969information,Smith1971,Shamai1995}) and \emph{support} (e.g.~\cite{Dytso2019,Yagli2019}) of the encoding in the literature, depending on whether the encoding is defined as a cumulative distribution function or a probability distribution over the attenuation-coefficients.
% However, they are the same object as we define below.
\begin{definition}\label{def:POI}
    Given a probability distribution on $[0,1]$ and its cdf $G(\eta):= \mathrm{Pr}[x \leq \eta]$. 
    The \emph{support} of the distribution is defined as the smallest closed set $\mathcal{S} \subseteq [0,1]$ such that $\mathrm{Pr}[\mathcal{S}]=1$. 
    A point $\eta \in \mathcal{S}$ is also called a \emph{point of increase} of the cdf $G$. 
    We will call $\mathcal{I}_G$ the set of points of increase (equivalently, the support set of the probability distribution) the \emph{POI} of $G$.

    % For a circularly symmetric encoding $F$, the set $\mathcal{I}_F = \{\eta\in[0,1] : \forall s<\eta,\; F(\eta)>F(s)\}$ is the \textit{points-of-increase} (POI) of $F$.
    % For a circularly symmetric encoding $F$, the set $\mathcal{I}_F = \{\eta\in[0,1] : dF(\eta)>0\}$ is the \textit{points-of-increase} (POI) of $F$.
\end{definition}
Now we are ready to state our second property.

\begin{prop}\label{prop_2}
    For any given resource state $\rho$, a circularly symmetric encoding $F^*(\eta)$ is uniquely optimal if and only if
    \begin{subequations}
        \begin{equation}\label{eq:9_main}
            i[\eta,F^*] \leq \chi_{te}[F^*] \quad\textup{for all $\eta\in[0,1]$}
        \end{equation}
        \begin{equation}\label{eq:9.2_main}
            i[\eta,F^*] = \chi_{te}[F^*] \quad\textup{if and only if $\eta\in \mathcal{I}_{F^*}$}
        \end{equation}
    \end{subequations}
    for POI $\mathcal{I}_{F^*}$.
    %   \begin{enumerate} \item 
    %     $i[\eta,F^*] \leq \chi[F^*]$ for $0\leq \eta \leq 1$
    %      \item $i[\eta,F^*] = \chi[F^*]$ for $\eta$ that are `points of increase' of $F^*$, where $\eta$ is a point of increase of $F^*$ if and only if $\eta$ where $F^*(\eta)$ is increasing (i.e., $dF^*(\eta)>0$).
    %   \end{enumerate}
\end{prop}

\mg{The phrasing of the second statement is very confusing, can you rephrase?}\at{OK, should be more precise now.}

The complete proof is provided in Appendix \ref{prop_2_proof}.
This property establishes a necessary and sufficient condition for a circularly symmetric encoding $F^*$ that achieves the optimal Holevo information $\sup_F \chi_{te}[F]$.
In particular, the first statement can be understood as the Holevo information $\chi_{te}[F^*]$ being an upper bound to the marginal information $i[\eta,F^*]$ regardless the value of $\eta$.
% Hence it is impossible to find an $\eta$ for which we can have more information capacity for the channel than the Holevo quantity of symmetric encoding $F^*$, for this will give us some encoding $F$ such that $\chi[F]>\chi[F^*]$ which contradicts Proposition \ref{prop_1}.
Whereas, the second statement says that for any point $\eta$ for which $F^*(\eta)$ is increasing (either in a continuous or discontinuous manner), its marginal information $i[\eta,F^*]$ is equal to the Holevo information.
Because an encoding is a cumulative distribution function it may be more intuitive for one to interpret this as a characterization of the support of probability distribution $dF^*$ induced by optimal encoding $F^*$.

\begin{figure*}[!ht]
\centering
    \begin{subfigure}[b]{0.33\textwidth}
    \subcaption[short for lof]{$|\alpha_\mathrm{max}|^2 \sim 1.1$}
    \includegraphics[width=0.95\linewidth]{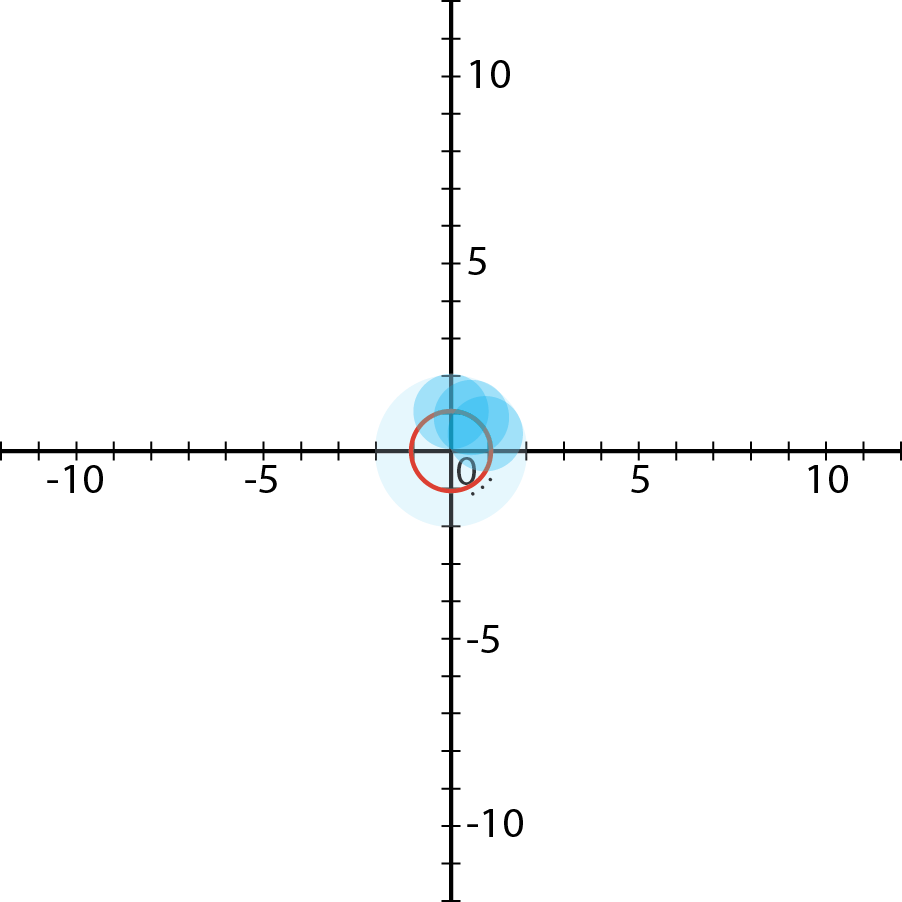}
    \end{subfigure}%
    \begin{subfigure}[b]{0.33\textwidth}
    \subcaption[short for lof]{$|\alpha_\mathrm{max}|^2 \sim 3.5$}
    \includegraphics[width=0.95\linewidth]{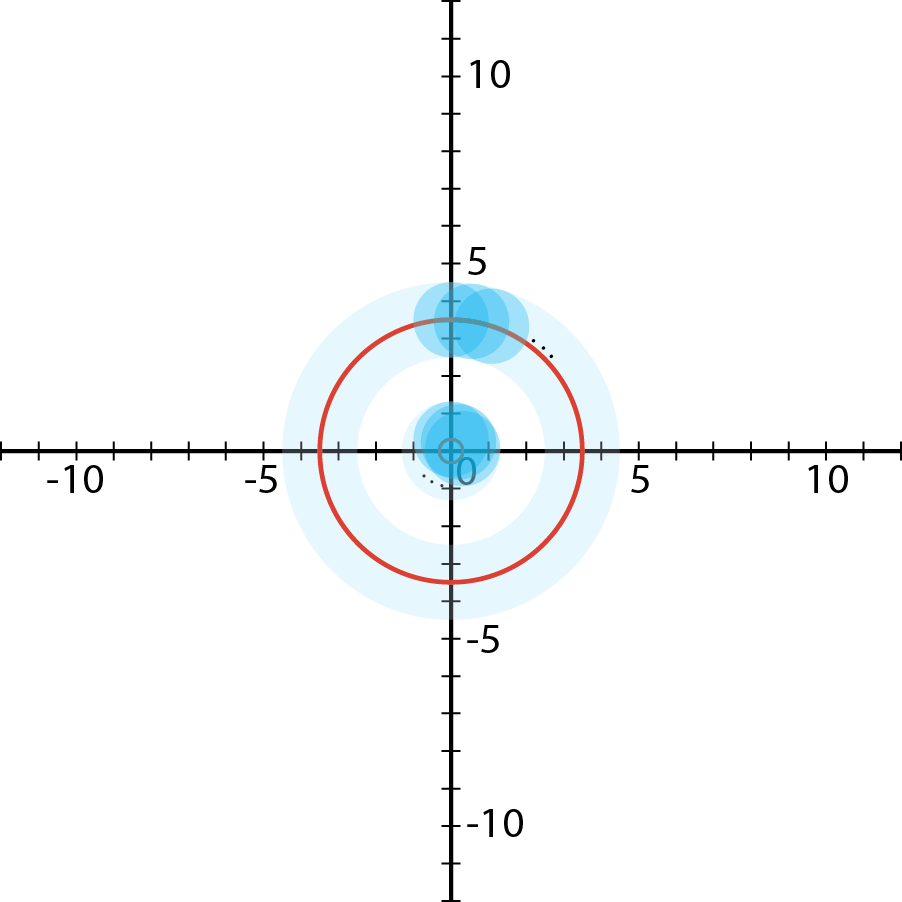}
    \end{subfigure}%
    \begin{subfigure}[b]{0.33\textwidth}
    \subcaption[short for lof]{$|\alpha_\mathrm{max}|^2 \sim 9.2$}
    \includegraphics[width=0.95\linewidth]{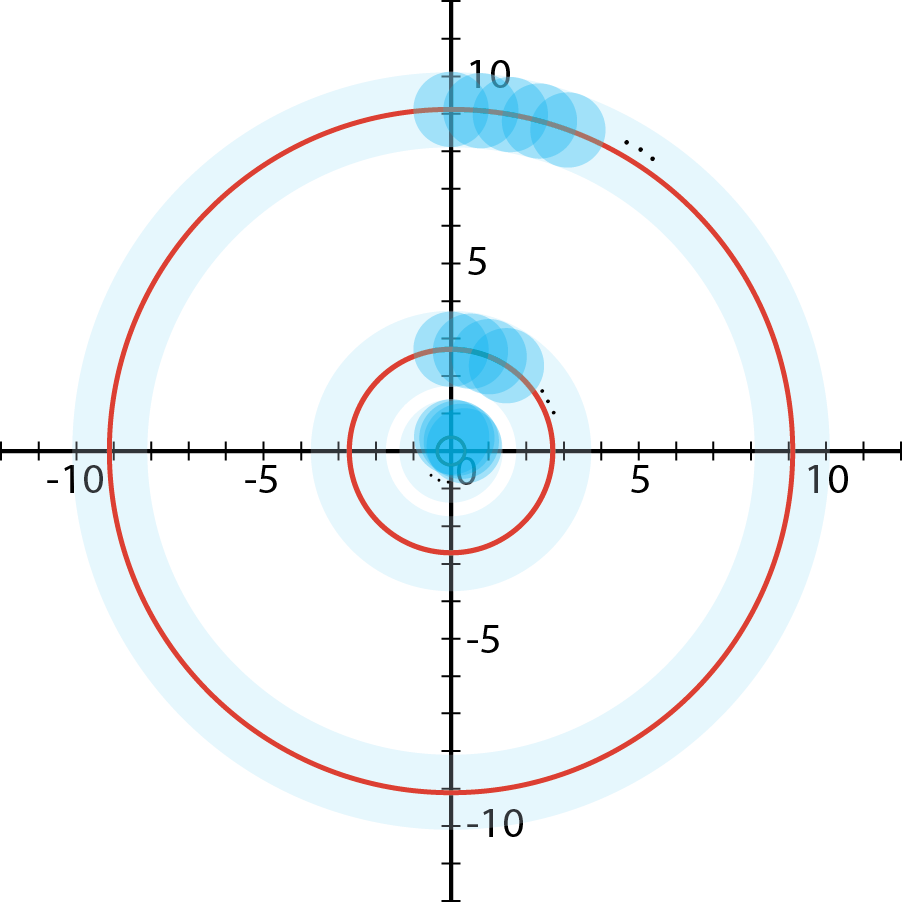}
    \end{subfigure}%
\caption{
    \textbf{Phase-space visualization of optimal circularly symmetric encoding with coherent state resource at different energies in zero-temperature environment.}
    These are phase-space representations of some instances of optimal encodings plotted in Fig.~\ref{fig:1}(a),(b).
    Radius of each red circle indicate energy of its corresponding ring state determined by points-of-increase while, the blue circles indicate attenuated and phase-shifted coherent state codewords that are in each ring state mixture.
    The coherent state codewords are centered on the circumference of each red circle, with the light blue rings providing a depiction of the ring state corresponding to each point of increase.
    \textbf{(a)} Optimal encoding for input energy $|\alpha_\mathrm{max}|^2 \sim 1.1$ with capacity $\chi_{te} \sim 2$ and one ring with energy $\sim 1.1$.
    \textbf{(b)} Optimal encoding for input energy $|\alpha_\mathrm{max}|^2 \sim 3.5$  with capacity $\chi_{te} \sim 3$ and two rings with energies $\sim0.24$ and $\sim3.5$.
    \textbf{(c)} Optimal encoding for input energy $|\alpha_\mathrm{max}|^2 \sim 9.2$ with capacity $\chi_{te} \sim 4$ and three rings with energies $\sim0.3$, $\sim2.7$, and $\sim9.2$.
}
\label{fig:phasespace_encoding_coherent_zerotemp}
\end{figure*}

Supported by analytical and numerical results we present in later sections for some large classes of Gaussian resource states, we conjecture that the Holevo-optimal distribution $F^*$ over attenuation $\eta$ for \emph{any} resource state $\rho$ has a finite support by showing that its corresponding POI $\mathcal{I}_{F^*}$ must have a finite cardinality.
\begin{conjecture}\label{conj:finite_POI_general}
    For any resource state $\rho$ and temperature $T \geq 0$, the set $\mathcal{I}_{F^*}$ of points of increase of the attenuation-coefficient distribution $F^*(\eta)$ associated with the optimal circularly-symmetric encoding has finite cardinality. 
\end{conjecture}
Conjecture~\ref{conj:finite_POI_general} states that given an optimal circularly symmetric encoding, the average output state of the thermal channel is a finite mixture of ring states, which can be nicely visualized as discrete rings in its phase-space representation as presented in Fig.~\ref{fig:phasespace_encoding_coherent_zerotemp}.
We show this property rigorously for coherent state resource in zero-temperature and thermal state resource in any temperature as stated in Proposition~\ref{prop:finite_POI_coht0} and Proposition~\ref{prop:finite_POI_thermal}, respectively, and discussed in the corresponding sections they are in.
Numerical evidence for more general cases where the channel output is a displaced thermal state are also given in the following sections.

% \begin{figure*}[!ht]
% \centering
%     \begin{subfigure}[b]{0.33\textwidth}
%     \subcaption[short for lof]{$|\alpha_\mathrm{max}|^2 \sim 1.1$}
%     \includegraphics[width=0.95\linewidth]{./rings/encoding1.png}
%     \end{subfigure}%
%     \begin{subfigure}[b]{0.33\textwidth}
%     \subcaption[short for lof]{$|\alpha_\mathrm{max}|^2 \sim 3.5$}
%     \includegraphics[width=0.95\linewidth]{./rings/encoding2.png}
%     \end{subfigure}%
%     \begin{subfigure}[b]{0.33\textwidth}
%     \subcaption[short for lof]{$|\alpha_\mathrm{max}|^2 \sim 9.2$}
%     \includegraphics[width=0.95\linewidth]{./rings/encoding3.png}
%     \end{subfigure}%
% \caption{
%     \textbf{Some discrete approximations of optimal circularly symmetric encoding with coherent state resource with zero-temperature environment in Fig.~\ref{fig:phasespace_encoding_coherent_zerotemp}.}
% }
% \label{fig:phasespace_encoding_coherent_zerotemp_discrete}
% \end{figure*}

\mg{The phrasing of this result again can make referees question our original contribution. Simply state the result, and forego methodology till appendix. Separately it is clear why this point would be interesting as the manuscript is about thermal capacity, not `points of increase' Perhaps remove altogether given we focus on this next section?} \at{OK, moved this discussion to the section 3.1, 2nd para.}

% \begin{prop}
%   The number of rings in the optimal encoding is discrete.
% \end{prop}
% Again, we can probably modify Smith's proof to show this. With these three propositions, it is easy to write a computer
% programme to search for the optimal encoding. If proposition 2 is
% true, we also have a way to certify that the numerical solution is
% indeed optimal.

\section{Zero temperature environment}\label{sec_0temp_environment}

In this section we discuss finding the optimal encoding when the environment has temperature $T=0$, which corresponds to the environment state being the vacuum state $\gamma_0 = \ketbra{0}{0}$.
Given a resource state $\rho$ with energy $E_\mathrm{max}$, the codeword $\rho(\eta,\theta)$ has average energy $\eta E_\mathrm{max}$ as a consequence of $\gamma_0$ environment having a $0$ energy.
Therefore, the average state $\rho_\mathrm{ave}$ has energy at most $E_\mathrm{max}$.
As such, the Holevo information of any encoding and any resource state is bounded above by the quantity \cite{Giovannetti2014,holevo_2019quantum}
\begin{align}\label{eqn:upper_bound_holevo_t0}
\begin{split}
    g(E_\mathrm{max}) &:= (1+E_\mathrm{max})\log(1+E_\mathrm{max}) \\ &\quad- E_\mathrm{max} \log E_\mathrm{max} \;.
\end{split}
\end{align}
We use this universal capacity upper-bound for $T=0$ environment to benchmark the encoding capacity of our optimal schemes -- see Appendix~\ref{app:vacBounds} for more details.

\subsection{Coherent state resource}\label{sec_0temp_environment_coherent_resource}

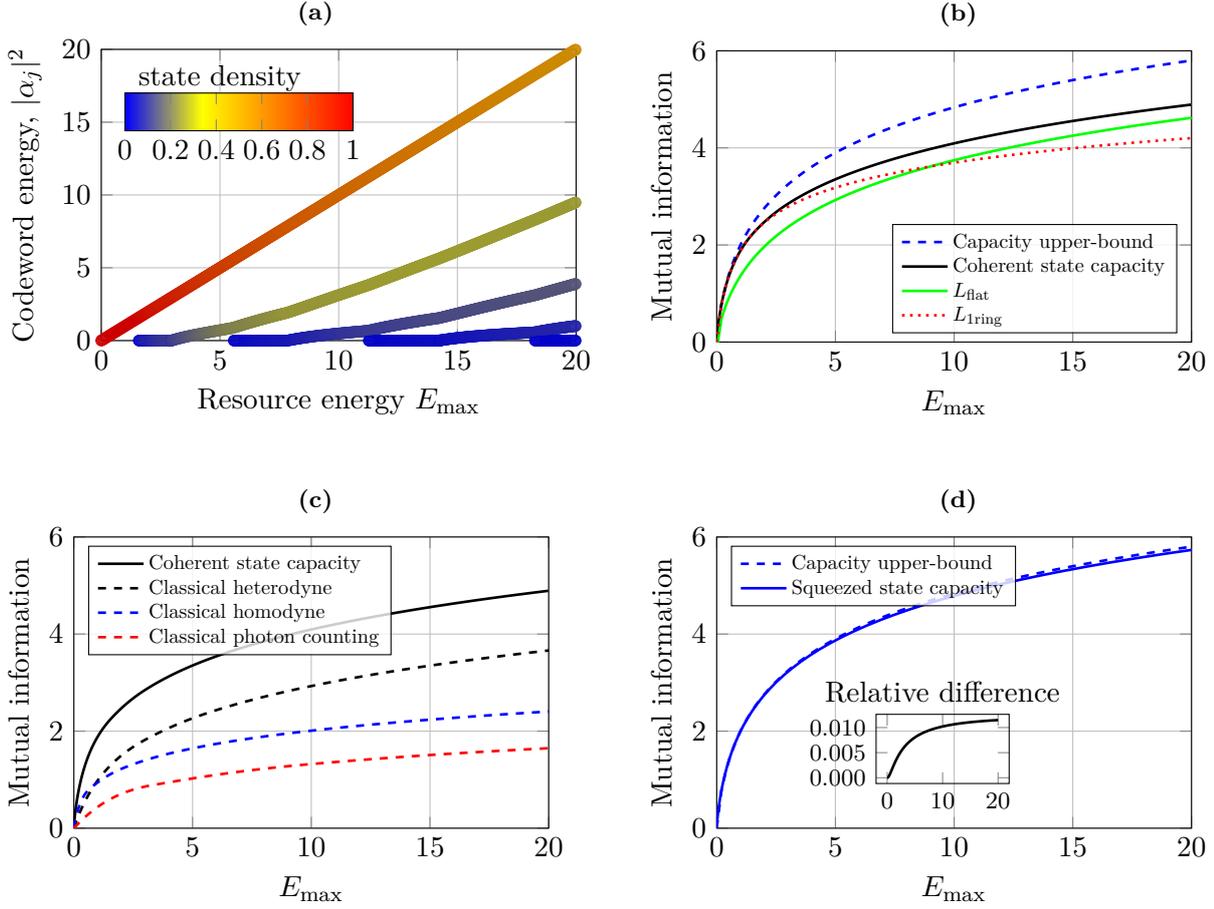
\begin{figure*}[!ht]
\centering
    \begin{subfigure}{0.49\textwidth}
    \subcaption[short for lof]{}
    \begin{tikzpicture}
%\def\epstwo{0}
%\addplot[color=blue,line width=\xlinewidth,samples=100][domain=0:1]
%{2*(1-x)*(1-\epstwo)/(1-\epstwo+1-x)};
%\addlegendentry{$\epsilon_2=0$};

\begin{axis}[myaxis,
    xmin=0, xmax=20,
    ymin=0, ymax=20,
    grid=both,
    colorbar horizontal,
    colorbar style={
      at={(.05,0.85)},
      anchor=north west,
      width=3cm
    },
    xlabel={Resource energy $E_\text{max}$} , 
    ylabel={Codeword energy, $|\alpha_j|^2$}
    ]

    \node[] at (axis cs:5,18) {state density};

    \addplot[scatter,only marks,scatter src=explicit, mark=o ]
    table  [col sep=comma,x index=0,y index=1,meta index=2] {scatterV0.csv};

    % \def\discontA{1.257646^2}
    % \draw [dashed,thick] (\discontA,0) -- (\discontA,20);

    % \def\discontB{1.833970^2}
    % \draw [dashed,thick] (\discontB,0) -- (\discontB,20);

    % \def\discontC{2.359160^2}
    % \draw [dashed,thick] (\discontC,0) -- (\discontC,20);

    % \def\discontD{3.132070^2}
    % \draw [dashed,thick] (\discontD,0) -- (\discontD,20);

    % \def\discontE{3.354295^2}
    % \draw [dashed,thick] (\discontE,0) -- (\discontE,20);

    % \def\discontG{4.171020^2}
    % \draw [dashed,thick] (\discontG,0) -- (\discontG,20);

    % \def\discontF{4.272969^2}
    % \draw [dashed,thick] (\discontF,0) -- (\discontF,20);
    
    % \node[anchor=west] at (0,17) {$1$r};
    % \node[anchor=west] at (\discontA,17) {\begin{tabular}{c}1r\\+\\v\end{tabular}};
    % \node[anchor=west] at (\discontB,17) {2r};
    % \node[anchor=west] at (\discontC,17) {2r+v};
    % \node[anchor=west] at (\discontD,17) {3r};
    % \node[anchor=west] at (\discontE,17) {3r+v};
    % \node[anchor=west] at (\discontG-0.2,17) {4r};
    % \node[anchor=west] at (\discontF,17) {\begin{tabular}{c}4r\\+\\v\end{tabular}};

\end{axis}
\end{tikzpicture}
%%% Local Variables:
%%% mode: latex
%%% TeX-master: "../cv_thermal_encoding.tex"
%%% End:
    \end{subfigure}
    \begin{subfigure}{0.49\textwidth}
    \subcaption[short for lof]{}
    \begin{tikzpicture}

  \begin{axis}[myaxis, legend style={nodes={scale=0.7, transform shape}}, legend pos=south east,xlabel={$E_\mathrm{max}$}, grid=both]

\def\xlinewidth{1.0 pt}

%\def\epstwo{0}
%\addplot[color=blue,line width=\xlinewidth,samples=100][domain=0:1]
%{2*(1-x)*(1-\epstwo)/(1-\epstwo+1-x)};
%\addlegendentry{$\epsilon_2=0$};

\addplot[color=blue,dashed,line width=\xlinewidth,samples=100][domain=0:20]
{-x*log2(x/(1+x))+log2(1+x)};
\addlegendentry{Capacity upper-bound};

%\addplot [allfigs,color=black,mark=none] table [col sep=comma, x
%index=0, y index=1]{dispSqzProbe_unoptimised_dim300.csv};
%\addlegendentry{$DSS$};

\addplot [allfigs,color=black,mark=none] table [col sep=comma]{T0_nRes0.csv};
\addlegendentry{Coherent state capacity};

\addplot [allfigs,color=green,mark=none] table [col sep=comma]{LFlatTop.csv};
\addlegendentry{$L_\text{flat}$};

\addplot [allfigs,color=red,dotted,mark=none] table [col sep=comma]{L1Ring.csv};
\addlegendentry{$L_\text{1ring}$};

% \def\ymax{6}

%     \def\discontA{1.257646^2}
%     \draw [dashed] (\discontA,0) -- (\discontA,\ymax);

%     \def\discontB{1.833970^2}
%     \draw [dashed] (\discontB,0) -- (\discontB,\ymax);

%     \def\discontC{2.359160^2}
%     \draw [dashed] (\discontC,0) -- (\discontC,\ymax);

%     \def\discontD{3.132070^2}
%     \draw [dashed] (\discontD,0) -- (\discontD,\ymax);

%     \def\discontE{3.354295^2}
%     \draw [dashed] (\discontE,0) -- (\discontE,\ymax);

%     \def\discontG{4.171020^2}
%     \draw [dashed] (\discontG,0) -- (\discontG,\ymax);

%     \def\discontF{4.272969^2}
%     \draw [dashed] (\discontF,0) -- (\discontF,\ymax);

\end{axis}
\end{tikzpicture}
%%% Local Variables:
%%% mode: latex
%%% TeX-master: "../cv_thermal_encoding.tex"
%%% End:
    \end{subfigure}\\
    \vspace{1.5em}
    \begin{subfigure}{0.49\textwidth}
    \subcaption[short for lof]{}
    \begin{tikzpicture}

  \begin{axis}[myaxis, legend style={nodes={scale=0.7, transform shape}},xlabel={$E_\mathrm{max}$}, grid=both]

\def\xlinewidth{1.0 pt}

%\def\epstwo{0}
%\addplot[color=blue,line width=\xlinewidth,samples=100][domain=0:1]
%{2*(1-x)*(1-\epstwo)/(1-\epstwo+1-x)};
%\addlegendentry{$\epsilon_2=0$};

%\addplot[color=blue,dashed,line width=\xlinewidth,samples=100][domain=0:20]
%{-x*log2(x/(1+x))+log2(1+x)};
%\addlegendentry{$U_\text{GDE}$};

%\addplot [allfigs,color=black,mark=none] table [col sep=comma, x
%index=0, y index=1]{dispSqzProbe_unoptimised_dim300.csv};
%\addlegendentry{$DSS$};

\addplot [allfigs,color=black,mark=none] table [col sep=comma]{T0_nRes0.csv};
\addlegendentry{Coherent state capacity};

\addplot[allfigs,color=black,dashed,mark=none]table[col sep=comma]{hetCap.csv};
\addlegendentry{Classical heterodyne};

\addplot[allfigs,color=blue,dashed,mark=none]table[col sep=comma]{homCap.csv};
\addlegendentry{Classical homodyne};

\addplot[allfigs,color=red,dashed,mark=none]table[col sep=comma]{poissonCap.csv};
\addlegendentry{Classical photon counting};

%\addplot [allfigs,color=green,mark=none] table [col sep=comma]{LFlatTop.csv};
%\addlegendentry{$L_\text{flat}$};

%\addplot [allfigs,color=red,dotted,mark=none] table [col sep=comma]{L1Ring.csv};
%\addlegendentry{$L_\text{1ring}$};

% \def\ymax{6}

%     \def\discontA{1.257646^2}
%     \draw [dashed] (\discontA,0) -- (\discontA,\ymax);

%     \def\discontB{1.833970^2}
%     \draw [dashed] (\discontB,0) -- (\discontB,\ymax);

%     \def\discontC{2.359160^2}
%     \draw [dashed] (\discontC,0) -- (\discontC,\ymax);

%     \def\discontD{3.132070^2}
%     \draw [dashed] (\discontD,0) -- (\discontD,\ymax);

%     \def\discontE{3.354295^2}
%     \draw [dashed] (\discontE,0) -- (\discontE,\ymax);

%     \def\discontG{4.171020^2}
%     \draw [dashed] (\discontG,0) -- (\discontG,\ymax);

%     \def\discontF{4.272969^2}
%     \draw [dashed] (\discontF,0) -- (\discontF,\ymax);

\end{axis}
\end{tikzpicture}
%%% Local Variables:
%%% mode: latex
%%% TeX-master: "../cv_thermal_encoding.tex"
%%% End:
% replace with classical capacity
    \end{subfigure}
    \begin{subfigure}{0.49\textwidth}
    \subcaption[short for lof]{}
    \begin{tikzpicture}

  \begin{axis}[myaxis,xlabel={$E_\mathrm{max}$}, legend style={nodes={scale=0.7, transform shape}}, grid=both]

\def\xlinewidth{1.0 pt}

%\def\epstwo{0}
%\addplot[color=blue,line width=\xlinewidth,samples=100][domain=0:1]
%{2*(1-x)*(1-\epstwo)/(1-\epstwo+1-x)};
%\addlegendentry{$\epsilon_2=0$};

\addplot[color=blue,dashed,line width=\xlinewidth,samples=100][domain=0:20]
{-x*log2(x/(1+x))+log2(1+x)};
\addlegendentry{Capacity upper-bound};

\addplot [allfigs,color=blue,mark=none] table [col sep=comma, x
index=0, y index=1]{dispSqzProbe_unoptimised_dim300.csv};
\addlegendentry{Squeezed state capacity};

\coordinate (insetPosition) at (axis cs:3.5,.2);
\end{axis}

% no scientific notation for y-axis labels
\pgfplotsset{scaled y ticks=false}

\begin{axis}[width=0.4\columnwidth,height=0.3\columnwidth,
  title={Relative difference},
  title style={yshift={-5pt}},
  at={(insetPosition)},
  anchor={outer south west},
  tick label style={font=\footnotesize},
  y tick label style={/pgf/number format/.cd,
    fixed,fixed zerofill,precision=3,/tikz/.cd}]
  
%  \tikzstyle{allfigs}=[line width=1pt,mark=none]
\addplot [allfigs,color=black,mark=none] table [col sep=comma, x
index=0, y index=6]{dispSqzProbe_unoptimised_dim300.csv};

\end{axis}
\end{tikzpicture}
%%% Local Variables:
%%% mode: latex
%%% TeX-master: "../cv_thermal_encoding.tex"
%%% End:
    \end{subfigure}
\caption{
    \textbf{(a)} For an optimal circularly symmetric encoding $F^*$ given a coherent-state resource $\rho = \ketbra{\sqrt{E_\mathrm{max}}}{\sqrt{E_\mathrm{max}}}$ in a $T=0$ environment, there are a finite number of points-of-increase (POI) $\mathcal{I}_{F^*} = \{\eta_j\}_j$ (see Definition~\ref{def:POI}).
    The optimal encoding $F^*$ at a given $E_\mathrm{max}$ is specified in the plot by the intersection(s) of the vertical line at $E_\mathrm{max}$ with the colored lines.
    % An encoding $F^*$ is described in the plot by the colored points in vertical line.
    The y-axis value of each colored point corresponds to a codeword with energy $|\alpha_j|^2 = \eta_jE_\mathrm{max}$ which correspond to attenuation $\eta_j$ in the POI of encoding $F^*$ with its density $dF^*(\eta_j)$ color encoded.
    For example, given energy $E_\mathrm{max}\sim 3.5$, there are two points of increase $\eta_0>\eta_1$ in an optimal encoding $F^*$.
    They correspond to codewords with energy $|\alpha_0|^2 = \eta_0 E_\mathrm{max} \sim 3.5$ and $|\alpha_1|^2 = \eta_1 E_\mathrm{max} \sim 0.24$ with density $dF^*(\eta_0)\sim 0.87$ and $dF^*(\eta_1)\sim 0.13$, respectively.
    % a given over possible codeword energies $|\alpha_j|^2 = \eta_jE_\mathrm{max}$ for points-of-increase $\mathcal{I}_{F^*} = \{\eta_j\}_j$ (see Definition~\ref{def:POI}) and coherent-state resource $\rho = \ketbra{\sqrt{E_\mathrm{max}}}{\sqrt{E_\mathrm{max}}}$ in a $T=0$ environment.
    % For each $E_\mathrm{max}$, we have finite number of points-of-increase (POI) $\mathcal{I}_{F^*} = \{\eta_j\}_j$.
    % Each POI forms a ring with its shading indicating the probability density of codeword energy $|\alpha_j|^2$.
    \textbf{(b)} Thermal encoding capacity with resource state energy $E_\mathrm{max}$ for the capacity upper-bound \eqref{eqn:upper_bound_holevo_t0} and coherent resource state with: optimal encoding $F^*$, one-ring encoding ($L_\mathrm{1ring}$), and flat encoding ($L_\mathrm{flat}$).
    One-ring encoding is optimal when $E_\mathrm{max}< 1.25034$ and flat encoding is almost optimal when $E_\mathrm{max}$ is large (see Appendix~\ref{sec:lower-bounds_one_ring} and Appendix~\ref{app:flat_dist}).
    \textbf{(c)} The information capacities over resource energy $E_\mathrm{max}$ using homodyne~\cite{Smith1971}, heterodyne~\cite{Shamai1995} and photon counting~\cite{Shamai1990} in comparison with the optimal capacity $\chi_{tc}$.
    \textbf{(d)} The channel capacity for displaced squeezed state resource approximates the capacity upper bound \eqref{eqn:upper_bound_holevo_t0}, which is achieved by the non-Gaussian state~\eqref{eq:optngs} (see Section~\ref{appendix:optimal-state}). 
    The inset plot shows that their relative difference is at most $\sim 1$ percent for resource energy $E_\mathrm{max}\leq 20$.
    \mg{I like the visual representation of (a), however as it is proprietary, we need to devote significant time explaining it - or readers will be confused. The definition of 'optimality' has not shown up here, so use of the term and reference results here will confuse readers.}\at{OK, put more explanations and references to the main text}
}
\label{fig:1}
\end{figure*}

% \begin{figure*}[!ht]
% \centering
%     \begin{subfigure}{0.5\textwidth}
%     % \subcaption[short for lof]{}
%     \input{figures/dispSqzCap.tikz}
%     \end{subfigure}
% \caption{
% The channel capacity for displaced squeezed state resource is nearly as optimal as the optimal non-Gaussian state~\eqref{eq:optngs} (see Section~\ref{appendix:optimal-state}). 
% The inset plot shows that their relative difference is at most $\sim 1$ percent for resource energy $E_\mathrm{max}\leq 20$.
% }
% \label{fig:optimal_state_vs_squeezed}
% \end{figure*}

\mg{We never formally defined what we meant by 'ring', may cause confusion. At minimal a figure is needed} \at{OK, this is defined now in 2.2}
Given a coherent state $\rho = \ketbra{\alpha_\text{max}}{\alpha_\text{max}}$ as resource (hence $E_\mathrm{max}=|\alpha_\mathrm{max}|^2$), the possible codewords $\rho(\eta,\theta)$ are coherent states $\ketbra{\alpha}{\alpha}$ with amplitude $\lvert\alpha\rvert = \sqrt{\eta}|\alpha_\mathrm{max}| \leq \lvert\alpha_\text{max}\rvert$ for $\eta$ distributed by the choice of encoding $F$. 
Since the entropy of every codeword is zero (the codewords are pure states), the Holevo information equals the entropy of the averaged state, namely, $\chi_{tc}[F] = S(\rho_\mathrm{ave})$. 
The problem of finding the optimal encoding is then equivalent to finding the encoding $F$ that maximises the entropy of the mixture distribution over $n\in\N$
\begin{align}\label{eq:13}
    \int dF(\eta) \, \mathcal{P}_n\left(\eta |\alpha_\mathrm{max}|^2\right) \;,
\end{align}
where $\mathcal{P}_n(\lambda) = e^{-\lambda}\frac{\lambda^n}{n!} $ is a Poisson probability distribution over outcome $n$ with mean $\lambda$.

% optimal encoding a coherent state resource with any amplitude $\alpha_\mathrm{max}$ on environment with temperature $T=0$, the support of its optimal encoding over $\eta$ is finite by showing that the set of points of increase if finite.
% \begin{prop}\label{prop3}
%     For an optimal circularly symmetric thermal encoding $F^*$ with coherent state resource and temperature $T=0$, the set of points of increase of $F^*$
%     \begin{align*}
%         \mathcal{I} = \{\eta\in[0,1] : dF^*(\eta)>0\}
%     \end{align*}
%     is finite.
% \end{prop}
% We give a proof in Appendix~\ref{appendix:coherent_t0_finite_POI}.
% This proposition shows that, any optimal encoding with coherent state resource on $0$ environment temperature consists of a finite number of ring states.

The optimal circularly-symmetric encoding finiteness property stated in Conjecture~\ref{conj:finite_POI_general} holds for coherent states at $T=0$, as shown in Appendix~\ref{appendix:coherent_t0_finite_POI}.
\begin{prop}\label{prop:finite_POI_coht0}
    % For any coherent-state resource $\rho = \ket{\alpha}{\alpha}$ and $T=0$, the optimal thermal encoding is circularly symmetric with the attenuation coefficients $\eta$ taking a finite number of values in $[0,1]$. 
    For any coherent-state resource $\rho = \ketbra{\alpha}{\alpha}$ and $T=0$, the attenuation-coefficients $\eta$ of the optimal circularly-symmetric encoding takes a finite number of values in $[0,1]$. 
    % For any optimal circularly symmetric thermal encoding $F^*$ with a coherent state resource in zero temperature, its probability distribution support over attenuation $\eta\in[0,1]$ must have a finite cardinality.
    % the set of points of increase of $F^*$
    % \begin{align*}
    %     \mathcal{I} = \{\eta\in[0,1] : dF^*(\eta)>0\}
    % \end{align*}
    % is finite.
\end{prop}

We plot this finite optimal distribution for resource state energy $|\alpha_\mathrm{max}|^2\leq20$ in Fig.~\ref{fig:1}(a) using numerical optimization\footnote{We used a generic Matlab constrained optimisation routine which worked well enough for our purposes. Algorithms such as in~\cite{Huang2005} may yield more efficient runtime.} to find the optimal encoding scheme and provide ring-shaped phase space representations of these optimal encodings at some fixed values of $|\alpha_\mathrm{max}|^2$ in Fig.~\ref{fig:phasespace_encoding_coherent_zerotemp}.
As the reader may notice in the figure, for any given amplitude $|\alpha_\mathrm{max}|$ of the input coherent state, the support of optimal encoding over attenuation $\eta$ is finite.
Fig.~\ref{fig:1}(a) also shows that the optimal encoding always includes an outermost ring at $|\alpha|=|\alpha_\mathrm{max}|$ and that the probability mass of this outermost ring is non-increasing as the $|\alpha_\mathrm{max}|$ increases. 
For small $|\alpha_\mathrm{max}|$, it is optimal to encode purely in phase, i.e. using states that forms just one ring with Holevo information given by 
\begin{align}\label{eqn:holevo_1ring}
    -\sum_{n=0}^\infty  \frac{e^{-E_\mathrm{max}} E_\mathrm{max}^n}{n!} \log\frac{e^{-E_\mathrm{max}}E_\mathrm{max}^n}{n!} \;,
\end{align}
which is the Shannon entropy of Poisson distributed random variable with values $n\in\N$ with mean $E_\mathrm{max}=|\alpha_\mathrm{max}|^2$ (see Appendix \ref{sec:lower-bounds_one_ring}).  
This is line with the findings of~\cite{Guha2011}, where this ring encoding was shown to be close to the upper bound $g(E_\mathrm{max})$ for $E_\mathrm{max} \ll 1$ (See Fig.~\ref{fig:1}(b)).
\mg{The exact forms of the encodings are a key part of our results. Should be expanded upon and not just relegated to appendix. Things like $\rho_\alpha$ should be mathematically defined in body}\at{OK, theyre here now.}

The threshold at which this single-ring encoding becomes suboptimal can be obtained by solving for $|\alpha_\mathrm{max}|$ when adding a vacuum state to the mixture of coherent states begins to increase the information content of the total state, as we know that the optimal encoding is finite over the channel attenuation.
That is, we want to solve for $\alpha_\mathrm{max}$ that makes
\begin{align}
    \label{eq:7}
    \frac{\partial}{ \partial \epsilon} S\left( \epsilon \rho_0 +
    (1-\epsilon)\rho_\mathrm{max} \right) \Big|_{\epsilon=0} = 0 \;,
\end{align}
where $\rho_\mathrm{max}$ is the ring state at amplitude $\alpha_\mathrm{max}$,
\begin{align}
    \rho_\mathrm{max} = \int_{-\pi}^\pi \frac{d\theta}{2\pi} R_\theta \ketbra{\alpha_\mathrm{max}}{\alpha_\mathrm{max}} R_\theta^\dag \;,
\end{align}
and vacuum state $\rho_0=\ketbra{0}{0}$.
Interestingly, the solution to this is when the entropy of the encoded states is proportional to its energy
\begin{align}
    \label{eq:8}
    S(\rho_\mathrm{max})=\frac{|\alpha_\mathrm{max}|^2}{\ln 2} \;,
\end{align}
where we get $|\alpha_\mathrm{max}|^2=1.25034$ (see Appendix~\ref{sec:alpStar}).
It can be observed further in Fig.~\ref{fig:1}(a) that as $|\alpha_\mathrm{max}|^2$ slightly exceeds $1.25034$, the optimal encoding now includes an additional codeword $\ketbra{0}{0}$.
As $|\alpha_\text{max}|$ increases further, this vacuum codeword gains probability mass until it transforms into a second ring as the first ring codeword diameter further increase.
% This pattern continues as $|\alpha_\text{max}|$ increases where a new vacuum codeword reappears as the existing ring codewords increases in diameter.
% As $E_\text{max}$ increases to $E_\mathrm{max}'>E_\text{max}$, the diameter of the rings at $E_\text{max}$ increases and at another value of $E_\text{max}$, the vacuum codeword reappears.
More and more rings are introduced in this way as $E_\text{max}$ increases (see Fig.~\ref{fig:phasespace_encoding_coherent_zerotemp}).
We also plot the encoding capacity of the single-ring encoding and the uniform "flat" encoding (see Appendix~\ref{app:vacBounds}) in Fig.~\ref{fig:1}(b).
For large $|\alpha_\mathrm{max}|^2$, the flat encoding is not far from the optimal capacity derived here.

The Wigner function of the average state $\rho_\mathrm{ave} = \int dF^*(\eta) \rho(\eta)$ where $F^*$ is optimal encoding and $\rho(\eta)$ is a ring state for resource $\ket{\alpha_\text{max}}$ is plotted in Fig.~\ref{fig:wigPlot} in Appendix~\ref{app:Wigner_fun} for multiple values of $|\alpha_\text{max}|$. 
\mg{Is this just for the first ring? Exposition unclear.}\at{OK, should be clearer now}
As $|\alpha_\mathrm{max}|$ gets larger, one can observe from Fig.~\ref{fig:wigPlot} and Fig.~\ref{fig:bigAlpHat} that the Wigner function of the average state of the optimal encoding tends to a flat distribution.  
\mg{Again, these are our main results! Should not be relegated to appendix. If we care about length, Proposition 2 can be dropped as it is really just methodology}\at{OK, expanded these below!}

% Because the capacity $\chi_{te}$ as well as the number of the rings in the encoding for a given $E_\mathrm{max}$ can only be computed numerically in general, it is useful to provide bounds that can be computed easily.
% An upper bound in the case of encoding $F$ such that energy of $\rho_\mathrm{ave}$ is bounded by $|\alpha_\text{max}|^2$ (instead of bounding energy of each $\rho(\eta,\theta)$) follows from the result on optimality of Gaussian states on thermal channels by Giovannetti, et.al.~\cite{Giovannetti2014} (see Appendix~\ref{app:upper-bound}).

For a given $|\alpha_\mathrm{max}|$, the one-ring encoding capacity $L_\mathrm{1ring}$ and the flat distribution capacity $L_\mathrm{flat}$ (corresponding to a uniform probability distribution on the disk with $\eta\in[0,1]$ and $\theta\in[0,2\pi]$ -- see also~\cite{Guha2011}) are clearly lower bounds on the thermal encoding capacity.
These bounds are visualized in Fig.~\ref{fig:1}(b), showing that the former is tighter for small $|\alpha_\mathrm{max}|$ and the latter tighter for larger $|\alpha_\mathrm{max}|$.
As our codewords are pure states, the one-ring capacity is the von Neumann entropy of the ring state at $|\alpha_\mathrm{max}|$, given by \eqref{eqn:holevo_1ring}.
% \begin{align}
%     -\sum_{n=0}^\infty  \frac{e^{-E_\mathrm{max}} E_\mathrm{max}^n}{n!} \log\frac{e^{-E_\mathrm{max}}E_\mathrm{max}^n}{n!} \;.
% \end{align}
See Appendix~\ref{sec:lower-bounds_one_ring} for a simple derivation.
For the flat distribution, its capacity is the von Neumann entropy of the average state with encoding $F(\eta)=\eta$, which is given by
\begin{align}
    -\sum_{n=1}^\infty P_n[F] \log P_n[F]
\end{align}
for $P_n[F] = \frac{n! - \Gamma(1+n,E_\mathrm{max})}{n!E_\mathrm{max}}$. 
For a derivation and further discussion, see Appendix~\ref{app:flat_dist}.

% Most of the bounds derived for the classical capacity can be modified to bound the quantum capacity~\cite{Rassouli2016,McKellips2004,Thangaraj2017,Lapidoth2009,Lapidoth2009_2,Dytso2019,Yagli2019}. 
\mg{Context unclear}\at{OK.removed.}

\subsection{Mixed state resource}\label{sec:zero_temp_mixed_state}

\begin{figure}[t]
    \centering
    \begin{tikzpicture}

  \begin{axis}[myaxis,
    xlabel={ $n_\mathrm{res}$},
ymin=0,
ymax=5,
ylabel={Capacity, $\chi$},
legend pos=north east, 
legend style={nodes={scale=0.7, transform shape}}, grid=both
]

\def\xlinewidth{1.0 pt}

%\def\epstwo{0}
%\addplot[color=blue,line width=\xlinewidth,samples=100][domain=0:1]
%{2*(1-x)*(1-\epstwo)/(1-\epstwo+1-x)};
%\addlegendentry{$\epsilon_2=0$};

%\addplot[color=blue,dashed,line width=\xlinewidth,samples=100][domain=0:20]
%{-x*log2(x/(1+x))+log2(1+x)};
%\addlegendentry{$U_\text{GDE}$};

\addplot [allfigs,color=pink,mark=none] table [col sep=comma,search path={data/thermalRes_T0}]{holE20.csv};
\addlegendentry{$E_0=20$};

\addplot [allfigs,color=green,mark=none] table [col sep=comma,search path={data/thermalRes_T0}]{holE16.csv};
\addlegendentry{$E_0=16$};

\addplot [allfigs,color=blue,mark=none] table [col sep=comma,search path={data/thermalRes_T0}]{holE8.csv};
\addlegendentry{$E_0=8$};

\addplot [allfigs,color=red,mark=none] table [col sep=comma,search path={data/thermalRes_T0}]{holE4.csv};
\addlegendentry{$E_0=4$};

\addplot [allfigs,color=black,mark=none] table [col sep=comma,search path={data/thermalRes_T0}]{holE2.csv};
\addlegendentry{$E_0=2$};

%\addplot [allfigs,color=green,mark=none] table [col sep=comma]{LFlatTop.csv};
%\addlegendentry{$L_\text{flat}$};

%\addplot [allfigs,color=red,dotted,mark=none] table [col sep=comma]{L1Ring.csv};
%\addlegendentry{$L_\text{1ring}$};

% \def\ymax{6}

%     \def\discontA{1.257646^2}
%     \draw [dashed] (\discontA,0) -- (\discontA,\ymax);

%     \def\discontB{1.833970^2}
%     \draw [dashed] (\discontB,0) -- (\discontB,\ymax);

%     \def\discontC{2.359160^2}
%     \draw [dashed] (\discontC,0) -- (\discontC,\ymax);

%     \def\discontD{3.132070^2}
%     \draw [dashed] (\discontD,0) -- (\discontD,\ymax);

%     \def\discontE{3.354295^2}
%     \draw [dashed] (\discontE,0) -- (\discontE,\ymax);

%     \def\discontG{4.171020^2}
%     \draw [dashed] (\discontG,0) -- (\discontG,\ymax);

%     \def\discontF{4.272969^2}
%     \draw [dashed] (\discontF,0) -- (\discontF,\ymax);

\end{axis}
\end{tikzpicture}

%%% Local Variables:
%%% mode: latex
%%% TeX-master: "../cv_thermal_encoding.tex"
%%% End:
    \caption{{\bf Capacity for a mixed state resource with fixed energy $E_0$ and mean photon number $n_\res$ at $T=0$.} For each line, the total energy of the resource is fixed at $E_0$. 
    The left end of each line (with $n_\res=0$) corresponds to a pure coherent state resource, while the right end is a thermal state resource with mean photon number $n_\res=E_0$. 
    For each $E_0$, the capacity $\chi$ decreases as the resource mean photon number $n_\res$ increases.}
    \label{fig:cap_therRes_T0}
\end{figure}
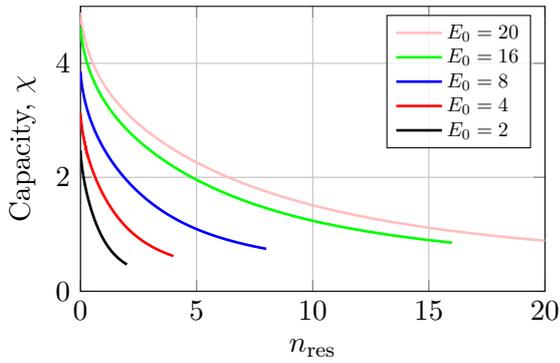

As so far we only discuss pure resource states in this section, we will now consider a mixed Gaussian state as a resource state, namely the displaced thermal state $\rho_\dts(n_\res,\alpha) = D(\alpha)\rho_\th(n_\res)D^\dag(\alpha)$ where $D$ is the displacement operator and $\rho_\th(N)$ is a thermal state with mean photon number $N$.
We numerically optimize the capacities over encodings for displaced thermal state resource with different mean photon number $n_\res$ and fixed total energy $E_0$ and obtained Fig.~\ref{fig:cap_therRes_T0}.
All of the optimal encodings obtained from this numerical optimization -- which satisfy the optimality cosnditions in Proposition~\ref{prop_2} -- has a finite POI, supporting Conjecture~\ref{conj:finite_POI_general}.

Note that when $n_\res=0$, we have a pure coherent state resource $\rho_\dts(0,\alpha)=\ketbra{\alpha}{\alpha}$, and when $n_\res=E_0$, a thermal state with zero amplitude $\rho_\dts(n_\res,0)=\rho_\th(n_\res)$. 
A thermal state with zero amplitude can still be used to transmit information as long as its energy $E_0=n_\res$ is greater than zero.
However given a fixed energy $E_0$, having a coherent state as a resource is more optimal compared to a resource thermal state with the same energy.
Hence, we can transmit more information as the resource state gets closer to a pure coherent state, as shown in Fig.~\ref{fig:cap_therRes_T0}.
% In other words, the capacity decreases as the state becomes more mixed. 
\mg{Seems like an informal statement made from ad-hoc case-study...} \at{OK. Changed that last bit}

\subsection{Nearly-optimal squeezed-state resource} \label{appendix:optimal-state}

Now we go back to the capacity upper-bound $g(E_\mathrm{max})$ in~\eqref{eqn:upper_bound_holevo_t0} for resource state energy $E_\mathrm{max}$.
It is shown in~\cite{Guha2011,Wilde2012} that the optimal resource state that achieves this is the pure non-Gaussian state $\rho=\ketbra{\phi}{\phi}$ where
\begin{align}
  \label{eq:optngs}
  \ket{\phi} = \frac{1}{\sqrt{E_\mathrm{max}+1}}\sum_{n=0}^\infty \ket{n} \sqrt{\frac{E_\mathrm{max}^n}{(E_\mathrm{max}+1)^n}} \;.
\end{align}
This state is optimal because resource state $\ket{\phi}$ and encoding with a uniform phase distribution gives us
\begin{align}\label{eq:2}
  \chi_{te} = (E_\mathrm{max}+1)\log(E_\mathrm{max}+1) - E_\mathrm{max}\log E_\mathrm{max} \;,
\end{align}
which is the maximum capacity for an average-energy constrained channel without any additional constraint on the encoding. 
This capacity (in the average-energy constrained case) can be achieved using coherent state codewords with a Gaussian distribution encoding (see Appendix~\ref{app:upper-bound}). 
% For small $E$, the optimal state has the same Fock state coefficients as the coherent state.
% This implies that encoding with a single ring of coherent states is almost optimal when $E$ is small.

\mg{Given our emphasize on optimal states, it is weird that this concept only comes up here. We should have probably motivated this as one of the main questions our manscript solved in the intro}

The state~(\ref{eq:optngs}) has very high fidelity with respect to a displaced phase-squeezed state even for large $E_\mathrm{max}\sim 20$. 
This motivates us to look at maximum $\chi_{te}$ achievable using a ring of displaced phase-squeezed state. 
We plot the capacity for a displaced squeezed state in Fig.~\ref{fig:1}(d) which shows that it performs almost as good as (\ref{eq:optngs}). 
The displacement is chosen to match the quadrature expectation value of~(\ref{eq:optngs}) while the squeezing is set such that the energy of the state is fixed at $E_\mathrm{max}$ \footnote{Optimising the displacement gives only a slight improvement which is indistinguishable in the plot.}. 
The Wigner function and photon number distribution of the average state using the displaced phase-squeezed state and the optimal state for $E_\mathrm{max}=3$ is plotted in Fig.~\ref{fig:wigOpt} of Appendix~\ref{app:Wigner_fun}.

\section{Non-zero temperature environment}\label{sec_mixed_state}

In this section we consider the case of $T>0$, i.e. the environment is in a thermal state $\gamma_T = \rho_\th(n_\mathrm{env})$ with mean photon number $n_\mathrm{env}>0$.
We will first discuss the encoding capacity when the resource state is a thermal state and show that the optimal encoding satisfies Conjecture~\ref{conj:finite_POI_general}.
We then proceed to discuss numerical results for coherent state resource $\ket{\alpha}$.

\subsection{Thermal state resource}

When the resource state $\rho$ is a thermal state $\rho_\th(n_\res)$ with mean photon number $n_\res$, the output codewords are also thermal states of the form $\rho_\th(n_\eta)$ where $n_\eta = \eta n_\res + (1-\eta)n_\mathrm{env}$ is a mixture of the resource mean photon number and environment mean photon number. 
Therefore, for a given encoding $F$ the output average state is a mixture of thermal states $\rho_\mathrm{ave} = \int dF(\eta) \rho_\th(n_\eta)$.
Since thermal states are invariant under phase shifts, it is useless to encode any information in the phase $\theta$ \footnote{There is no contradiction with Proposition 1, however, since for any encoding $F$, the circularly symmetrized encoding $F'$ has the same Holevo information.}. 
While this is not ideal for maximizing information transfer, it also means that the optimum decoding measurement is just a measurement of the photon number, which is readily made in practice.
In this scenario, the cardinality of the distribution support corresponding to an optimal encoding is finite, as proven in Appendix~\ref{appendix:finite_POI_thermal_res}.

\begin{prop}\label{prop:finite_POI_thermal}
    For a thermal-state resource $\rho = \rho_\th(n_\res)$ and any $T \geq 0$, the attenuation-coefficients $\eta$ of the optimal circularly-symmetric encoding takes a finite number of values in $[0,1]$.
    % For any optimal circularly symmetric thermal encoding $F^*$ with a thermal state resource in any temperature, its probability distribution support over attenuation $\eta\in[0,1]$ must have a finite cardinality.
\end{prop}

The special case when the resource state has mean photon number $n_\res=0$ (i.e. a vacuum state) reduces to the scenario in section~\ref{sec:zero_temp_mixed_state}.
This is due to the codeword corresponding to the transmittance value $\tilde{\eta}$ being equal to the codeword corresponding to $\eta = 1 -\tilde{\eta}$ in section~\ref{sec:zero_temp_mixed_state} with $\rho = \rho_\th(n_\mathrm{env})$ (the displacement is set to zero). 
The optimal encoding in this case is plotted in Fig.~\ref{fig:vrCode}(a). 
We find that for $n_\mathrm{env}< 8.67754$, the optimal encoding has exactly two codewords: the resource vacuum state ($\eta=0$) and the thermal state ($\eta=1$).
In Appendix~\ref{app:vacThres} we obtain an analytical form of the Holevo information
\begin{align}\label{eq:27}
    \chi_{te} = \log\left( 1+\frac{n_\mathrm{env}}{(1+n_\mathrm{env})^\frac{1+n_\mathrm{env}}{n_\mathrm{env}}} \right) \;.
\end{align}
This quantity tends to $1$ as $n_\mathrm{env}\rightarrow\infty$.
As $n_\mathrm{env}$ increases beyond $8.67754$, it can be seen in Fig.~\ref{fig:vrCode}(a) that adding more codewords is necessary to achieve the thermal encoding capacity (black line in Fig.~\ref{fig:vrCode}(b)).
Thus interestingly, it is possible to encode information using a vacuum state resource whenever the environmental temperature is non-zero.
% If the role of the resource and environment are swapped 
\mg{Not sure what this means}\at{fixed (the "swapping" stuff)}
% , the same encoding, with $\eta$ replaced by $1-\eta$ will still be optimal. 

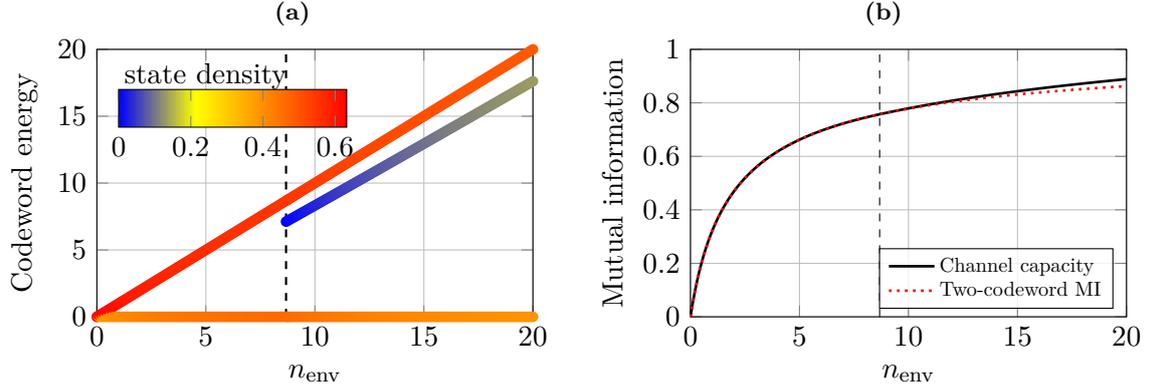
\begin{figure*}
    \centering
    \begin{subfigure}{0.45\textwidth}
    \subcaption[short for lof]{}
    \begin{tikzpicture}

%\def\epstwo{0}
%\addplot[color=blue,line width=\xlinewidth,samples=100][domain=0:1]
%{2*(1-x)*(1-\epstwo)/(1-\epstwo+1-x)};
%\addlegendentry{$\epsilon_2=0$};

\begin{axis}[myaxis,
    ymax=20,
    grid=both,
    colorbar horizontal,
    colorbar style={
      at={(.05,0.85)},
      anchor=north west,
      width=3cm
    },
    xlabel=$n_\mathrm{env}$ , 
    ylabel=Codeword energy
    ]

    \node[] at (axis cs:5,18) {state density};

% vacuum resource, and thermal free state
%    \addplot[scatter,only marks,scatter src=explicit,draw opacity=0]
%    table [col sep=comma,x index=0,y index=1,meta index=2] {VRscatter.csv};

%    thermal resource and vacuum free state T=0.
%    \addplot[scatter,only marks,scatter src=explicit,draw opacity=0]
 %   table [col sep=comma,x index=0,y index=1,meta index=2] {VRscatterOneMinusEta.csv};

    \addplot[scatter,only marks,scatter src=explicit,draw opacity=0]
    table [col sep=comma,x index=0,y index=1,meta index=2] {VRscatter_E0_Ecode.csv};

    \def\discontA{8.67754}
    \draw [dashed,thick] (\discontA,0) -- (\discontA,20);

\end{axis}
\end{tikzpicture}
%%% Local Variables:
%%% mode: latex
%%% TeX-master: "../cv_thermal_encoding.tex"
%%% End:
%
    \end{subfigure}
    \begin{subfigure}{0.45\textwidth}
    \subcaption[short for lof]{}
    \begin{tikzpicture}

  \begin{axis}[myaxis,xlabel={$n_\mathrm{env}$},ymax=1,
    legend pos=south east, 
    legend style={nodes={scale=0.7, transform shape}}, grid=both]

\def\xlinewidth{1.0 pt}

%\def\epstwo{0}
%\addplot[color=blue,line width=\xlinewidth,samples=100][domain=0:1]
%{2*(1-x)*(1-\epstwo)/(1-\epstwo+1-x)};
%\addlegendentry{$\epsilon_2=0$};

\addplot [allfigs,color=black,mark=none] table [col sep=comma,x index=0,y index=1]{vacRes.csv};
\addlegendentry{Channel capacity};

\addplot[color=red,dotted,line width=\xlinewidth,samples=100][domain=0:20]
{log2(1+x/(1+x)^((1+x)/x))};
\addlegendentry{Two-codeword MI};

%\addplot [allfigs,color=green,mark=none] table [col sep=comma]{LFlatTop.csv};
%\addlegendentry{$L_\text{flat}$};

%\addplot [allfigs,color=red,dotted,mark=none] table [col sep=comma]{L1Ring.csv};
%\addlegendentry{$L_\text{1ring}$};

\def\ymax{6}

    \def\discontA{8.67754}
    \draw [dashed] (\discontA,0) -- (\discontA,\ymax);

\end{axis}
\end{tikzpicture}
%%% Local Variables:
%%% mode: latex
%%% TeX-master: "../cv_thermal_encoding.tex"
%%% End:
    \end{subfigure}
    \caption{ 
    \textbf{Encoding using vacuum state resource at temperature $T>0$ (i.e. environment mean photon number $n_\mathrm{env}>0$).}
    \textbf{(a)} Optimal encoding using vacuum state resource with $n_\mathrm{env}>0$. 
    The shading indicates the probability $p_j$ of the $j$-th codeword.  
    For $n_\mathrm{env} < 8.67754$, the two-codeword encoding consisting of the vacuum resource and the thermal environment is optimal (Appendix~\ref{app:vacThres}). 
    When $n_\mathrm{env}$ increases beyond $8.67754$ (vertical line), the optimal encoding includes a third codeword.
    \textbf{(b)} Capacity and bounds for vacuum state resource as a function of $n_\mathrm{env}$.
    The mutual information from a two-codeword encoding (dotted red line) is optimal when $n_\mathrm{env}< 8.67754$, coinciding with the numerically computed capacity (black line).
    }
    \label{fig:vrCode}
\end{figure*}

\subsection{Coherent state resource}\label{sec:coherent_resource_nonzero_T}

Finally, we consider the capacity of a coherent state resource $\ket{\alpha_\mathrm{max}}$ (resource energy $E_\mathrm{max}=|\alpha_\mathrm{max}|^2$) when the environment is in a thermal state with mean photon number $n_\text{env}>0$.
\mg{Do we just mean an environment of non-zero temperature? Note sure what thermal photon numbers mean.} \at{OK. changed it to "mean photon number of environment state"}
The numerically optimized thermal encoding capacity is shown in Fig.~\ref{fig:finiteTCap}. 
For a fixed $n_\mathrm{env}$, a larger $E_\mathrm{max}$ gives a higher capacity.
However when we vary $E_\mathrm{max}$, a lower $n_\mathrm{env}$ has higher capacity for larger $E_\mathrm{max}$, whereas larger $n_\mathrm{env}$ has higher capacity for smaller $E_\mathrm{max}$ (see Fig.~\ref{fig:finiteTCap} inset), which is akin to the phenomenon seen in the previous subsection for a vacuum-state resource.
In fact, when $n_\mathrm{env}\gg 1$ and $\alpha_\mathrm{max}=0$, we can get $\chi_{te}=1$ using a two-codeword distribution (Appendix~\ref{app:vacThres}).

\begin{figure}
    \centering
    \begin{tikzpicture}[spy using outlines={rectangle, magnification=3,connect spies}]

\begin{axis}[myaxis,legend pos=south east,enlargelimits=false,ymax=5, legend style={nodes={scale=0.7, transform shape}}, grid=both,xlabel=$E_\mathrm{max}$, ylabel=Mutual information]

%\def\xlinewidth{0.5 pt}

%\def\epstwo{0}
%\addplot[color=blue,line width=\xlinewidth,samples=100][domain=0:1]
%{2*(1-x)*(1-\epstwo)/(1-\epstwo+1-x)};
%\addlegendentry{$\epsilon_2=0$};

\tikzstyle{allfigs}=[line width=1pt,samples=100,domain=0:20]

\addplot [allfigs,color=black,mark=none] table [col sep=comma]{T0_nRes0.csv};
\addlegendentry{ $n_\text{env}=0$};

\addplot [allfigs,color=green,mark=none] table [col sep=comma]{hol_n1_dim150.csv};
\addlegendentry{$n_\text{env}=1$};

\addplot [allfigs,color=green,mark=none,forget plot] table [col sep=comma]{hol_n1_dim150_run2.csv};

\addplot [allfigs,color=brown,mark=none] table [col sep=comma]{hol_n2_dim150.csv};
\addlegendentry{$n_\text{env}=2$};

\addplot [allfigs, color=blue,mark=none] table [col sep=comma]{hol_n4_dim150.csv};
\addlegendentry{$n_\text{env}=4$};

\addplot [allfigs, color=purple,mark=none] table [col sep=comma]{hol_n20_dim150.csv};
\addlegendentry{$n_\text{env}=20$};

% \addplot [allfigs,color=red,dotted,mark=none] table [col sep=comma]{L1Ring.csv};
% \addlegendentry{$L_\text{1ring}$};

%\coordinate (spypoint) at (axis cs:0,0);
%\coordinate (spyviewer) at (axis cs:4,1);
%\spy[width=0.1\columnwidth,height=0.4\columnwidth] on (spypoint) in node [fill=white] at (spyviewer);
\coordinate (insetPosition) at (axis cs:2.3,0.1);

\end{axis}

\begin{axis}[width=0.45\columnwidth,height=0.35\columnwidth,
%\begin{axis}[width=0.8\columnwidth,height=0.6\columnwidth,
  %enlargelimits=false,
  ymin=-0.1,ymax=1.5,
  xmin=0.001,xmax=1,
  at={(insetPosition)},
  anchor={outer south west},
  tick label style={font=\footnotesize},
  xtick={0.001,0.01,0.1,1},
  xticklabels={$10^{-3}$,$0.01$,$0.1$,$1$},
  xmode=log]
  
  \tikzstyle{allfigs}=[line width=1pt,mark=none]
\input{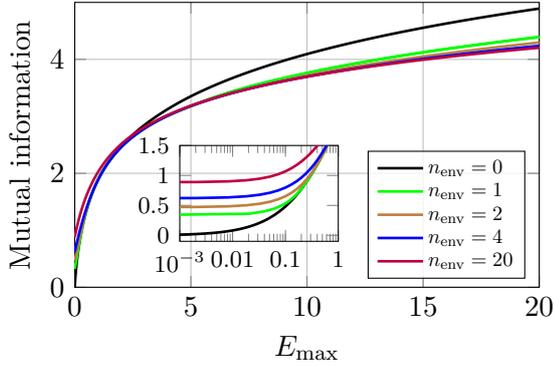}
\end{axis}

\end{tikzpicture}
% inset as stand alone figure
%\begin{tikzpicture}
%\begin{axis}[width=0.8\columnwidth,height=0.6\columnwidth,
%  enlargelimits=false,ymax=1.5,xmax=0.6]  
%  \tikzstyle{allfigs}=[line width=0.5pt,mark=*]
%\input{./figures/finT_inset.tikz}
%\end{axis}
%\end{tikzpicture}
%%% Local Variables:
%%% mode: latex
%%% TeX-master: "../cv_thermal_encoding.tex"
%%% End:
    \caption{
    {\bf Thermal-encoding capacity for coherent state resource $\ket{\alpha_\mathrm{max}}$ at $T>0$ environment.} 
    Capacity varies with environment mean photon number $n_\mathrm{env}$ (see Section~\ref{sec:coherent_resource_nonzero_T}).
    % For large $\alpha_\mathrm{max}$, less environment noise gives a higher capacity. When $n_\mathrm{env}$ is large, the one-ring distribution is close to being optimal. 
    % For small $\alpha_\mathrm{max}$, more environment noise gives a higher capacity (see inset).
    }
\label{fig:finiteTCap}
\end{figure}

We now look into a communication scenario over a lossy channel which mixes a low-energy $E_\mathrm{max}\ll1$ coherent state resource at a fixed attenuation-coefficient $\eta_\mathrm{ch}$.
In this scenario we have a single-ring circularly symmetric encoding and displaced thermal state codewords
\begin{align}\label{eqn:codeword_thermal_photon_number}
    \rho(\eta_\mathrm{ch},\theta) = R_\theta D(\sqrt{\eta_\mathrm{ch}}\alpha) \rho_\th(n_\mathrm{ch}) D^\dagger(\sqrt{\eta_\mathrm{ch}}\alpha)R_\theta^\dagger \;,
\end{align}
where $\rho_\th(n_\mathrm{ch})$ is a thermal state with mean photon number $n_\mathrm{ch} = (1-\eta_\mathrm{ch})n_\mathrm{env}$, we call $n_\mathrm{ch}$ the codeword \emph{thermal photon number}.
The thermal encoding Holevo information is therefore given by 
\begin{align}\label{eqn:one_ring_capacity_coherent_nonzero_T}
    \chi[E_\mathrm{max}] = S(\rho_\mathrm{ave}) - S(\rho_\th(n_\mathrm{ch})) \;,
\end{align}
as the codeword entropy is equal for all phase $\theta$.

For $E_\mathrm{max}\ll1$, it is shown in Fig.~\ref{fig:coh_nonzero_t_optimal} that by using a coherent state resource we can almost reach the upper bound of the lossy-channel capacity~\cite[supplementary information, eqn.(12)]{Giovannetti2014}, given by
\begin{align}\label{eqn:cap_upper_bound_lossy}
    g(\eta_\mathrm{ch} E_\mathrm{max} + (1-\eta_\mathrm{ch}) n_\mathrm{env}) - g((1-\eta_\mathrm{ch}) n_\mathrm{env}) \;.
\end{align}
In fact for $E_\mathrm{max}\ll1$, we can approximate $\chi[E_\mathrm{max}]$ by
\mg{Lossing $\neq$ Noisy. May want to change subtitle}\at{OK done}\mg{Paragraphs appear disjointed. Also seems to be more focused on results of others than out own}\at{OK, moved this part to after our result below as a comparison}\mg{Still not sure what this means? Do we mean environmental temperature $T$?}\at{OK, changed it}
\begin{align}\label{eq:1}
  \Tilde{\chi}_{te}[E_\mathrm{max}] = \eta_\mathrm{ch} E_\mathrm{max} \log\frac{1+n_\mathrm{ch}}{n_\mathrm{ch}} \;,
\end{align}
which is derived in Appendix~\ref{app:low_E_limit} and shown as the dotted line in Fig.~\ref{fig:coh_nonzero_t_optimal}.
However as the dotted lines goes above the solid lines for large $E_\mathrm{max}$, indicating that \eqref{eq:1} is not a good approximation of \eqref{eqn:one_ring_capacity_coherent_nonzero_T} as $E_\mathrm{max}$ increases.
Interestingly as $n_\mathrm{ch}$ increases, indicating more thermal noise (see blue and black lines in Fig.~\ref{fig:coh_nonzero_t_optimal}), \eqref{eq:1} is a better approximation of \eqref{eqn:cap_upper_bound_lossy} as the two lines are close to each other even up to the values $E_\mathrm{max}>1$.
We also found that this phenomena occurs as $\eta_\mathrm{ch}$ decreases.

\begin{figure}[!ht]
    \centering
    \begin{tikzpicture}[spy using outlines={rectangle, magnification=3,connect spies}]

\begin{axis}[myaxis,legend pos=south east,enlargelimits=false,ymax=1.9,xmax=10, legend style={at={(0.35,0.45)},nodes={scale=0.7, transform shape}}, grid=both,xlabel=$E_\mathrm{max}$, ylabel=Mutual information]

%\def\xlinewidth{0.5 pt}

%\def\epstwo{0}
%\addplot[color=blue,line width=\xlinewidth,samples=100][domain=0:1]
%{2*(1-x)*(1-\epstwo)/(1-\epstwo+1-x)};
%\addlegendentry{$\epsilon_2=0$};

\tikzstyle{allfigs}=[line width=1pt,samples=100,domain=0:20]

\addplot [allfigs, color=red,mark=none,line width=0.75pt] table [col sep=comma]{lossy_channel/cap_ub_2.csv};
\addlegendentry{$n_\text{env}=2$};
\addplot [allfigs, color=purple,mark=none,line width=0.75pt] table [col sep=comma]{lossy_channel/cap_ub_5.csv};
\addlegendentry{$n_\text{env}=5$};
\addplot [allfigs, color=brown,mark=none,line width=0.75pt] table [col sep=comma]{lossy_channel/cap_ub_7.csv};
\addlegendentry{$n_\text{env}=7$};
\addplot [allfigs, color=green,mark=none,line width=0.75pt] table [col sep=comma]{lossy_channel/cap_ub_10.csv};
\addlegendentry{$n_\text{env}=10$};
\addplot [allfigs, color=blue,mark=none,line width=0.75pt] table [col sep=comma]{lossy_channel/cap_ub_20.csv};
\addlegendentry{$n_\text{env}=20$};
\addplot [allfigs, color=black,mark=none,line width=0.75pt] table [col sep=comma]{lossy_channel/cap_ub_30.csv};
\addlegendentry{$n_\text{env}=30$};

\addplot [allfigs, color=red,dotted,mark=none,line width=0.5pt] table [col sep=comma]{lossy_channel/cap_onering_2.csv};
\addplot [allfigs, color=purple,dotted,mark=none,line width=0.5pt] table [col sep=comma]{lossy_channel/cap_onering_5.csv};
\addplot [allfigs, color=brown,dotted,mark=none,line width=0.5pt] table [col sep=comma]{lossy_channel/cap_onering_7.csv};
\addplot [allfigs, color=green,dotted,mark=none,line width=0.5pt] table [col sep=comma]{lossy_channel/cap_onering_10.csv};
\addplot [allfigs, color=blue,dotted,mark=none,line width=0.5pt] table [col sep=comma]{lossy_channel/cap_onering_20.csv};
\addplot [allfigs, color=black,dotted,mark=none,line width=0.5pt] table [col sep=comma]{lossy_channel/cap_onering_30.csv};

\addplot [allfigs, color=red,dashed,mark=none] table [col sep=comma]{lossy_channel/cap_exact_2.csv};
\addplot [allfigs, color=purple,dashed,mark=none] table [col sep=comma]{lossy_channel/cap_exact_5.csv};
\addplot [allfigs, color=brown,dashed,mark=none] table [col sep=comma]{lossy_channel/cap_exact_7.csv};
\addplot [allfigs, color=green,dashed,mark=none] table [col sep=comma]{lossy_channel/cap_exact_10.csv};
\addplot [allfigs, color=blue,dashed,mark=none] table [col sep=comma]{lossy_channel/cap_exact_20.csv};
\addplot [allfigs, color=black,dashed,mark=none] table [col sep=comma]{lossy_channel/cap_exact_30.csv};

%\coordinate (spypoint) at (axis cs:0,0);
%\coordinate (spyviewer) at (axis cs:4,1);
%\spy[width=0.1\columnwidth,height=0.4\columnwidth] on (spypoint) in node [fill=white] at (spyviewer);
\coordinate (insetPosition) at (axis cs:2.3,0.1);

\end{axis}

% \begin{axis}[width=0.45\columnwidth,height=0.35\columnwidth,
% %\begin{axis}[width=0.8\columnwidth,height=0.6\columnwidth,
% %enlargelimits=false,
% ymin=-0.1,ymax=1.5,
% xmin=0.001,xmax=1,
% at={(insetPosition)},
% anchor={outer south west},
% tick label style={font=\footnotesize},
% xtick={0.001,0.01,0.1,1},
% xticklabels={$10^{-3}$,$0.01$,$0.1$,$1$},
% xmode=log]
% \tikzstyle{allfigs}=[line width=1pt,mark=none]
% \input{./figures/finT_inset.tikz}
% \end{axis}

\end{tikzpicture}
% inset as stand alone figure
%\begin{tikzpicture}
%\begin{axis}[width=0.8\columnwidth,height=0.6\columnwidth,
%  enlargelimits=false,ymax=1.5,xmax=0.6]  
%  \tikzstyle{allfigs}=[line width=0.5pt,mark=*]
%\input{./figures/finT_inset.tikz}
%\end{axis}
%\end{tikzpicture}
%%% Local Variables:
%%% mode: latex
%%% TeX-master: "../cv_thermal_encoding.tex"
%%% End:
    \caption{
    {\bf
    Lossy-channel with attenuation $\eta_\mathrm{ch}=0.4$ at $T>0$ environment.}
    For $E_\mathrm{max}\ll1$ and various environment mean photon-number $n_\mathrm{env}$, 1-ring encoding capacity with coherent state resource~\eqref{eqn:one_ring_capacity_coherent_nonzero_T} (dashed lines) almost reaches the lossy-channel capacity upper-bound~\eqref{eqn:cap_upper_bound_lossy} (solid lines) and can be approximated by \eqref{eq:1} (dotted lines), see Section~\ref{sec:coherent_resource_nonzero_T}.
    }
    % \vspace{-1.5\baselineskip}
\label{fig:coh_nonzero_t_optimal}
\end{figure}
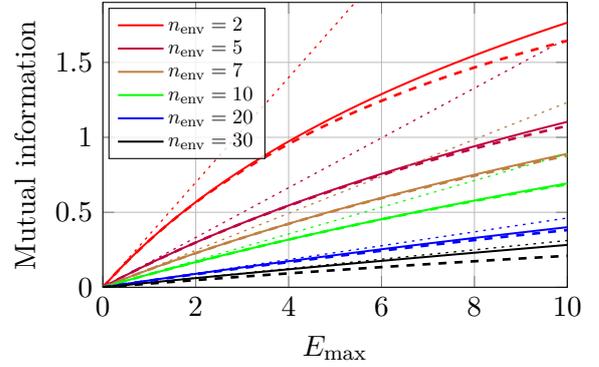

Approximation \eqref{eq:1} implies that the information capacity per unit photon tends to the value $\lim_{E_\mathrm{max}\rightarrow0} \Tilde{\chi}_{te}[E_\mathrm{max}]/E_\mathrm{max} = \eta_\mathrm{ch} \log\frac{1+n_\mathrm{ch}}{n_\mathrm{ch}}$, which is finite because $n_\mathrm{ch}>0$.
Interestingly, a result by Guha and Shapiro~\cite{Guha2013} shows that a coherent state resource and binary phase-shift-keying encoding in a vacuum environment ($n_\mathrm{ch}\rightarrow0$) has $\chi_{te}[E]/E \rightarrow\infty$ as $E\rightarrow0$, allowing transmission of unlimited bits per photon when the average photon in the coherent state is small.

\section{Discussion}\label{sec:discussion}

\mg{Discussion should always begin with a paragraph that summarizes what we have discovered. Ideally, it should also be written in a way that can be understood by readers who did not necessarily read the results (e.g. an editor). So try to avoid using mathematical symbols defined in the body.}\at{OK Done.}
\mg{The discussion should also be a place to relate our work to other related work e.g., quantum reading}\at{Still need to discuss related work!}

% In particular, our result on the optimality of circulary symmetric encodings (Proposition~\ref{prop_1}) finds an analog in the optimality of such encodings for the peak-constrained 2-dimensional AWGN channel~\cite{Shamai1995}, while our condition on the optimal distribution of the attenuation coefficient (Proposition~\ref{prop_2}) is similar to the seminal result of Smith~\cite{smith_1969information,Smith1971} on the capacity of the real-valued AWGN channel under a peak energy constraint.
% These similarities extend to more recent work on the capacity of $n$-dimensional AWGN channels under a peak energy constraint (see, e.g,. \cite{Dytso2019,Yagli2019} and reference therein).

In this work, we studied the encoding of information by modulating the phase and transmittance of a peak-energy constrained resource state using passive linear optical channels.
% motivated by practical considerations on linear optical channels that are 
The thermal channel family -- which has been used to model various tasks such as optical memory reading, communication, and algorithmic cooling -- is simple to implement and requires no external energy source.
% we consider the encoding of information by modulating a resource state using passive linear optical channels called the thermal channel.
We showed that the maximum amount of information that can be encoded by such channels applied to a given resource state can always be achieved by an encoding scheme that uniformly distributes the phase introduced by the channel (Proposition~\ref{prop_1}).
% The information capacity associated with transmitting a given resource state through this family of channels, as quantified by the Holevo information, can always be achieved by circularly symmetric encodings which uniformly distribute the phase introduced by the channel (Proposition~\ref{prop_1}).
Among such encodings, the unique optimal encoding is characterized by information-theoretic conditions of how the thermal channel transmittance is distributed (Proposition~\ref{prop_2}).
% Among such encodings, the unique optimal encoding is characterized by information-theoretic conditions on the marginal information density associated to the transmittance distribution of the encoding (Proposition~\ref{prop_2}).

We conjectured that all optimal circularly symmetric encoding schemes involve only a finite number of values of the channel transmittance.
% Moreover in the case of coherent state input in zero-temperature, its attenuation distribution has finite support and thus nicely forms rings in the phase space with their radius determined by the channel attenuation and energy of the coherent state.
This conjecture is supported by showing that this condition holds in the case of coherent state resource in zero-temperature environments (Proposition~\ref{prop:finite_POI_coht0}) and in the case of thermal state resource interacting with an environment at any temperature (Proposition~\ref{prop:finite_POI_thermal}).
We also supported this conjecture by numerical evidence based on the aforementioned encoding optimality conditions for cases where the channel output is a displaced thermal state.
% with varying input energy and environment temperature.
An intriguing question that one may ask is whether there is a fundamental reason why the optimal circularly-symmetric encoding has a discrete attenuation distribution.
% In a zero-temperature environment, we also investigate how the Holevo information of the thermal channel changes with various input energy and with different input resource states such as the displaced phase-squeezed state and displaced thermal state.
% Whereas in the case of environments with nonzero-temperature, we investigate the thermal channel capacity for input states: the vacuum state and coherent state. 
% where for the latter we numerically show a trade-off between input energy and environment temperature.

Interestingly, Proposition~\ref{prop_1} finds an analog in the optimality of such encodings for the classical peak-constrained 2-dimensional AWGN channel~\cite{Shamai1995}, while the properties of optimal distribution of the attenuation coefficient in Proposition~\ref{prop_2} is similar to the seminal result of Smith~\cite{smith_1969information,Smith1971} on the capacity of the real-valued AWGN channel under a peak energy constraint.
These similarities extend to more recent work on $n$-dimensional AWGN channel capacity under a peak energy constraint (see~\cite{Dytso2019,Yagli2019} and references therein).
In the quantum domain, our work opens the question of the ultimate capacity of bosonic Gaussian channels under a peak energy constraint on the input ensemble, i.e., to the problem of optimizing the Holevo information over such ensembles generated by thermal encodings and beyond.

\mg{This paragraph is very difficult to decipher, and I don't think any general reader will understand why we care... may need rephrasing, reordering (to a more technical section), or removal}\at{OK. Simplified.}

\mg{Note sure if I follow the arguments here? Appears awfully speculative? Also not sure what last sentence means? We already know that the three codeword solution on qubits is optimal.}\at{OK. I agree. Removed.}

\section*{Acknowledgement}

This work is supported by the Singapore Ministry of Education Tier 2 Grant MOE-T2EP50221-0005, The Singapore Ministry of Education Tier 1 Grants RG146/20 and RG77/22, grant no.~FQXi-RFP-1809 (The Role of Quantum Effects in Simplifying Quantum Agents) from the Foundational Questions Institute and Fetzer Franklin Fund (a donor-advised fund of Silicon Valley Community Foundation) and the National Research Foundation, Singapore, and Agency for Science, Technology and Research (A*STAR) under its QEP2.0 programme (NRF2021-QEP2-02-P06). 
R.N. thanks Saikat Guha for useful discussions on quantum reading capacity. 
V.N. acknowledges support from the Lee Kuan Yew Endowment Fund (Postdoctoral Fellowship).
A.T. acknowledges support from CQT PhD scholarship.

% \bibliographystyle{quantum}
% \bibliography{short}

\bibliographystyle{quantum} % We choose the "plain" reference style
\bibliography{short} % Entries are in the refs.bib file

% \bibliographystyleother{quantum}
% \bibliographyother{short}

% \printbibliography

\newpage
\onecolumngrid
\appendix

\section{Proofs of general properties of optimal thermal encoding}

\subsection{Proof of proposition~\ref{prop_1}}\label{prop_1_proof}

Consider a thermal encoding $F$ and its corresponding codewords $\{\rho(\eta,\theta)\}_{\eta,\theta}$.
The Holevo information of this encoding $F$ is given by
\begin{align}
  \label{eq:20}
  \chi_{tc}[F] = S(\rho_\text{ave}) - \int dF(\eta,\theta) S(\rho(\eta,\theta)) \;,
\end{align}
where $S$ is the von Neumann entropy and $\rho_\text{ave} =\int dF(\eta,\theta) \rho(\eta,\theta)$ is the average density operator with respect to $F$.
First we will show that any state that is an average of codewords over uniformly distributed phase is diagonal in Fock basis.
% Now, consider a symmetric encoding $\Tilde{F}(\eta,\theta)$, where phase $\theta$ is uniformly distributed and transmissivity $\eta$ is distributed according to some arbitrary $F(\eta)$, hence $\eta$ and $\theta$ are mutually independent.
% The average density operator with encoding $\tilde{F}$ is therefore
% \begin{align} \label{eq:21}
%     \tilde\rho_\text{ave} &= \int d\tilde{F}(\eta,\theta) \rho(\eta,\theta) = \int \frac{d\theta}{2\pi} \int dF(\eta) \rho(\eta,\theta)  \;.
% \end{align}

\begin{lemma}\label{lem:avg_phase_diagonal_fock_basis_state}
    For any state $\rho = \int dF(x) \rho_x$ with pure $\rho_x$'s, its average state with uniform phase distribution $\Tilde{\rho} = \frac{1}{2\pi} \int d\phi\; R_\phi\rho R_\phi^\dag$ for $\phi\in[-\pi,\pi]$ phase rotation operator $R_\phi$ is diagonal in the Fock basis.
\end{lemma}
\begin{proof}
    This can be shown by writing each pure state in terms of Fock basis $\rho_x = \sum_{n,m} z_n(x) z_m(x)^* \ketbra{n}{m}$ for $z_n(x),z_m(x)\in\mathbb{C}$, then we get
    \begin{align*}
        \Tilde{\rho} &= \int dF(x) \int \frac{d\phi}{2\pi} R_\phi \rho_x R_\phi^\dag \\
        &= \int dF(x) \sum_{n,m} z_n(x) z_m(x)^* \int \frac{d\phi}{2\pi} R_\phi \ketbra{n}{m} R_\phi^\dag \\
        &= \int dF(x) \sum_{n,m} z_n(x) z_m(x)^* \int \frac{d\phi}{2\pi} e^{i(n-m)\phi} \ketbra{n}{m} \\
        &= \sum_{n,m} \int dF(x) z_n(x) z_m(x)^* \delta_{n,m} \ketbra{n}{m} = \sum_n \ketbra{n}{m} \int dF(x) |z_n(x)|^2 \;,
    \end{align*}
    which shows that $\Tilde{\rho}$ is diagonal in Fock basis $\{\ket{n}\}_{n=0}^\infty$ with real diagonal elements $\{\int dF(x) |z_n(x)|^2\}_{n=0}^\infty$.
\end{proof}

Now we are ready for the proof of Proposition~\ref{prop_1}.
Namely, we show that for an arbitrary thermal encoding $F$, there is a circularly symmetric encoding $\Tilde{F}$ giving an average state $\Tilde{\rho}_\mathrm{ave} =  \int dF(\eta) \int \frac{d\phi}{2\pi}\; \rho(\eta,\phi)$ with at least the same entropy.
The first term of \eqref{eq:20} is upper bounded by the entropy of $\Tilde{\rho}_\mathrm{ave}$
\begin{align}
    S(\rho_\mathrm{ave})
    &\stackrel{(1)}{=} \int \frac{d\phi}{2\pi}\; S\left( \int dF(\eta,\theta) R_\phi \rho(\eta,\theta) R_\phi^\dag \right) \\
    &\stackrel{(2)}{\leq} S\left( \int dF(\eta,\theta) \int \frac{d\phi}{2\pi}\; R_\phi \rho(\eta,\theta) R_\phi^\dag \right) \\
    &\stackrel{(3)}{=} S\left( \int dF(\eta,\theta) \int \frac{d\phi}{2\pi}\; \rho(\eta,\phi) \right) \\
    &\stackrel{(4)}{=} S\left( \int dF(\eta) \int \frac{d\phi}{2\pi}\; \rho(\eta,\phi) \right) = S(\Tilde{\rho}_\mathrm{ave})
\end{align}
where (1) is because phase rotation $R$ is a unitary operation, thus does not change the entropy of a state, and (2) by the concavity of the von Neumann entropy, and (3) by defining $\rho(\eta,\phi) = R(\phi) \rho(\eta,\theta) R^\dag(\phi)$ because state $\int \frac{d\phi}{2\pi}\; \rho(\eta,\phi)$ is diagonal in Fock basis by Lemma \ref{lem:avg_phase_diagonal_fock_basis_state}, and hence independent of phase $\theta$, and lastly (4) by marginalizing $\int dF(\eta,\theta)$ over phase $\theta$.

Now we look at the second term of \eqref{eq:20}.
In the case that all codeword $\rho(\eta,\theta)$ is a pure state, the second term is zero and therefore $\chi_{tc}[\Tilde{F}] \geq \chi_{tc}[F]$ holds because $S(\Tilde{\rho}_\mathrm{ave}) \geq S(\rho_\mathrm{ave})$. 
In the case of general codewords $\rho(\eta,\theta)$, the second term of \eqref{eq:20} is equal to $\int d\Tilde{F}(\eta,\theta) S(\rho(\eta,\theta))$ because $S(\rho(\eta,\theta))$ is independent of phase $\theta$.
Hence we constructed a circularly symmetric encoding $\Tilde{F}$ such that $\chi_{tc}[\Tilde{F}] \geq \chi_{tc}[F]$ for an arbitrary encoding $F$.

\subsection{Levy metric}\label{app:levy_metric}

For completeness, in this section we define the Levy metric which we will use as a metric on the space of encodings (i.e. cumulative distribution functions) in order to prove some of the properties of an optimal circularly symmetric encoding.
Here we specifically define the metric on $\mathscr{F}$, the space of encodings (i.e. cumulative distribution over the interval $[0,1]$).
The Levy metric $d_{Le}:\mathscr{F}\times\mathscr{F}\rightarrow\R$ is defined as 
\begin{align*}
    d_{Le}(F,G) = \frac{1}{\sqrt{2}} \sup_{x,x' : F(x)+x=G(x')+x'} d\Big( (x,F(x)) , (x',G(x')) \Big) \;.
\end{align*}
Here $d(\cdot,\cdot):\R^2\times\R^2\rightarrow\R$ is the euclidean distance in $\R^2$ and we maximize the distance between two points on the graphs $F$ and $G$ intersecting a diagonal line $a=x+y$ (on an $(x,y)$-coordinate plane), over all $a\in\R$.
% is over all diagonal $a=x+y$ on an $(x,y)$-coordinate plane 
% the distance between two points on the diagonal line $a=x+y$ (on an $(x,y)$-coordinate plane) that intersect the graphs $F$ and $G$.
% More precisely, 
% \begin{gather*}
%     d(F,G,a) = \sqrt{ |F(a_F)-G(a_G)|^2 + |a_F-a_G|^2}
% \end{gather*}
% where $a_F,a_G$ are numbers in the interval $[0,1]$ such that $F(a_F)+a_F=a$ and $G(a_G)+a_G=a$ (note that $a_F$ and $a_G$ are unique given $F$, $G$, and $a$).
(for discussions on its geometric interpretations, see Chapter 2 of \cite{smith_1969information}).
We show that $d_{Le}$ is indeed a metric as $d(r,s)=d(s,r)$ for all $r,s\in\R^2$, and $d((x,F(x)) , (x',G(x')))=0$ for a maximizing $x,x'$ iff $F(x)=G(x)$ for all $x$, and satisfies the triangle inequality
\begin{align*}
    d_{Le}(F,G)+d_{Le}(G,H) &= \frac{1}{\sqrt{2}} \sup_{x,x'} d\Big( (x,F(x)),(x'G(x')) \Big) + \frac{1}{\sqrt{2}} \sup_{z,z'} d\Big( (z,G(z)),(z',H(z')) \Big) \\
    &\geq \frac{1}{\sqrt{2}} \sup_{x,x',z,z'} \Big( d\Big( (x,F(x)),(x'G(x')) \Big) + d\Big( (z,G(z)),(z',H(z')) \Big) \Big) \\
    &\geq \frac{1}{\sqrt{2}} \sup_{x,z'} d\Big( (x,F(x)),(z',H(z')) \Big) = d_{Le}(F,H) 
\end{align*}
% We show that $d_{Le}$ is indeed a metric as $d(F,G,a)=d(G,F,a)$ for all $a$, and $d(F,G,a^*)=0$ for a maximizing $a^*$ iff $F(x)=G(x)$ for all $x$, and satisfies the triangle inequality
% \begin{align*}
%     d_{Le}(F,G)+d_{Le}(G,H) &= \frac{1}{\sqrt{2}} \sup_a d(F,G,a) + \frac{1}{\sqrt{2}} \sup_{a'} d(G,H,a') \\
%     &\geq \frac{1}{\sqrt{2}} \sup_a \Big( d(F,G,a) + d(G,H,a) \Big) \\
%     &\geq \frac{1}{\sqrt{2}} \sup_a d(F,H,a) = d_{Le}(F,H) 
% \end{align*}
where the maximization in the second line is over $x,x',z,z'$ such that $F(x)+x=G(x')+x'=G(z)+z=H(z')+z'$.
The last line is obtained by simply noting that $x'=z$ must be true since $G(x')+x'=G(z)+z$  and for any given three points $c_F=(x,F(x)), c_G=(x',G(x')), c_H=(z',H(z'))$ in $\R^2$ on a line, distance between $c_F$ and $c_H$ simply cannot exceed the sum of distances between $c_F$ and $c_G$ and between $c_G$ and $c_H$.
% for any given three points $a_F,a_G,a_H$ on the line $a=x+y$, distance between $a_F$ and $a_H$ simply cannot exceed the sum of distances between $a_F$ and $a_G$ and between $a_G$ and $a_H$.

\subsection{Proof of proposition \ref{prop_2}} \label{prop_2_proof}

We give a proof of Proposition \ref{prop_2} analogous to the optimization argument given by Smith \cite{Smith1971} about the optimality of AWGN channel encoding.
First, we show that the thermal channel Holevo information $\chi_{tc}$ is concave (Lemma~\ref{lem:holevo_info_concave}) and weakly differentiable (Lemma~\ref{lem_holevo_derivative}) and continuous (Lemma~\ref{lem:holevo_info_continuous}) over the space of circularly symmetric encoding $F$ equipped with the Levy metric (for definition, see Appendix~\ref{app:levy_metric}), then argue that $\mathscr{F}$ is convex and compact.
If $\chi_{tc}$ and $\mathscr{F}$ satisfy these properties, we may use the optimization theorem to obtain conditions for an optimal circularly symmetric encoding.

\begin{lemma}[{Optimization theorem}]\label{lem:optimization_theorem}
    Consider mapping $f$ from compact and convex topological space $\Omega$ to real numbers $\mathbb{R}$.
    If $f$ is continuous, concave, and weakly differentiable, then there exists a $x^*\in\Omega$ such that $f(x^*)=\sup_{x\in\Omega} f(x)$, and it is necessary and sufficient for such $x^*$ that the weak derivative of $f$ at $x^*$ in the direction of $x$, defined by
    \begin{align*}
        f_{x^*}'(x) = \lim_{\epsilon\rightarrow0} \frac{f(x^*+\epsilon(x-x^*)) - f(x^*)}{\epsilon}
    \end{align*}
    for $\epsilon>0$, is non-positive for all $x\in\Omega$.
    Moreover if $f$ is strictly concave, then such $x^*$ is unique.
\end{lemma}
A proof of this optimization theorem can be found in \cite[p.15]{smith_1969information}.
After establishing the conditions required to use Lemma~\ref{lem:optimization_theorem}, we will later conclude the proof by showing that
\begin{subequations}
    \begin{equation}\label{eq:9}
        i[\eta,F^*] \leq \chi_{tc}[F^*] \quad\textup{for all $\eta\in[0,1]$}
    \end{equation}
    \begin{equation}\label{eq:9.2}
        i[\eta,F^*] = \chi_{tc}[F^*] \quad\textup{if and only if $\eta\in \mathcal{I}$}
    \end{equation}
\end{subequations}
hold for encoding $F^*$ if and only if $F^*$ has non-positive weak derivative in all directions, thus \eqref{eq:9} and \eqref{eq:9.2} hold if and only if $\chi_{tc}[F^*]=\sup_F\chi_{tc}[F]$.

Note that all encodings $F$ in $\mathscr{F}$ distribute the channel phase uniformly and independently of the distribution over the channel attenuation, therefore optimization of $\chi_{tc}[F]$ is only over encodings over the attenuation, i.e. space of circularly symmetric encodings $\mathscr{F}$ is the space of cumulative distribution functions over the interval $[0,1]$.
Space $\mathscr{F}$ is convex because an encoding $F_\lambda = \lambda F_1 + (1-\lambda)F_2$ is clearly in $\mathscr{F}$ for any $\lambda\in[0,1]$ and $F_1,F_2,\in\mathscr{F}$.
The compactness of $\mathscr{F}$ follows from the compactness proof of the space of cumulative distribution functions over interval $[-a,a]$ for some real $a$ in the Levy metric by Smith~\cite[pp.21-25]{smith_1969information} where he shows that this space is totally bounded, from which compactness follows.
Smith's proof for total boundedness of encodings over $[-a,a]$ applies to space of encodings over any bounded and connected interval of $\mathbb{R}$, thus works for $\mathscr{F}$.

\begin{lemma}\label{lem:holevo_info_concave}
    Thermal channel Holevo information $\chi_{tc}$ is a strictly concave function of encodings in $\mathscr{F}$.
\end{lemma}

\begin{proof}
    For $\lambda\in[0,1]$ and an arbitrary pair of encodings $F_1$ and $F_2$, let $F_\lambda = \lambda F_1 + (1-\lambda) F_2$ for $\lambda\in(0,1)$.
    We will show that $\chi_{tc}[F_\lambda] \geq \lambda\chi_{tc}[F_1] + (1-\lambda)\chi_{tc}[F_2]$ with equality iff $F_1=F_2$.
    By the definition of $\chi_{tc}$ of $F_\lambda$,
    \begin{align}\label{prop2proof_eq1}
        \chi_{tc}[F_\lambda] = S\left(\int dF_\lambda(\eta,\theta) \rho(\eta,\theta)\right) - \int dF_\lambda(\eta,\theta) S(\rho(\eta,\theta)).
    \end{align}
    Now let $\rho_\text{ave$j$} = \int dF_j(\eta,\theta) \rho(\eta,\theta)$ for $j\in\{1,2\}$. 
    Because for any vector $|k\rangle$ from basis $\{|k\rangle\}_k$ the integral $\int dF \langle k|\rho|k\rangle$ is finite for any $F$, so we have
    \begin{align}\label{prop2proof_eq2}
        \int dF_\lambda(\eta,\theta) \rho(\eta,\theta) = \lambda\rho_\text{ave$1$} + (1-\lambda) \rho_\text{ave$2$} \;.
    \end{align}
    Hence by the strict concavity of the von Neumann entropy, we obtain
    \begin{align}\label{prop2proof_eq3}
        S\left(\int dF_\lambda(\eta,\theta) \rho(\eta,\theta)\right) &\geq \lambda S(\rho_\text{ave$1$}) + (1-\lambda) S(\rho_\text{ave$2$}) \;,
    \end{align}
    where equality holds iff $\rho_\textup{ave$1$} = \rho_\textup{ave$2$}$, which happens iff $F_1=F_2$.
    
    As for the second term of $\chi_{tc}[F_\lambda]$, since the entropy $S(\rho(\eta,\theta))$ is integrable with respect to any encoding, we have
    \begin{align}\label{prop2proof_eq4}
        \int dF_\lambda(\eta,\theta) S(\rho(\eta,\theta)) = \lambda \int dF_1(\eta,\theta) S(\rho(\eta,\theta)) + (1-\lambda) \int dF_2(\eta,\theta) S(\rho(\eta,\theta)).
    \end{align}
    By applying the equations \eqref{prop2proof_eq3} and \eqref{prop2proof_eq4} to \eqref{prop2proof_eq1}, the required concavity condition $\chi_{tc}[F_\lambda] \geq \lambda\chi_{tc}[F_1] + (1-\lambda)\chi_{tc}[F_2]$ is therefore shown.
\end{proof}

% \begin{itemize}
%     \item \url{https://math.stackexchange.com/questions/1236404/integrating-with-respect-to-a-linear-combination-of-two-signed-measures}
    
%     \item Cappinski, Kopp thm.4.29
% \end{itemize}

Now we show the weak-differentiability of $\chi_{tc}$.

\begin{lemma}[Weak differentiability of $\chi_{tc}$]\label{lem_holevo_derivative}
    Thermal channel Holevo information $\chi_{tc}$ is weakly differentiable.
    Particularly, given any pair of encodings $F_0$ and $F$, for weak derivative of $\chi_{tc}$ at $F_0$ in the direction of $F$
    \begin{align}
      \label{eq:25}
      \chi'_{F_0}[F] \coloneqq \lim_{\epsilon \to0}
      \frac{\chi_{tc}[F_0+\epsilon(F-F_0)] -  \chi_{tc}[F_0]}{\epsilon}
    \end{align}
    for $\epsilon>0$, we have
    \begin{align}
      \label{eq:26}
      \chi'_{F_0}[F] = \left(\int dF(\eta)\, i[\eta,F_0]\right) - \chi_{tc}[F_0]\;.
    \end{align}
\end{lemma}

\begin{proof}
    First, let $F_\epsilon = (1-\epsilon)F_0 + \epsilon F$ so $dF_\epsilon = (1-\epsilon)dF_0 + \epsilon dF$. 
    Then we have
    \begin{align}
        \chi'_{F_0}[F] &= \lim_{\epsilon\rightarrow0} \frac{1}{\epsilon} \left( \chi_{tc}[F_\epsilon] - \chi_{tc}[F_0] \right) \\
        &= \lim_{\epsilon\rightarrow0} \frac{1}{\epsilon} \left( (1-\epsilon) \int dF_0(\eta) \;i[\eta,F_\epsilon] + \epsilon\int dF(\eta) \;i[\eta,F_\epsilon] - \chi_{tc}[F_0]  \right) \\
        &= \int dF(\eta) i[\eta,F_0] - \int dF_0(\eta) i[\eta,F_0] + \lim_{\epsilon\rightarrow0} \frac{1}{\epsilon} \left( \int dF_0(\eta) i[\eta,F_\epsilon] - \chi_{tc}[F_0] \right) \;,
    \end{align}
    where we used the expression of the Holevo information in terms of the marginal information density
    \begin{align}
        \chi_{tc}[F] = \int dF(\eta) i[\eta,F] \label{eq:18}
    \end{align}
    as in \eqref{eq:18_main} to obtain the second equality and note that $\lim_{\epsilon\rightarrow 0} i[\eta,F_\epsilon]=i[\eta,F_0]$ in the last line.
    Noting that $\int dF_0(\eta) i[\eta,F_0] = \chi_{tc}[F_0]$ by \eqref{eq:18}, it remains for us to show that the limit term equals to $0$.
    We do this by again using \eqref{eq:18} as follows
    \begin{align}
        \lim_{\epsilon\rightarrow0} \frac{1}{\epsilon} \left( \int dF_0(\eta) i[\eta,F_\epsilon] - \chi_{tc}[F_0] \right) &= \lim_{\epsilon\rightarrow0} \frac{1}{\epsilon} \int dF_0(\eta) \left( i[\eta,F_\epsilon] - i[\eta,F_0] \right) \\
        &= \lim_{\epsilon\rightarrow0} \frac{1}{\epsilon} \int dF_0(\eta) \left( \sum_n P_n[\rho(\eta)] \log\frac{P_n[F_0]}{P_n[F_\epsilon]} \right) \\
        &= \lim_{\epsilon\rightarrow0} \frac{1}{\epsilon} \left( \sum_n P_n[F_0] \log\frac{P_n[F_0]}{P_n[F_\epsilon]} \right) \;.
    \end{align}
    % We can write each term in the sum as $P_n[F_0] \log \frac{P_n[F_0]}{(1-\epsilon)P_n[F_0] + \epsilon P_n[F]}$ so $\lim_{\epsilon\rightarrow0} \frac{P_n[F_0]}{\epsilon} \log \frac{P_n[F_0]}{(1-\epsilon)P_n[F_0] + \epsilon P_n[F]}$.
    By the virtue of the L'H\^opital's rule, we evaluate the limit of each terms in the sum to obtain
    \begin{align}
        \lim_{\epsilon\rightarrow0} \frac{P_n[F_0]}{\epsilon} \log\frac{P_n[F_0]}{P_n[F_\epsilon]} &= \lim_{\epsilon\rightarrow0} \frac{1}{\frac{d}{d\epsilon}\epsilon} \left( \frac{d}{d\epsilon} P_n[F_0]\log\frac{P_n[F_0]}{P_n[F_\epsilon]} \right) \\
        &= P_n[F_0] \lim_{\epsilon\rightarrow0} \frac{P_n[F]-P_n[F_0]}{((1-\epsilon)P_n[F_0] + \epsilon P_n[F])\ln2} \\
        &= \frac{P_n[F]-P_n[F_0]}{\ln2} \;.
    \end{align}
    Substituting this back into the sum gives us
    \begin{align}
        \lim_{\epsilon\rightarrow0} \frac{1}{\epsilon} \left( \int dF_0(\eta) i[\eta,F_\epsilon] - \chi_{tc}[F_0] \right) = \frac{1}{\ln2} \left(\sum_n P_n[F] - \sum_m P_m[F_0]\right) = 0 \;,
    \end{align}
    concluding the proof for \eqref{eq:26}.
\end{proof}

Here we will give a lemma to show weak continuity of $\rho(\eta)$, i.e. that $\<\varphi'|\rho(\eta)|\varphi\>$ is a continuous function over $\eta$ for any $|\varphi'\>$ and $|\varphi\>$, which we will need later in showing that Holevo information $\chi_\mathrm{tc}$ is continuous over the space of circularly symmetric encodings equipped with the Levy metric (for definition, see Appendix~\ref{app:levy_metric}).

\begin{lemma}\label{lem:ring_state_continuous}
    Given a resource state $\rho$ with energy $E<\infty$, its ring state $\rho(\eta)$ is weakly continuous over $\eta\in[0,1]$; namely for any pair of states $|\varphi'\>$ and $|\varphi\>$, $\<\varphi'|\rho(\eta)|\varphi\>$ is continuous.
\end{lemma}

\begin{proof}
    We will show that $\lim_{\eta'\rightarrow\eta} ||\rho(\eta) - \rho(\eta')||_1 = 0$ where $||\cdot||_1$ is the trace norm, which implies that $\<\varphi'|\rho(\eta)|\varphi\>$ is a continuous function over $\eta$ for any pair of states $|\varphi'\>$ and $|\varphi\>$.
    First, we 
    % write explicitly the ring state as $\rho(\eta) = \int \frac{d\theta}{2\pi} \mathcal{T}_{\eta,\theta}^{n_\mathrm{env}}(\rho)$ and 
    consider amplifier channel $\mathcal{A}_G$ with gain $G=(1-\eta)n_\mathrm{env}+1$ and attenuator channel $\mathcal{L}_{\tilde\eta}$ with attenuation $\tilde\eta = \frac{\eta}{G}$, defined by unitaries $U$ ad $V$ over the input system and environment system as 
    \begin{align*}
        \mathcal{A}_G(\rho) = \mathrm{tr}_E(U \rho\otimes\rho_\env U^\dag)
        \quad\textup{and}\quad
        \mathcal{L}_{\tilde{\eta}}(\rho) = \mathrm{tr}_E(V \rho\otimes\rho_\env V^\dag) \;,
    \end{align*}
    respectively.
    Unitaries $U$ and $V$ are defined by their actions (in the Heisenberg picture) on the annihilation operator $\hat{a}$ of the input system \cite{holevo_werner_2001evaluating,caruso_etal_2006one}
    \begin{gather*}
        U^\dag (\hat{a}\otimes I_\env) U = \sqrt{G}(\hat{a}\otimes I_\env) + \sqrt{G-1}(I\otimes\hat{a}_\mathrm{env}^\dag) \\
        \quad\textup{and}\quad \\
        V^\dag (\hat{a}\otimes I_\env) V = \sqrt{\tilde\eta} (\hat{a}\otimes I_\env) + \sqrt{1-\tilde\eta}(I\otimes\hat{a}_\mathrm{env})
    \end{gather*}
    where $\hat{a}_\env$ and $I_\env$ are the annihilation operator and identity operator of the environment, respectively.
    Similarly for attenuation $\eta'$ we define $G'=(1-\eta')n_\mathrm{env}+1$ and $\tilde\eta' = \frac{\eta'}{G'}$.
    Using the one-mode channel decomposition result in \cite{caruso_etal_2006one}, we may write a channel $\mathcal{N}_\eta^{(n_\env)}$ that mixes an input state with attenuation $\eta$ with a thermal environment with mean photon number $n_\env$ as a composition of an attenuation and an amplifier channel
    \begin{gather*}
        \mathcal{N}_\eta^{(n_\env)} = \mathcal{A}_G \circ \mathcal{L}_{\tilde{\eta}} \;.
    \end{gather*}
    Therefore we may write the thermal channel output state $\rho(\eta) = \mathcal{A}_G \circ \mathcal{L}_{\tilde\eta} (\rho)$ (similarly with attenuation $\eta'$), so we get
    \begin{align*}
        ||(\mathcal{A}_G\circ\mathcal{L}_{\tilde\eta} - \mathcal{A}_{G'}\circ\mathcal{L}_{\tilde\eta'})(\rho)||_1 &= ||(\mathcal{A}_G\circ\mathcal{L}_{\tilde\eta} - \mathcal{A}_{G'}\circ\mathcal{L}_{\tilde\eta} + \mathcal{A}_{G'}\circ\mathcal{L}_{\tilde\eta} - \mathcal{A}_{G'}\circ\mathcal{L}_{\tilde\eta'})(\rho)||_1 \\
        &\stackrel{(1)}{\leq} ||(\mathcal{A}_G\circ\mathcal{L}_{\tilde\eta} - \mathcal{A}_{G'}\circ\mathcal{L}_{\tilde\eta})(\rho)||_1 + ||(\mathcal{A}_{G'}\circ\mathcal{L}_{\tilde\eta} - \mathcal{A}_{G'}\circ\mathcal{L}_{\tilde\eta'})(\rho)||_1 \\
        &= ||(\mathcal{A}_G - \mathcal{A}_{G'})\circ\mathcal{L}_{\tilde\eta}(\rho)||_1 + ||\mathcal{A}_{G'}\circ(\mathcal{L}_{\tilde\eta} - \mathcal{L}_{\tilde\eta'})(\rho)||_1 \\
        &\stackrel{(2)}{\leq} ||(\mathcal{A}_G - \mathcal{A}_{G'})\circ\mathcal{L}_{\tilde\eta}(\rho)||_1 + ||(\mathcal{L}_{\tilde\eta} - \mathcal{L}_{\tilde\eta'})(\rho)||_1
    \end{align*}
    where we again used the triangle inequality and data processing inequality to obtain (1) and (2), respectively.
    Because $\rho$ has energy $E$, we can use the energy-constrained diamond norm \cite{shirokov_2018energy,becker_datta_2020convergence} on the amplifier channels and attenuator channels in the last line, defined as $||\mathcal{N}-\mathcal{M}||_{\diamond E} = \sup_\rho ||(\mathcal{N}-\mathcal{M}) \otimes\mathcal{I}_C(\rho)||_1$ where $\mathcal{I}_C$ is the identity channel for some ancilla system $C$ and the supremum is taken over all states $\rho$ on the product space of the input system and the ancilla system $C$ with energy constraint $E$ on the input system.
    Thus we obtain the following upper bound
    \begin{align*}
        % \Big|\Big| \int\frac{d\theta}{2\pi} \mathcal{T}_{\eta,\theta}^{n_\mathrm{env}}(\rho) - \mathcal{T}_{\eta',\theta}^{n_\mathrm{env}}(\rho) \Big|\Big|_1 
        ||(\mathcal{A}_G\circ\mathcal{L}_{\tilde\eta} - \mathcal{A}_{G'}\circ\mathcal{L}_{\tilde\eta'})(\rho)||_1 &\leq ||\mathcal{A}_G - \mathcal{A}_{G'}||_{\diamond E\tilde\eta} + ||\mathcal{L}_{\tilde\eta} - \mathcal{L}_{\tilde\eta'}||_{\diamond E} \\
        &\leq 2\beta_{E\eta}(\mathcal{A}_G, \mathcal{A}_{G'}) + 2\beta_E(\mathcal{L}_{\tilde\eta}, \mathcal{L}_{\tilde\eta'})
    \end{align*}
    as attenuator channel $\mathcal{L}_{\tilde\eta}$ scales the energy constraint by $\tilde\eta$ for the first term, and for the second inequality we use \cite[Lemma 10.17]{holevo_2019quantum} to bound the diamond norms by energy-constrained Bures distance $\beta_E$ which can be expressed as \cite{nair_2018quantum,nair_etal_2022optimal}
    \begin{gather*}
        \beta_{E\eta}(\mathcal{A}_G, \mathcal{A}_{G'}) = \sqrt{2} \sqrt{ 1-(1-\{E\eta\})\nu^{\lfloor E\eta\rfloor} - \{E\eta\}\nu^{\lfloor E\eta\rfloor+1} } \\
        \beta_E(\mathcal{L}_{\tilde\eta}, \mathcal{L}_{\tilde\eta'}) = \sqrt{2} \sqrt{ 1-(1-\{E\})\mu^{\lfloor E\rfloor} - \{E\}\mu^{\lfloor E\rfloor+1} }
    \end{gather*}
    where $\lfloor x\rfloor$ and $\{x\}$ are the integer part and fractional part of $x$, respectively, and $\mu = \sqrt{\tilde\eta\tilde\eta'} + \sqrt{(1-\tilde\eta)(1-\tilde\eta')}$ and $\nu = (\sqrt{GG'} + \sqrt{(G-1)(1-G'-1)})^{-1}$.
    
    As $\eta'\rightarrow\eta$ we have $\tilde\eta'\rightarrow\tilde\eta$ and $G'\rightarrow G$ which implies $\mu\rightarrow1$ and $\nu\rightarrow1$.
    Hence both $\beta_{E\eta}(\mathcal{A}_G, \mathcal{A}_{G'})$ and $\beta_E(\mathcal{L}_{\tilde\eta}, \mathcal{L}_{\tilde\eta'})$ limit as $\eta'\rightarrow\eta$ is $0$, implying that $\lim_{\eta'\rightarrow\eta} ||\rho(\eta) - \rho(\eta')||_1 = 0$, and therefore as a consequence $\rho(\eta)$ is weakly continuous.
\end{proof}

\begin{lemma}\label{lem:holevo_info_continuous}
    Given resource state $\rho$ with energy $E<\infty$, the thermal channel Holevo information $\chi_{tc}$ is continuous over space of circularly symmetric encodings $\mathscr{F}$ equipped with the Levy metric $d_{Le}$.
\end{lemma}

We emphasize that this continuity condition partially rests on our assumption that the resource state has a finite energy $E_\mathrm{max}$ as indicated in the beginning of Section~\ref{sec_thermal_operation}.
Here we use the continuity of von Neumann entropy $S$ over $\mathscr{G}_E$, which is the set of states with energy at most $E_\mathrm{max}$, shown in \cite[Lemma 11.8]{holevo_2019quantum} to show the continuity of $\chi_{tc}$ over $\mathscr{F}$ in 4 steps:
\begin{enumerate}
    \item Convergence of a sequence of encodings $\{F_n\}_{n\in\N} \subseteq\mathscr{F}$ in the Levy metric to some $F\in\mathscr{F}$ implies that the corresponding sequence of average states $\{\rho_{F_n}\}_{n\in\N} \subseteq\mathscr{G}_E$ weakly converges to $\rho_F \in\mathscr{G}_E$.
    
    \item Weak convergence (for detailed discussions see \cite[Section 11.1]{holevo_2019quantum}) of $\{\rho_{F_n}\}_{n\in\N}$ to some $\rho_F$ implies that $\lim_{n\rightarrow\infty} S(\rho_{F_n}) = S(\rho_F)$ because $S$ is continuous (by \cite[Lemma 11.8]{holevo_2019quantum}).
    Therefore, $S(\rho_F)$ is continuous over $F\in\mathscr{F}$.
    
    \item The second term of $\chi_{tc}$ , $\int dF'(\eta) \rho(\eta)$ is continuous over $F'\in\mathscr{F}$, i.e.
    \begin{align*}
        \lim_{n\rightarrow\infty} \int dF_n(\eta) S(\rho(\eta)) = \int dF(\eta) S(\rho(\eta))
    \end{align*}
    if $\lim_{n\rightarrow\infty} F_n = F$ in the Levy metric.
    
    \item Putting the previous two steps together, the Holevo information
    \begin{align*}
        \chi_{tc}[F] = S(\rho_F) - \int dF(\eta) S(\rho(\eta))
    \end{align*}
    is continuous over $F\in\mathscr{F}$.
\end{enumerate}

\begin{proof}
    First we show that if a sequence of encodings $\{F_n\}_{n\in\N} \subseteq\mathscr{F}$ converging to some $F\in\mathscr{F}$ under the Levy metric, then sequence of states $\{\rho_{F_n}\}_{n\in\N}$ is weakly converging to some $\rho_F$, where $\rho_{F'} = \int dF'(\eta) \rho(\eta)$ with $F'\in\mathscr{F}$.
    More precisely, for all encoding sequence $\{F_n\}_{n\in\N} \subseteq\mathscr{F}$ such that $\lim_{n\rightarrow\infty} d_{Le}(F_n,F) = 0$ for some $F\in\mathscr{F}$, we have $\lim_{n\rightarrow\infty} \<\varphi'|\rho_{F_n}|\varphi\> = \<\varphi'|\rho_F|\varphi\>$ for all states $|\varphi'\>,|\varphi\>$.
    First, we note that for any $F'\in\mathscr{F}$
    \begin{align*}
        \<\varphi'|\rho_{F'}|\varphi\> = \int dF'(\eta) \<\varphi'|\rho(\eta)|\varphi\> \;.
    \end{align*}
    Given that a convergence in the Levy metric implies a weak convergence of distributions (defined as $\lim_{n\rightarrow\infty} \int dF_n(x) f(x) = \int dF(x) f(x)$ for continuous and bounded $f$) in compact metric space $\mathscr{F}$ (hence is separable), and since $\<\varphi'|\rho(\eta)|\varphi\>$ is continuous (by Lemma~\ref{lem:ring_state_continuous}) and bounded (as $\rho(\eta)$ is a density matrix hence is a bounded operator) over $\eta\in[0,1]$ for all $|\varphi'\>,|\varphi\>$, we have
    \begin{align*}
        \lim_{n\rightarrow\infty} \<\varphi'|\rho_{F_n}|\varphi\> = \lim_{n\rightarrow\infty} \int dF_n(\eta) \<\varphi'|\rho(\eta)|\varphi\> = \int dF(\eta) \<\varphi'|\rho(\eta)|\varphi\> = \<\varphi'|\rho_F|\varphi\> \;.
    \end{align*}
    Therefore, Levy-metric convergence $\lim_{n\rightarrow\infty} d_{Le}(F_n,F) = 0$ implies weak convergence of states (defined as $\lim_{n\rightarrow\infty} \<\varphi'|\rho_{F_n}|\varphi\> = \<\varphi'|\rho_F|\varphi\>$).
    
    By \cite[Lemma 11.8]{holevo_2019quantum}, weak convergence of a sequence of bounded-energy states $\{\rho_n\}_{n\in\N} \subseteq\mathscr{G}_E$ to some $\rho\in\mathscr{G}_E$ implies the convergence of their von Neumann entropy.
    More precisely, $\lim_{n\rightarrow\infty} \<\varphi'|\rho_n|\varphi\> = \<\varphi'|\rho|\varphi\>$ implies 
    \begin{align*}
        \limsup_{n\rightarrow\infty} S(\rho_n) \leq S(\rho) \leq \liminf_{n\rightarrow\infty} S(\rho_n) \;.
    \end{align*}
    
    Now, in order to show that the second term of $\chi_{tc}$ is continuous over $\mathscr{F}$, we first note that $S(\rho(\eta))$ is continuous and bounded over $\eta\in[0,1]$ (as the thermal channel does not increase energy, hence \cite[Lemma 11.8]{holevo_2019quantum} applies) and that the convergence of a sequence of encodings $\{F_n\}_{n\in\N} \subseteq\mathscr{F}$ in the Levy metric to $F\in\mathscr{F}$ implies weak convergence of distributions (again, because $\mathscr{F}$ is a compact metric space).
    Therefore, $\lim_{n\rightarrow\infty} F_n = F$ in the Levy metric implies that
    \begin{align*}
        \lim_{n\rightarrow\infty} \int dF_n(\eta) S(\rho(\eta)) = \int dF(\eta) S(\rho(\eta)) \;.
    \end{align*}
    Hence the limit of a sequence of thermal channel Holevo information with sequence encodings $\{F_n\}_{n\in\N}$ converging to $F$ in the Levy metric is
    \begin{align*}
        \lim_{n\rightarrow\infty} \chi_{tc}[F_n] = \lim_{n\rightarrow\infty} S(\rho_{F_n}) - \int dF_n(\eta) S(\rho(\eta)) = S(\rho_F) - \int dF(\eta) S(\rho(\eta)) = \chi_{tc}[F] \;.
    \end{align*}
    Which implies that $\chi_{tc}$ is continuous over $\mathscr{F}$.
\end{proof}

Now we have shown that $\chi_{tc}$ is strictly concave, weakly differentiable, and continuous (because von Neumann entropy is continuous) over $\mathscr{F}$, also that the space of circularly symmetric encodings $\mathscr{F}$ is convex and compact.
Therefore, $\chi_{tc}$ and $\mathscr{F}$ satisfy all conditions of the optimization theorem (Lemma~\ref{lem:optimization_theorem}) to guarantee the existence of a unique optimal encoding $F^*$, i.e $\chi_{tc}[F^*] = \sup_F \chi_{tc}[F]$ if and only if $\chi'_{F^*}[F]\leq 0 $ for all $F$. 
In other words, the Holevo function is non-increasing in any direction $F$ at $F^*$.
We will use this fact to show that $F^*$ is optimal iff \eqref{eq:9} and \eqref{eq:9.2} holds for $F^*$.

First we will show that if \eqref{eq:9} and \eqref{eq:9.2} holds for encoding $F^*$, then $F^*$ must satisfy $\chi_{F^*}'[F]\leq 0$ for all $F$.
This can be shown by simply noting that if \eqref{eq:9} holds, then for any encoding $F$ 
\begin{align}
    \int dF(\eta) i[\eta,F^*] \leq \int dF(\eta) \chi_{tc}[F^*] = \chi_{tc}[F^*] \;.
\end{align}
Hence $\chi_{F^*}'[F]\leq0$ for all $F$ by lemma \ref{lem_holevo_derivative}.
Now recall the set of points of increase $\mathcal{I}$ (see definition~\ref{def:POI}). 
Then by noting that $\int_{\mathcal{I}} dF^*(\eta) = 1$ we have
\begin{align}
    \int dF^*(\eta) i[\eta,F^*] = \int_\mathcal{I} dF^*(\eta) i[\eta,F^*] + \int_{[0,\infty)\backslash \mathcal{I}} dF^*(\eta) i[\eta,F^*] = \int_\mathcal{I} dF^*(\eta) \chi_{tc}[F^*] = \chi_{tc}[F^*] \;,
\end{align}
showing that equality $\chi_{F^*}'[F^*]=0$ is achieved for encoding $F^*$.

For the converse, we will show that both \eqref{eq:9} and \eqref{eq:9.2} must hold for any optimal encoding $F^*$ by showing that if either \eqref{eq:9} or \eqref{eq:9.2} does not hold, we will get a contradiction.
Suppose there exist $\eta^*\in[0,1]$ such that $i[\eta^*,F^*] > \chi_{tc}[F^*]$.
By defining an encoding 
\begin{align}
    G(\eta) = \begin{cases}
    1,\; \text{for $\eta\geq\eta^*$}\\ 
    0,\; \text{for $\eta<\eta^*$}\end{cases},
\end{align}
we can write $i[\eta^*,F^*] = \int dG(\eta) i[\eta,F^*] > \chi[F^*]$ which contradicts the necessary and sufficient weak derivative condition for optimal $F^*$ that $\chi'_{F^*}[F]\leq 0$ for all $F$.

Similarly for \eqref{eq:9.2}, suppose that there exist some points of increase $C \subset \mathcal{I}$ such that $\int_C dF^*(\eta)=\delta>0$ and $i[\eta,F^*] < \chi[F^*]$ for all $\eta\in C$.
For the remaining points of increase $\eta\in \mathcal{I}\backslash C$, we have $\int_{\mathcal{I}\backslash C} dF^*(\eta) = 1-\delta$ and $i[\eta,F^*] = \chi_{tc}[F^*]$.
Because $\chi'_{F^*}(F^*)=0$ and $F^*([0,\infty) \backslash \mathcal{I})=0$, we have
\begin{align}
    \chi[F^*] = \int dF^*(\eta) i[\eta,F^*] = \int_{\mathcal{I}\backslash C} dF^*(\eta) i[\eta,F^*] + \int_C dF^*(\eta) i[\eta,F^*] < (1-\delta)\chi[F^*] + \delta \chi[F^*] = \chi[F^*]
\end{align}
giving us a contradiction $\chi[F^*]<\chi[F^*]$, hence both \eqref{eq:9} and \eqref{eq:9.2} must hold for any optimal encoding $F^*$. 
Therefore, concluding the proof of Proposition \ref{prop_2}.

\section{Proof of finite points-of-increase for coherent state resource at temperature $T=0$}\label{appendix:coherent_t0_finite_POI}

In this appendix, we give a proof of proposition~\ref{prop:finite_POI_coht0}.
This closely follows the proof for the finiteness of the set of points-of-increase for optimal photon channel by Shamai~\cite{Shamai1990}, which uses the Bolzano-Weierstrass theorem and the identity theorem of analytic function to show that infinite points-of-increase implies that the marginal information density $i[\lambda,F^*]$ is equal to the mutual information $I(F^*)$ for all photon channel intensity $\lambda\geq0$.
Shamai then shows that this leads to a contradiction.
We adopt the same reasoning to show that infinite points-of-increase implies that $i[\eta,F^*] = \chi[F^*]$ for all positive reals $\eta$, which leads to a contradiction.
Because $\rho(\eta,0)$ is a pure state, which implies that its entropy is $0$, the marginal information density is
\begin{align}
    i[\eta,F^*] = -\sum_n P_n[\rho(\eta)] \log P_n[F^*] 
    % = S(\rho(\eta)) + \sum_n P_n[\rho(\eta)] \log \frac{P_n[\rho(\eta)]}{P_n[F^*]}
    \;.
\end{align}
So, to establish the contradiction if the points-of-increase is assumed to has infinite cardinality we show that $i[\cdot,F^*]$ is analytic on positive reals.

\begin{lemma}
    Marginal information density $i[\eta,F^*]$ for coherent thermal state resource state in zero-temperature environment is analytic for all $\eta\in(0,\infty)$.
\end{lemma}
\begin{proof}
    Marginal information density in the case of direct detection channel
    \begin{align*}
        i_{dd}[x,F] = -\sum_n e^{-x}\frac{x^n}{n!} \log \frac{\int dF(\tilde{x}) e^{-\tilde{x}} \frac{\tilde{x}^n}{n!}}{e^{-x} \frac{x^n}{n!}}
    \end{align*}
    is analytic for all real $x>0$ as argued in \cite{Shamai1990}.
    Setting $x=\eta|\alpha|^2$ we can write
    \begin{align*}
        i[\eta,F^*] = i_{dd}[\eta|\alpha|^2,F^*] - \underbrace{ \sum_n P_n[\rho(\eta)] \log P_n[\rho(\eta)] }_{(1)}
    \end{align*}
    where $P_n[\rho(\eta)] = \textup{Poi}_n(\eta|\alpha|^2) = e^{-\eta|\alpha|^2} \frac{(\eta|\alpha|^2)^n}{n!}$ therefore term (1) is $-H(\textup{Poi}(\eta|\alpha|^2))$, the negative of Shannon entropy of Poisson random variable with mean $\eta|\alpha|^2$.
    It is shown in \cite{evans_boersma_1988entropy} that 
    \begin{align}\label{eqn:coherent_0_temp_prob_upper_bound}
    \begin{split}
        H(\textup{Poi}(\eta|\alpha|^2)) &= \eta|\alpha|^2(1-\log \eta|\alpha|^2) + \sum_n P_n[\rho(\eta)] \log n! \\
        &= \eta|\alpha|^2(1-\log \eta|\alpha|^2) + \mathbb{E}_{\mathrm{Poi(\eta|\alpha|^2)}}[\log N!]
    \end{split}
    \end{align}
    where $N$ is a Poisson random variable with mean $\eta|\alpha|^2$.
    The first term $\eta|\alpha|^2(1-\log \eta|\alpha|^2)$ is clearly analytic, therefore we only need to show that $\mathbb{E}_{\mathrm{Poi(\eta|\alpha|^2)}}[\log N!]$ is analytic for all $\eta>0$.
    
    Now for the rest of the proof we will show that the second term of \eqref{eqn:coherent_0_temp_prob_upper_bound} is analytic in the set $\C_\delta = \{z\in\C : \mathrm{Re}(z)>\delta \;\wedge\; |\mathrm{Im}(z)|<\delta\}$ for all $\delta>0$ by showing that the sum is uniformly convergent in any closed disk in $\C_\delta$.
    Its analyticity over $\R_{>0}$ then follows because it is analytic over all $\C_\delta$ for $\delta>0$ and because $\R_{>0} \subseteq \C_\delta$.
    
    Let us now extend $P_n[\rho(\eta)]$ to all $\eta\in\C_\delta$ for some small $\delta>0$, which is an open subset of $\C$ containing all real numbers in $(\delta,\infty)$.
    Now let $\eta = x+iy$ and consider
    \begin{align}
        |P_n[\rho(\eta)]| &= \Big| e^{-\eta|\alpha|^2} \frac{(\eta|\alpha|^2)^n}{n!} \Big| \\
        &= e^{-x|\alpha|^2} \frac{(\sqrt{x^2+y^2}|\alpha|^2)^n}{n!} \\
        &\leq e^{-x|\alpha|^2} \frac{(x|\alpha|^2)^n + (|y||\alpha|^2)^n}{n!} = P_n[\rho(x)] + e^{-x|\alpha|^2} \frac{(|y||\alpha|^2)^n}{n!} \label{eqn:coh_st_zero_temp_marg_info_complex_bound}\;.
    \end{align}
    So we can now upper bound the modulus of the expected value of $\log N!$ for the extended Poisson distribution function $P_n[\rho(\eta)]$
    \begin{align*}
        |\mathbb{E}_{\mathrm{Poi}(\eta|\alpha|^2)}[\log N!]| &= \Big| \sum_n P_n[\rho(\eta)]\log n! \Big| \\
        &\leq \sum_n |P_n[\rho(\eta)]| \log n! \\
        &\leq \mathbb{E}_{\mathrm{Poi}(x|\alpha|^2)}[\log N!] + e^{-x|\alpha|^2} \sum_n \frac{(|y||\alpha|^2)^n}{n!} \log n! \\
        &< 2 \mathbb{E}_{\mathrm{Poi}(x|\alpha|^2)}[\log N!]
    \end{align*}
    where the last inequality is because $|y|<\delta<x$ for any $\eta=x+iy\in\C_\delta$.
    
    Now
    % by Stirling's approximation
    we have $\log n! \leq \log n^n = n\log n \leq n(n-1)$, and therefore by considering a closed disk $D_r(z_0) = \{z\in\C_\delta : |z-z_0|\leq r\}$ with $z_0\in\C_\delta$ and $r>0$ and $\xi$ the maximum of the real part of any $z\in D_r(z_0)$ we get upper bound
    \begin{align*}
        \mathbb{E}_{\mathrm{Poi}(x|\alpha|^2)}[\log N!] &\leq \mathbb{E}_{\mathrm{Poi}(x|\alpha|^2)}[N\log N] \leq \mathbb{E}_{\mathrm{Poi}(x|\alpha|^2)}[N(N-1)] \leq (x|\alpha|^2)^2 \leq (\xi|\alpha|^2)^2 \;.
    \end{align*}
    
    Now we will show that $\mathbb{E}_{\mathrm{Poi}(\eta|\alpha|^2)}[\log N!]$ uniformly converges in any closed disk $D_r(z_0) = \{z\in\C_\delta : |z-z_0|\leq r\}$ with $z_0\in\C_\delta$ and $r>0$ such that $D_r(z_0)\subseteq\C_\delta$.
    Using a standard theorem of complex analysis (see e.g. \cite[Theorem 3.1.8]{marsden_etal_1999basic}) this implies that $\mathbb{E}_{\mathrm{Poi}(\eta|\alpha|^2)}[\log N!]$ is analytic in $\C_\delta$ because each of the term in the sum is analytic.
    
    To show uniform convergence of $\mathbb{E}_{\mathrm{Poi}(\eta|\alpha|^2)}[\log N!]$ in $D_r(z_0)$ we will use the Weierstrass M-test (see e.g. \cite[Theorem 3.1.7]{marsden_etal_1999basic}) which states that $f_n$ uniformly converges to $f$ in some set $A$ if there is some real numbers $\{M_n\}_n$ such that $\sum_n M_n <\infty$ and $|f_n(z)|\leq M_n$ for all $n$ and all $z\in A$.
    Let $\xi = \max\{\mathrm{Re}(z) : z\in D_r(z_0)\}$ (i.e. the largest real part of complex numbers in disk $D_r(z_0)$).
    By \eqref{eqn:coh_st_zero_temp_marg_info_complex_bound} we know that $|P_n[\rho(\eta)]|\log n! \leq 2P_n[\rho(x)]\log n!$ and for sufficiently large $N$ we have $P_n[\rho(x)] \leq P_n[\rho(\xi)]$ for all $n>N$.
    So, let 
    \begin{align*}
        M_n =
        \begin{cases}
        2P_n[\rho(\xi)]\log n! \quad&\textup{if $n>N$}\\
        (\xi|\alpha|^2)^2 \quad&\textup{if $n\leq N$}
        \end{cases} \;.
    \end{align*}
    Given that $\mathbb{E}_{\mathrm{Poi}(x|\alpha|^2)}[\log N!] \leq (\xi|\alpha|^2)^2$ implies that $\sum_{n>N} M_n \leq 2(\xi|\alpha|^2)^2$ and $\sum_{n\leq N} M_n = N(\xi|\alpha|^2)^2$ we have $\sum_n M_n \leq (N+1)(\xi|\alpha|^2)^2 < \infty$ with $|P_n[\rho(\eta)]| \leq M_n$ for all $n$.
    
    Hence, since the second term of \eqref{eqn:coherent_0_temp_prob_upper_bound} uniformly converges in any closed disk $D_r(z_0)\subseteq\C_\delta$ and each of its term $P_n[\rho(\eta)]\log n!$ is analytic in $\C$, it is analytic in $\C_\delta$.
    % to a function that at most grows quadratically in the energy of the output state $x|\alpha|^2$ for all $\eta \in \C_\delta$.
    % Therefore because each term in the sum of the second term \eqref{eqn:coherent_0_temp_prob_upper_bound}, $P_n[\rho(\eta)]\log n!$ is analytic in $\C$ and that the sum converges, it converges to a function $\mathbb{E}_{\mathrm{Poi}(\eta|\alpha|^2)}[\log N!]$ that is analytic on all reals in $(\delta,\infty)$ which is a subset of $\C_\delta$.
    Consequently the restriction of $i[\cdot,F^*]$ over $(\delta,\infty)$ is analytic for any $\delta>0$, concluding the proof.
    % see https://math.stackexchange.com/questions/3711263/expectation-of-x-log-x-for-x-sim-mathsfpoisson-lambda and https://stackoverflow.com/questions/8118221/what-is-ologn-on-and-stirlings-approximation and https://en.wikipedia.org/wiki/Stirling%27s_approximation
\end{proof}

% \begin{proof}
%     First, we note that $\log P_n[F^*] \in [0,1]$ is independent of $\eta$ and because a sum of analytic functions is analytic, it is sufficient to show that $P_n[\rho(\cdot)]$ is analytic on $\C$ for all $n\in\N$.
%     Recall that these probability functions are Poisson probability functions 
%     \begin{align}
%         P_n[\rho(\eta)] = \textup{Poi}_n(\eta|\alpha|^2) = \underbrace{ e^{-\eta|\alpha|^2} }_{(1)} \underbrace{ \frac{(\eta|\alpha|^2)^n}{n!} }_{(2)}
%     \end{align}
%     where we will show that (1) and (2) are analytic.
%     Term (1) is analytic because it is a composition of exponential and constant multiplication of $\eta$, which are analytic functions.
%     Term (2) is analytic because it is a composition of multiplication and exponentiation of $\eta$, which are analytic.
%     Because $P_n[\rho(\eta)]$ is analytic for all $n$, therefore $i[\cdot,F^*]$ is analytic.
% \end{proof}

Note that the transmissivity can only physically takes value in $[0,1]$, nevertheless because $i[\eta,F^*]$ is analytic on $(0,\infty)$, a contradiction will follow from assuming that there are infinite POIs.
First we consider a lower bound
\begin{align}
    i[\eta,F^*] &= -\sum_n P_n[\rho(\eta)] \log \Big(\int dF^*(\tilde\eta) e^{-\tilde\eta|\alpha|^2} \frac{(\tilde\eta|\alpha|^2)^n}{n!} \Big) \nonumber\\
    &\stackrel{(1)}{\geq} \sum_n e^{-\eta|\alpha|^2} \frac{(\eta|\alpha|^2)^n}{n!} \Big( n \log\frac{1}{|\alpha|^2} + \log n! \Big) \nonumber\\
    &\stackrel{(2)}{=} \eta|\alpha|^2 \log\frac{1}{|\alpha|^2} + \sum_n e^{-\eta|\alpha|^2} \frac{(\eta|\alpha|^2)^n}{n!} \log n! \;, \label{eqn:prop3_proof:lowerbound_marginal_info_1}
    % &\stackrel{(3)}{\geq} \eta|\alpha|^2 \log\frac{1}{|\alpha|^2} + 1 \label{eqn:prop3_proof:lowerbound_marginal_info_2}
\end{align}
where (1) is because $\log \int dF^*(\tilde\eta) e^{-\tilde\eta|\alpha|^2} \frac{(\tilde\eta|\alpha|^2)^n}{n!} \leq n \log |\alpha|^2 - \log n!$ as the integral is only over $[0,1]$, and (2) is because $\sum_n \frac{(\eta|\alpha|^2)^n}{n!} n = \eta|\alpha|^2 e^{\eta|\alpha|^2}$. 
% and (3) is because $\sum_n \frac{(\eta|\alpha|^2)^n}{n!} \log n! \geq e^{\eta|\alpha|^2}$.

% Now we split into cases when $|\alpha|^2 \neq 1$ and $|\alpha|^2 = 1$.
% For $|\alpha|^2 \neq 1$, consider the lower bound \eqref{eqn:prop3_proof:lowerbound_marginal_info_2}. 
% Because $\chi_{tc}[F^*]=i[\eta,F^*]$ for all POIs $\eta$, if $|\alpha|^2 >1$, we can take $\eta\rightarrow -\infty$ to get an arbitrarily large $\chi_{tc}$.
% Whereas in the case that $|\alpha|^2 \in(0,1)$, we can take $\eta\rightarrow \infty$ to get an arbitrarily large $\chi_{tc}$ because $\log\frac{1}{|\alpha|^2}>0$.
The second term in \eqref{eqn:prop3_proof:lowerbound_marginal_info_1} is an expectation of $\log N!$ where $N$ is a Poisson random variable with mean $\eta|\alpha|^2$
\begin{align}
    \mathbb{E}_{\mathrm{Poi(\eta|\alpha|^2)}}[\log N!] = \sum_n e^{-\eta|\alpha|^2} \frac{(\eta|\alpha|^2)^n}{n!} \log n! \;.
\end{align}
By Stirling's approximation for log of factorials we can use big omega notation to obtain lower bound $\log n! \geq \Omega(n \log n)$, so 
% \textcolor{red}{[we may be able to use $\log n! \geq (n/2)^{n/2}$]}
% \begin{align}
%     \mathbb{E}_{\mathrm{Poi(\eta)}}[\log N!] \geq \Omega\Big( \mathbb{E}_{\mathrm{Poi(\eta)}}[N\log N] \Big) \geq \Omega\Big( \mathbb{E}_{\mathrm{Poi(\eta)}}[N] \Big)
% \end{align}
\begin{align}
    \mathbb{E}_{\mathrm{Poi(\eta|\alpha|^2)}}[\log N!] &\geq \Omega\Big( \mathbb{E}_{\mathrm{Poi(\eta|\alpha|^2)}}[N\log N] \Big) \\
    &\geq \Omega\Big( \mathbb{E}_{\mathrm{Poi(\eta|\alpha|^2)}}[N] \log(\mathbb{E}_{\mathrm{Poi(\eta|\alpha|^2)}}[N]) \Big) = \Omega(\eta|\alpha|^2\log\eta|\alpha|^2)
\end{align}
% again, see https://math.stackexchange.com/questions/3711263/expectation-of-x-log-x-for-x-sim-mathsfpoisson-lambda
as the Poisson-distributed expectation $\mathbb{E}_{\mathrm{Poi(z)}}[N] = z$.
Hence applying this lower bound to \eqref{eqn:prop3_proof:lowerbound_marginal_info_1} gives
\begin{align*}
    i[\eta,F^*] \geq \eta|\alpha|^2 \log\frac{1}{|\alpha|^2} + \mathbb{E}_{\mathrm{Poi}(\eta|\alpha|^2)}[\log N!] \geq \Omega\Big( \eta|\alpha|^2 \log \eta \Big)
\end{align*}
which implies that $i[\eta,F^*]$ grows arbitrarily large as $\eta\rightarrow\infty$.
Hence for any resource energy $|\alpha|^2$ and any optimal thermal encoding $F^*$, an infinite points-of-increase leads to a contradiction as Holevo information $\chi_{tc}[F^*]=i[\eta,F^*]$ goes to infinity as $\eta\rightarrow\infty$, so the cardinality of the points-of-increase must be finite.

\section{Threshold when one ring is no longer optimal}\label{sec:alpStar}

Here for ring state at amplitude $\alpha_\mathrm{max}$, $\rho_\alpha = \int_{-\pi}^\pi \frac{d\theta}{2\pi} R_\theta \ketbra{\alpha}{\alpha} R_\theta^\dag$ and vacuum state $\rho_0=\ketbra{0}{0}$ we will solve for $\alpha$ that satisfies
\begin{align}
  \label{eq:22}
  \frac{\partial}{ \partial \epsilon} S(\epsilon \rho_0 +
  (1-\epsilon)\rho_\alpha) \Big|_{\epsilon=0} &=0
\end{align}  
which can be written more explicitly in terms of limit
\begin{align}
  \lim_{\epsilon \to0} \frac{S(\rho_\alpha + \epsilon(\rho_0-\rho_\alpha))-S(\rho_\alpha)}{\epsilon}&=0 \;.
\end{align}
Using the Taylor expansion of $S(A+\epsilon B)$ for two commuting matrices $A$ and $B$,
\begin{align}
  \label{eq:23}
  S(A+\epsilon B) &= \frac{1}{\ln 2}\text{Tr}\left[ (A+\epsilon B) \ln
                    (A+\epsilon B)\right]\\
  &= \frac{1}{\ln 2}\text{Tr}\left[ (A+\epsilon B) \left(\ln A +
    \frac{\epsilon B}{A}\right)\right]\\
  &=S(A) + \epsilon \text{Tr} B \log A + \frac{\epsilon}{\ln 2}\text{Tr}B\;,
\end{align}
we get
\begin{align}
  \label{eq:24}
  \text{Tr}\left[(\rho_0-\rho_\alpha)\log \rho_\alpha \right]
  +\frac{1}{\ln 2} \text{Tr}\left[\rho_0-\rho_\alpha \right] &=0\\
  \text{Tr}\left[\rho_0 \log \rho_\alpha\right] + S(\rho_\alpha) &=0\\
  \log e^{-|\alpha|^2}+S(\rho_\alpha)&=0\\
  \frac{|\alpha|^2}{\ln 2}&= S(\rho_\alpha)
\end{align}
as claimed. The solution to this is $|\alpha|=1.25034$.

\section{Bounds on the capacity for coherent state resource at $T=0$ and average energy constraint}\label{app:vacBounds}

A lot of work have been done to derive bounds for capacities of classical channels.
For example, see~\cite{Rassouli2016,McKellips2004,Thangaraj2017} for bounds for the
AWGN channels and~\cite{Thangaraj2017} for quadrature modulations and
measurements.  Lapidoth derives bounds for a discrete time Poissonian
channel~\cite{Lapidoth2009} and intensity modulation with additive
Gaussian noise~\cite{Lapidoth2009_2}. There are also bounds on the
number of rings for the $n$-dimensional AWGN channel by
Dytso~\cite{Dytso2019,Yagli2019}.
In this section, we will present the upper and lower bounds for the Holevo information $\chi$ when we have a coherent state $\ket{\alpha_\mathrm{max}}$ resource with environment temperature $T=0$.

\subsection{Upper bounds}
\label{app:upper-bound}
If we relax the channel by allowing all encoding $F$ such that the average energy (i.e: energy of $\rho_\mathrm{ave}$) to be bounded by $E_\mathrm{max}=|\alpha_\mathrm{max}|^2$ instead bounding the energy of individual codewords $\rho(\eta,\theta)$, we may use the result by Giovannetti, et al.~\cite{Giovannetti2014} to get an upper bound to Holevo information $\chi$.
For the average-energy constrained channel, $\chi$ is attained by a Gaussian distribution encoding and is given in the supplemental material of \cite{Giovannetti2014}
\begin{align}\label{eq:10}
    \chi_\text{GDE}(\alpha_\mathrm{max}) &= (E_\mathrm{max}+1)\log(E_\mathrm{max}+1) - E_\mathrm{max}\log E_\mathrm{max} \nonumber\\ 
    &= -E_\mathrm{max}\log_2\frac{E_\mathrm{max}}{1+E_\mathrm{max}}+\log_2(1+E_\mathrm{max}) \;.
\end{align}

\subsection{Lower bounds}
\label{sec:lower-bounds}

While the classical mutual information gives a lower bound to its quantum analogue $\chi$, this lower bound can be obtained by a single-ring encoding and an encoding correspond to a "flat" distribution.

\subsubsection{Single-ring encoding}\label{sec:lower-bounds_one_ring}
We can obtain a lower bound by considering a distribution with just
one ring at $\alpha_\mathrm{max}$. 
For $E_\mathrm{max}=|\alpha_\mathrm{max}|^2$, this is given by the von Neumann entropy of the ring state at amplitude $\alpha_\mathrm{max}$, which is a uniformly random phase-shifted mixture of $\ket{\alpha_\mathrm{max}}$
\begin{align}
  \label{eq:11}
    L_\text{1ring}(\alpha_\mathrm{max}) &= \chi[F_\mathrm{1ring}] \\
    &= S\left( \int_{-\pi}^\pi \frac{d\theta}{2\pi} R_\theta \ketbra{\alpha_\mathrm{max}}{\alpha_\mathrm{max}} R_\theta^\dag \right) \\
    % \text{entropy of phase randomised } \ket{\alpha_\mathrm{max}}\\
    &= -\sum_{n=0}^\infty  \frac{e^{-E_\mathrm{max}} E_\mathrm{max}^n}{n!} \log_2\frac{e^{-E_\mathrm{max}}E_\mathrm{max}^n}{n!} \;.
\end{align}

\subsubsection{Flat encoding}\label{app:flat_dist}

Consider an encoding that uniformly distributes the attenuation, i.e. $F(\eta)=\eta$. 
This distribution will (by construction) has a flat Wigner function in the middle and drops off at the edge (Figure~\ref{fig:bigAlpHat}). 
Curiously, the photon number distribution is also rather flat at low photons numbers when energy $E_\mathrm{max}$ is large enough (see Figure~\ref{fig:bigAlpHat}). 
This encoding has a Fock state distribution
\begin{align}
  \label{eq:12}
  P_n[F] &= \int dF(\eta) \,P_n[\alpha_\mathrm{max}(\eta)]\\
         &= \int d \eta  \,  e^{-\eta E_\mathrm{max}}
           \frac{(\eta E_\text{max} )^{n}}{n!}\\
         &=\frac{1}{E_\mathrm{max}}-\frac{1}{n!E_\mathrm{max}}\Gamma(1+n,E_\mathrm{max})\;.
\end{align}
The lower bound (which we denote as $L_\text{flat}$) can then be computed by taking the entropy of this distribution. 
This bound becomes more informative than $L_\text{1ring}$ as $E_\mathrm{max}$ increases.
These bounds and the numerically computed capacity are plotted in Figure~\ref{fig:1}(b).
The information per unit energy diverges as the input $E_\mathrm{max}$ tends to zero.

Now consider the Poisson distribution parameter of $P_n[\alpha_\mathrm{max}(\eta)]$, which is $\eta E_\mathrm{max}$ and denote a random variable $\mathcal{X}$ for such parameter.
% Let the photon count random variable $\mathcal{N}_\mathrm{ave}$ for coherent state input $|\alpha_\mathrm{max}\>$ is distributed according to Poisson distribution
% \begin{align}
%   \label{eq:14}
%   \text{Pr}(n|\eta) = e^{-\eta E_\mathrm{max}}\frac{(\eta E_\mathrm{max})^n}{n!}\;.
% \end{align}
Such random variables have been studied in the classical context of discrete-time Poisson channels where the communication channel consists of a coherent state intensity modulation followed by direct detection at the receiver~\cite{Lapidoth2009}.
Proposition 11 of~\cite{Lapidoth2009} gives the lower bound on the entropy of such channel
\begin{align}
    \label{eq:15}
    H \geq h[\mathcal{X}] + \left(1+\mathbb{E}[\mathcal{X}]\right) \log \left( 1+\frac{1}{\mathbb{E}[\mathcal{X}]}\right) -\frac{1}{\ln 2}\;,
\end{align}
where $h[\mathcal{X}]$ is the differential entropy of $\mathcal{X}$ and $\mathbb{E}[\mathcal{X}]$ is the mean of $\mathcal{X}$.
Using this result on the flat distribution $F(\eta)=\eta$ for $0\leq \eta \leq 1$, we could have a uniform encoding for random variable $\mathcal{X}$, $F(x)=x/E_\text{max}$ so that $h[\mathcal{X}]=\log E_\text{max}$ and $\mathbb{E}[\mathcal{X}]=E_\text{max}/2$. This gives a lower bound for the capacity $L_\mathrm{flat}$ for flat encoding
\begin{align}
  \label{eq:16}
  L_\mathrm{flat} \geq  \log E_\mathrm{max} + \left(1+\frac{E_\mathrm{max}}{2}
  \right)\log\left( 1+\frac{2}{E_\mathrm{max}}\right)- \frac{1}{ \ln
  2} \;.
\end{align}

\section{Proof of finite points-of-increase for thermal state resource}\label{appendix:finite_POI_thermal_res}

In this appendix we give a proof of proposition~\ref{prop:finite_POI_thermal}.
Here we employ similar argument to the case of coherent state input in Appendix \ref{appendix:coherent_t0_finite_POI} where we show that the marginal information density function for the thermal state is analytic for all positive reals then use the identity theorem of analytic functions along with the Bolzano-Weierstrass theorem to obtain $i[\eta,F^*] = \chi_{tc}[F^*]$ for all positive real $\eta$ from assuming that the points-of-increase is infinite.
We again show that this leads to a contradiction to conclude that the assumption of the cardinality of the set of points-of-increase being infinite cannot be true.
Here we assume that both mean photon number of the input state $n_\res$ and mean photon number of the environment $n_\mathrm{env}$ are strictly larger than zero.
We also assume that $n_\res \neq n_\mathrm{env}$ because otherwise we have $\rho(\eta)=\rho_\mathrm{ave}$ for all $\eta$ hence we cannot encode any information as reflected by $\chi_{tc}[F] = \int dF(\eta) S(\rho(\eta)||\rho_\mathrm{ave}) = 0$ for any choice of encoding $F$.
For simplicity, the above assumptions will be taken into account in the proofs below, but we will give an argument later that the same technique can still be used in the case that either $n_\res$ or $n_\mathrm{env}$ is zero (i.e. vacuum).

Noting that the $S(\rho(\eta,0))$ term in the marginal information density \eqref{eq:17} is the entropy of thermal state with mean photon number $n_\eta = \eta n_\res + (1-\eta) n_\mathrm{env}$, we can rewrite it as
\begin{align}\label{eqn:thermal_input:marginal_info_density1}
    i[\eta,F^*] = \sum_k P_k[\rho(\eta)] \log\frac{1}{P_k[F^*]} - \Big( (n_\eta+1)\log(n_\eta+1) - n_\eta\log n_\eta \Big)
\end{align}
where $P_k[F^*] = \int dF^*(\eta') P_k[\rho(\eta')]$ and
\begin{align}
    P_k[\rho(\eta)] = \frac{n_\eta^k}{(n_\eta+1)^{k+1}} \;.
\end{align}
Now we will show that $i[\cdot,F^*]$ is analytic over all positive reals.

\begin{lemma}\label{lem:thermal_input:marginal_info_density_analytic}
    % For thermal state input with mean photon number $n_\res$, 
    Marginal information density $i[\eta,F^*]$ for thermal state input is analytic for all $\eta\in(0,\infty)$.
\end{lemma}

Before proving Lemma~\ref{lem:thermal_input:marginal_info_density_analytic}, we first need the following technical lemma.

\begin{lemma}\label{lem:thermal_input:log_avg_proba_upper_bound}
    Let $Q_k(z) = \frac{z^k}{(z+1)^{k+1}}$ and $n_\mathrm{max} = \max\{n_\res,n_\mathrm{env}\}$ and $n_\mathrm{min} = \min\{n_\res,n_\mathrm{env}\}$.
    The following inequality holds for all circularly symmetric encoding $F$ and all $k>n_\mathrm{max}$
    \begin{align*}
        \log\frac{1}{P_k[F]} \leq \log\frac{1}{Q_k(n_\mathrm{min})} = \log(n_\mathrm{min}+1) + k \log\frac{n_\mathrm{min}+1}{n_\mathrm{min}} \;.
    \end{align*}
\end{lemma}

\begin{proof}
    We can write $P_k[F] = \int_0^1 dF(\eta) Q_k(n_\eta)$.
    Note that the derivative of $Q_k$ is
    \begin{align*}
        \frac{d}{dz} Q_k(z) = (k-z) z^{k-1} (z+1)^{-k-2} \;.
    \end{align*}
    So by solving $\frac{d}{dz} Q_k(z) = 0$ for $z$ we can find that $Q_k$ achieves its maximum at $z=k$.
    Moreover, one can observe from the $(k-z)$ term that $Q_k$ is monotonically increasing for $z<k$ and monotonically decreasing for $z>k$ (i.e. $\forall z'<k, z'>z \Rightarrow Q_k(z') > Q_k(z)$ and $\forall z>k, z'>z \Rightarrow Q_k(z') < Q_k(z)$ ).
    Therefore
    \begin{align*}
        P_k[F] = \int_0^1 dF(\eta) Q_k(n_\eta)
        % \geq \min_{\eta'\in[0,1]} P_k[\rho(\eta')]
        \geq Q_k(n_\mathrm{min})
    \end{align*}
    if $k > n_\mathrm{max}$.
    This implies that for all $k > n_\mathrm{max}$ we have $\frac{1}{P_k[F]} \leq \frac{1}{Q_k(n_\mathrm{min})}$, and thus
    \begin{align*}
        \log\frac{1}{P_k[F]} \leq \log\frac{1}{Q_k(n_\mathrm{min})} = \log\frac{(n_\mathrm{min}+1)^{k+1}}{n_\mathrm{min}^k}
        = \log(n_\mathrm{min}+1) + k \log\frac{n_\mathrm{min}+1}{n_\mathrm{min}}
    \end{align*}
    because $\log$ is a monotonically increasing function. 
\end{proof}

Now we are set to show the analiticity of the marginal information density.

\begin{proof}[proof of Lemma~\ref{lem:thermal_input:marginal_info_density_analytic}]
    First, let $h[\eta,F^*] = \sum_{k=0}^\infty P_k[\rho(\eta)] \log\frac{1}{P_k[F^*]}$ so that
    \begin{align}
        i[\eta,F^*] = h[\eta,F^*] - \Big( (n_\eta+1)\log(n_\eta+1) - n_\eta\log n_\eta \Big) \;.
    \end{align}
    The second term is clearly analytic as it is a finitely many compositions of log functions, multiplications, and additions over $z\in\R_{>0}$, hence it remains to show that $h[\cdot,F^*]$ is analytic over $\R_{>0}$.
    Again by considering $\C_\delta = \{z\in\C : \mathrm{Re}(z)>\delta \;\wedge\; |\mathrm{Im}(z)|<\delta\}$, we will show that $h[\cdot,F^*]$ is analytic over $\C_\delta$ for any $\delta>0$ by showing that $h[\cdot,F^*]$ converges uniformly in any closed disk $D_r(z_0) = \{z\in\C_\delta : |z-z_0|\leq r\} \subseteq\C_\delta$.
    Denoting $h_N[\eta,F^*] = \sum_{k>N} P_k[\rho(\eta)] \log\frac{1}{P_k[F^*]}$, it suffices to show that for any $\varepsilon>0$ there exists $M\in\N$ such that for all $N\geq M$, $|h_N[\eta,F^*]|<\varepsilon$ for all $\eta\in D_r(z_0)$ because the partial sum $h[\eta,F^*]-h_N[\eta,F^*]$ is a composition of additions, multiplications, and log functions, hence is analytic.
    
    By using Lemma~\ref{lem:thermal_input:log_avg_proba_upper_bound} we have an upper bound
    \begin{align}\label{eqn:thermal_input:hnmax_bound}
    \begin{split}
        |h_{n_\mathrm{max}}[\eta,F^*]| &= \Big| \sum_{k>n_\mathrm{max}} P_k[\rho(\eta)] \log\frac{1}{P_k[F^*]} \Big| \\
        &\leq \sum_{k>n_\mathrm{max}} |P_k[\rho(\eta)]| \Big( \log(n_\mathrm{min}+1) + k\log\frac{n_\mathrm{min}+1}{n_\mathrm{min}} \Big) \\
        % &= \Big( \log(n_\mathrm{min}+1) \sum_{j>n_\mathrm{max}} |P_j[\rho(\eta)]| \Big) + \Big( \log\frac{n_\mathrm{min}+1}{n_\mathrm{min}} \sum_{k>n_\mathrm{max}} |P_k[\rho(\eta)]| \; k \Big) \\
        % &\leq \log(n_\mathrm{min}+1) + \log\frac{n_\mathrm{min}+1}{n_\mathrm{min}} \sum_{k>n_\mathrm{max}} |P_k[\rho(\eta)]| \; k \\
        &= \frac{1}{|n_\eta+1|} \sum_{k>n_\mathrm{max}} \frac{|n_\eta|^k}{|n_\eta+1|^k} \Big( \log(n_\mathrm{min}+1) + k\log\frac{n_\mathrm{min}+1}{n_\mathrm{min}} \Big) \\
        &= \frac{\log(n_\mathrm{min}+1)}{|n_\eta+1|} \sum_{j>n_\mathrm{max}} \frac{|n_\eta|^j}{|n_\eta+1|^j} + \frac{1}{|n_\eta+1|} \log\frac{n_\mathrm{min}+1}{n_\mathrm{min}} \sum_{k>n_\mathrm{max}} \frac{|n_\eta|^k}{|n_\eta+1|^k} \; k \;.
    \end{split}
    \end{align}
    From the series identity $\sum_{k=0}^\infty k c^k = \frac{c}{(c-1)^2}$ we have
    \begin{align*}
        \sum_{k=0}^\infty \frac{|n_\eta|^k}{|n_\eta+1|^k} \; k = \frac{|n_\eta||n_\eta+1|}{(|n_\eta|-|n_\eta+1|)^2} \;,
    \end{align*}
    and for any $N\in\N$, the partial series identity $\sum_{k=0}^N k c^k = \frac{(cN-N-1)c^{N+1}+c}{(1-c)^2}$ gives
    \begin{align*}
        \sum_{k=0}^N \frac{|n_\eta|^k}{|n_\eta+1|^k} \; k &= \frac{|n_\eta+1|^2}{(|n_\eta+1|-|n_\eta|)^2} \bigg( \Big( \frac{|n_\eta|}{|n_\eta+1|}N-N-1 \Big) \Big(\frac{|n_\eta|}{|n_\eta+1|}\Big)^{N+1} + \frac{|n_\eta|}{|n_\eta+1|} \bigg) \\
        &= \frac{|n_\eta+1||n_\eta|}{(|n_\eta+1|-|n_\eta|)^2} \Big(\frac{|n_\eta|}{|n_\eta+1|}\Big)^N \Big(\frac{|n_\eta|}{|n_\eta+1|}N-N-1\Big) + \frac{|n_\eta+1||n_\eta|}{(|n_\eta+1|-|n_\eta|)^2} \;.
    \end{align*}
    Hence we get
    \begin{align}
        \sum_{k>N} \frac{|n_\eta|^k}{|n_\eta+1|^k} \; k = \frac{|n_\eta+1||n_\eta|}{(|n_\eta+1|-|n_\eta|)^2} \Big(\frac{|n_\eta|}{|n_\eta+1|}\Big)^N \Big(N+1-\frac{|n_\eta|}{|n_\eta+1|}N\Big) \;,
    \end{align}
    which goes to $0$ as $N\rightarrow\infty$ because $\frac{|n_\eta|}{|n_\eta+1|}<1$.
    Clearly $\sum_{j>N} \frac{|n_\eta|^j}{|n_\eta+1|^j}$ also goes to $0$ as $N\rightarrow\infty$ because $(\frac{|n_\eta|}{|n_\eta+1|})^k \leq k(\frac{|n_\eta|}{|n_\eta+1|})^k$ for all $k\in\N$.
    So by setting $N=n_\mathrm{max}$ and then plugging it into \eqref{eqn:thermal_input:hnmax_bound} we get the bound
    \begin{align}
        \Big| h_{n_\mathrm{max}}[\eta,F^*] \Big| &\leq \frac{2\log(n_\mathrm{min}+1) - \log n_\mathrm{min}}{|n_\eta+1|} \sum_{k>n_\mathrm{max}} \frac{|n_\eta|^k}{|n_\eta+1|^k} \; k \;.
    \end{align}
    
    Now consider a disk $D_r(z_0)$ and a complex number in it $\xi = \arg\max_{\eta\in D_r(z_0)} \sum_{k>N} \frac{|n_\xi|^k}{|n_\xi+1|^{k+1}}$.
    This can be obtained by setting $N$ sufficiently large $N > \max\{n_{\Re(\eta)+\delta}+1 : \eta\in D_r(z_0)\}$.
    So for all $\eta\in D_r(z_0)$ we can find an integer $M$ larger than such $N$ and larger than $n_\mathrm{max}$ to obtain
    \begin{align}
        \Big| h_M[\eta,F^*] \Big| &\leq \frac{2\log(n_\mathrm{min}+1) - \log n_\mathrm{min}}{|n_\xi+1|} \sum_{k>M} \frac{|n_\xi|^k}{|n_\xi+1|^k} \; k \;.
    \end{align}
    Because the right hand side goes to $0$ as $M\rightarrow\infty$, for any $\varepsilon>0$ we can find sufficiently large $M$ such that for all $l>M$ we have $|h_l[\eta,F^*]|<\varepsilon$ for all $\eta\in D_r(z_0)$.
    As this holds for any closed disk in $\C_\delta$ for any $\delta>0$, therefore $h[\cdot,F^*]$ is analytic over all positive reals $\R_{>0}$.
\end{proof}

% $|n_\eta| = n_\mathrm{env} + |\eta| |n_\res-n_\mathrm{env}|$ ? $|n_\eta+1|=|n_\eta|+1$ ?

Note that the analiticity of the marginal information density function still holds even if one of $n_\res,n_\mathrm{env}$ is zero by the following argument.
Without loss of generality assume that $n_\mathrm{env}=0$ (i.e. vacuum environment) and suppose that the optimal encoding is $F^*$, hence $n_\eta$ can only takes up value in $[0,n_\res]$.
In the above's proof this leads to a pathological case where $n_\mathrm{min}=n_\mathrm{env}=0$.
However we can still find a $z\in(0,n_\res)$ such that $Q_k(z)\leq \int_0^1 dF^*(\eta) Q_k(n_\eta) = P_k[F^*]$ for all $k>n_\res$ because $Q_k$ is monotonically increasing and it cannot be the case that $dF^*(\eta)=1$ for $\eta=0$ (or any other $\eta\in[0,1]$ for that matter), namely where the only codeword is the vacuum state of the environment as this leads to zero capacity.
One can see this by noting that $dF(0)=1$ implies that $\chi = S(\rho_\mathrm{ave}) - \int dF(\eta) S(\rho(\eta)) = S(\rho(0)) - S(\rho(0))=0$ because phase shift operation by the channel leaves the output state invariant for a thermal state input.
This means that $P_k[F^*]$ must take value in the interval $(0, Q_k(n_\res))$ and therefore for any $F^*$ we can always find some $z>0$ such that $Q_k(z) < P_k[F^*]$ for all $k>n_\res$.
Hence, we can still use the technique in the proof of Lemma~\ref{lem:thermal_input:marginal_info_density_analytic} and Lemma~\ref{lem:thermal_input:log_avg_proba_upper_bound} to upper bound the $h_N$ function to obtain the same analiticity result.

Now, as $i[\cdot,F^*]$ is analytic over positive reals by Lemma~\ref{lem:thermal_input:marginal_info_density_analytic}, by the identity theorem of analytic functions and the Bolzano-Weierstrass theorem, infinite points of increase over $[0,1]$ implies that $i[\eta,F^*] = \chi[F^*]$ for all positive real $\eta$.
Now we show that this leads to a contradiction, then conclude that it is not possible to have an infinite points-of-increase.

First, consider the case of $n_\res<n_\mathrm{env}$.
As we can write $n_\eta = n_\mathrm{env} + \eta (n_\res - n_\mathrm{env})$, then we have a negative number inside the $\log$ function in the $-(n_\eta+1)\log(n_\eta+1) + n_\eta\log n_\eta$ term of $i[\eta,F^*]$ as written in \eqref{eqn:thermal_input:marginal_info_density1} for sufficiently large $\eta$.
This implies that $i[\eta,F^*]=\chi[F^*]$ is not real for sufficiently large $\eta$, which is not possible.
On the other hand, for the case of $n_\res>n_\mathrm{env}$ we have $\lim_{\eta\rightarrow\infty} i[\eta,F^*]<0$ because
\begin{align*}
    \lim_{\eta\rightarrow\infty} -(1+n_\eta)\log(1+n_\eta) + n_\eta\log n_\eta = \lim_{\eta\rightarrow\infty} -\log(1+n_\eta) - n_\eta\log\frac{1+n_\eta}{n_\eta} <0
\end{align*}
because $\lim_{\eta\rightarrow\infty} n_\eta = \infty$ and $\lim_{x\rightarrow\infty} x\log\frac{1+x}{x}=1$, and because $\lim_{\eta\rightarrow\infty} P_n[\rho(\eta)] = 0$ for each $n\in\N$, implying that $-\sum_n P_n[\rho(\eta)]\log P_n[F^*] = 0$.
As an infinite points-of-increase over $[0,1]$ leads to contradictions, therefore there can only be finitely many points-of-increase.

% First note that the marginal information density can be written as a relative entropy
% \begin{align*}
%     i[\eta,F^*] = S(\rho(\eta)||\rho_\mathrm{ave}) = \sum_{k=0}^\infty P_k[\rho(\eta)] \log\frac{P_k[\rho(\eta)]}{P_k[F^*]} \;.
% \end{align*}

\section{Capacity of two-codeword encoding with a vacuum resource at $T>0$.} \label{app:vacThres}

Here we will derive equation~\eqref{eq:27}, which is the Holevo information over two-codeword encodings when the resource state is a vacuum state and the environment temperature is $T>0$ (i.e. $n_\text{env}>0$).
For convenience, we restate it below
\begin{align}
 \chi =
  \log\left( 1+\frac{n_\text{env}}{(1+n_\text{env})^\frac{1+n_\text{env}}{n_\text{env}}} \right) \;
\end{align}
where thermal photon number $n_\text{env}$ is equal to energy $E_0$, which is achieved when encoding $F(\eta)$ has two points of increase at $\eta=0$ and $\eta=1$. 
We denote the prior probabilities of sending the thermal state as $q_0$ (namely when transmissivity $\eta=0$) and of sending the vacuum state as $q_1$ (when transmissivity $\eta=1$) with $q_0+q_1=1$.
So, we define an encoding
\begin{align}
  \label{eq:29}
  F_{q_0}(\eta) =\begin{cases}
    q_0,& \text{for } 0 \leq \eta< 1\;,\\
    1,&\text{for } \eta=1
  \end{cases} \;.
\end{align}
We want to find $q_0$ and $q_1$ that maximise the Holevo quantity for such an encoding,
\begin{align}
  \label{eq:30}
  \chi[F_{q_0}] = S(\rho_\text{ave}) + q_0 S(\rho(0)) + q_1 \underbrace{S(\rho(1))}_{0}\;
\end{align}
where $S(\rho(1))=0$ because $\rho(1)$ is a vacuum state, which is a pure state.
The averaged state is diagonal in the Fock basis with
\begin{align}
  \label{eq:31}
  P_n[\rho_\text{ave}] =\begin{cases}
    q_0 P_0[\rho(0)] + q_1,& \text{for } n=0\;,\\
    q_0 P_n[\rho(0)], &\text{for } n\geq 1\;.
    \end{cases}
\end{align}
By direct computation, we have
\begin{multline}
  \label{eq:33}
  S(\rho_\text{ave})=q_0 S(\rho(0))-\big(q_0 P_0[\rho(0)]+q_1\big)\log\big(q_0
  P_0[\rho(0)]+q_1\big)
  +\big(q_0 P_0[\rho(0)]\big)\log \big(q_0 P_0[\rho(0)]\big) - q_0 \log q_0\;,
\end{multline}
which gives
\begin{multline}
  \label{eq:chi0}
  \chi[F_{q_0}]=-\big(q_0 P_0[\rho(0)]+q_1\big)\log\big(q_0
  P_0[\rho(0)]+q_1\big)
  +\big(q_0 P_0[\rho(0)]\big)\log \big(q_0 P_0[\rho(0)]\big) - q_0 \log q_0\;.
\end{multline}
To find the maximum value of $\chi$, we differentiate it with respect
to $q_0$ and solve for $q_0$ in
\begin{align}
  \label{eq:34}
  \frac{d\chi}{d q_0}&=0\;.
\end{align}
The value of $q_0$ that solves this equation is
\begin{align}
  \label{eq:35}
  q_0 &=\frac{1}{1-P_0[\rho(0)] +P_0[\rho(0)]^{\left(\frac{P_0[\rho(0)]}{P_0[\rho(0)]-1}\right)}} \\
      &=\frac{1+n_\text{env}}{n_\text{env}+(1+n_\text{env})^\frac{1+n_\text{env}}{n_\text{env}}} \;,
\end{align}
where we used $P_0[\rho(0)]=\frac{1}{1+n_\text{env}}$ to arrive at the last line. 
Substituting this into \eqref{eq:chi0} gives the desired result. 
As $n_\text{env}\rightarrow\infty$, the optimal $q_0$ tends to $0.5$ and $\chi\rightarrow 1$.

\section{Lossy channel with low-energy coherent state resource}\label{app:low_E_limit}

Here, we consider a coherent state resource $\ket{\alpha}$ which passes through a lossy channel, i.e. with an environment mean photon number $n_\mathrm{env}>0$ and a fixed attenuation $\eta_\mathrm{ch}$. 
We will derive an analytic expression of an approximation of the Holevo information \eqref{eq:1} in terms of attenuation $\eta_\mathrm{ch}$ and thermal photon number $n_\mathrm{ch} = (1-\eta_\mathrm{ch})n_\mathrm{env}$ for a single-ring circularly symmetric encoding given a small energy $E=|\alpha|^2\ll1$.

\mg{What does `amplitude' of a thermal state mean? Is this defined somewhere?}\at{OK, expained this below}
Note that this is equivalent to a scenario where we have a thermal state with mean photon number $n_\mathrm{ch}$ displaced by $\alpha$ as a resource state and an environment in thermal state with the same mean photon number $n_\mathrm{ch}$.
We can disregard the phase $\theta$ and simply write \eqref{eqn:codeword_thermal_photon_number} as
\begin{align}
    \rho(\eta_\mathrm{ch}) = D(\sqrt{\eta_\mathrm{ch}}\alpha) \rho_\th(n_\mathrm{ch}) D^\dagger(\sqrt{\eta_\mathrm{ch}}\alpha) \;.
\end{align}
because the entropy of $\rho(\eta_\mathrm{ch},\theta)$ is equal to the entropy of $\rho_\th(n_\mathrm{ch})$.
This is due to identical photon number distribution between codeword $\rho(\eta_\mathrm{ch},\theta)$ and thermal state $\rho_\th(n_\mathrm{ch})$, which is given by
\begin{align}
  \label{eq:39}
  P_n[\rho(\eta_\mathrm{ch})] = \frac{1}{1+n_\mathrm{ch}} \left( \frac{n_\mathrm{ch}}{1+n_\mathrm{ch}} \right)^n \;,
\end{align}
and the averaged state $\rho_\textup{ave}$ photon number distribution given by
\begin{align}
  \label{eq:40}
  P_n[\rho_\text{ave}] =\bra{n} D(\sqrt{\eta_\mathrm{ch}}\alpha) \rho_\th(n_\mathrm{ch})
  D^\dagger(\sqrt{\eta_\mathrm{ch}}\alpha) \ket{n} \;.
\end{align}
Hence we may write the Holevo information of the lossy channel as
\begin{align}
  \label{eq:38}
  \chi[E] = S(\rho_\text{ave}) - S(\rho(\eta_\mathrm{ch})) \;.
\end{align}

When $|\alpha|$ is small, we can approximate the displacement operator as
\begin{align}
  \label{eq:28}
  D(\alpha) \approx 1+\alpha(a^\dagger-a)+\frac{|\alpha|^2}{2}\left(a^\dagger-a\right)^2 \;,
\end{align}
allowing us to obtain the approximate probability of the average output
\begin{equation}
\begin{split}
  \label{eq:32}
  P_n[\rho_\text{ave}] 
  &\approx P_n + \eta_\mathrm{ch} |\alpha|^2 n P_{n-1} + \eta_\mathrm{ch} |\alpha|^2 (n+1)^2 P_{n+1} - \eta_\mathrm{ch} |\alpha|^2 (1+2n) P_{n} \;,
\end{split}
\end{equation}
where $P_n$ is short for $P_n[\rho(\eta_\mathrm{ch})]$.
Substituting this into~\eqref{eq:38}, we get
\begin{align}
  \label{eq:36}
  \chi[E] \approx \eta_\mathrm{ch} |\alpha|^2 \sum_n
  \left(nP_{n-1}+(n+1)^2P_{n+1}-(1+2n)P_n \right) \log P_n \;.
\end{align} 
Finally, substituting Eq.~(\ref{eq:39}) for $P_n$ to the right-hand side, and after some simplification, we arrive at
\begin{align}
  \label{eq:37}
  \chi[E] \approx \eta_\mathrm{ch} |\alpha|^2 \log\frac{1+n_\mathrm{ch}}{n_\mathrm{ch}} \;,
\end{align}
which is precisely the approximation $\Tilde{\chi}_\mathrm{tc}[E]$ in~\eqref{eq:1}.

The same result is obtained for the optimal resource state~(\ref{eq:optngs}), restated below for convenience 
\begin{align}
    \ket{\phi} = \frac{1}{\sqrt{E+1}}\sum_{n=0}^\infty \ket{n} \sqrt{\frac{E^n}{(E+1)^n}} \;.
\end{align}
One way to see this is by considering the
displaced squeezed state $D(\alpha)S(r)\ket{0}$ where $D$ is the
displacement operator and $S$ is the squeezing operator with
$\alpha=\sqrt{E}$ and $r=-(\sqrt{2}-1)E$, hence
\begin{align}
  \label{eq:4}
  D(\alpha)S(r) \ket{0} &= \ket{0} +\ket{1}\alpha + \ket{2}\frac{|\alpha|^2-r}{2}\\
  &=\ket{0} +\ket{1}\sqrt{E} + \ket{2} E\;,
\end{align}
which is a good approximation for the optimal transmitter codeword given a small $E$. Transmitting this (now Gaussian state) through the noisy channel, we have a squeezed Gaussian state at the output with amplitude $\sqrt{\eta_\mathrm{ch} E}$, thermal photon expectation $n_\mathrm{ch}$, and squeezing factor $-\frac{\eta_\mathrm{ch}(\sqrt{2}-1)E}{1+2n_\mathrm{ch}}$. 
The entropy of each codeword is the same as the entropy of a thermal state with $n_\mathrm{ch}$ photons. 
We also find that the photon number distribution of the averaged state is given by~(\ref{eq:32}) with $\alpha$ replaced by $\sqrt{E}$, so that we arrive at same result~(\ref{eq:37}) for the Holevo's quantity. We leave the derivations as an exercise for the interested reader.
Note that although the state $|\phi\>$ is a resource state that maximizes $\chi$ when $T=0$ (see Section~\ref{appendix:optimal-state}), it is not known whether it maximizes $\chi$ at $T>0$ environment nor when the channel is lossy.
The fact we established above that $\chi$ for both resource states $|\phi\>$ and $|\alpha\>$ coincide at $E\ll 1$, could however be a good starting point for future work on studying the thermal channel capacity at $T>0$.

\subsection{Upper bound on the lossy channel capacity}

Here we will proceed to derive an upper bound for $\chi$ for the lossy channel scenario which coincides with \eqref{eq:1}.
If each codeword is a Gaussian distribution of coherent states with average energy constraint $E$, this Gaussian distribution is optimal for a channel with no peak power constraint~\cite{Giovannetti2014}.
At the output of a thermal channel, the averaged state is also a Gaussian state with mean photon number $n_\text{ave}=\eta_\mathrm{ch} E+ n_\mathrm{ch}$ and codewords with mean photon number $n_\mathrm{ch}$. 
If we define $\delta_n=P_n[\rho_\text{ave}]-P_n[\rho(\eta_\mathrm{ch})]$
which is small when $E$ is small, we can approximate
$S(\rho_\text{ave})$ as
\begin{align}
  \label{eq:41}
  S(\rho_\text{ave}) &= -\sum_n P_n[\rho_\text{ave}] \log P_n[
                       \rho_\text{ave}]\\
  &= -\sum_n \left(P_n[\rho(\eta_\mathrm{ch})]+\delta_n\right) \log\left( P_n[
    \rho(\eta_\mathrm{ch})] + \delta_n \right)\\
  &=  -\sum_n \left(P_n[\rho(\eta_\mathrm{ch})]+\delta_n\right)\left( \log  P_n[
    \rho(\eta_\mathrm{ch})] + \frac{\delta_n}{ P_n[\rho(\eta_\mathrm{ch})]\ln 2} \right)\\
  &=S(\rho(\eta_\mathrm{ch})) - \sum_n \delta_n \log P_n[\rho(\eta_\mathrm{ch})]\;.
\end{align}
Using \eqref{eq:38} and the expression for $P_n[\rho(\eta_\mathrm{ch})]$ from~(\ref{eq:39}), the Holevo information is
\begin{align}\label{eq:42}
    S(\rho_\text{ave}) - S(\rho(\eta_\mathrm{ch})) &= -\sum_n \delta_n \left( n \log\frac{n_\mathrm{ch}}{1+n_\mathrm{ch}}-\log(1+n_\mathrm{ch})\right)\\
    &=\sum_n n \delta_n \log \frac{1+n_\mathrm{ch}}{n_\mathrm{ch}} \\
    &= \left(n_\text{ave}-n_\mathrm{ch}\right)\log \frac{1+n_\mathrm{ch}}{n_\mathrm{ch}}\\
    &= \eta_\mathrm{ch} E \log \frac{1+n_\mathrm{ch}}{n_\mathrm{ch}} \;,
\end{align}
which gives a tight upper bound for the lossy thermal channel capacity given resource state energy $E$.

\section{Wigner functions and photon number distribution plots}\label{app:Wigner_fun}

\begin{figure*}[!h]
  \centering
\pgfplotsset{compat=newest,height=7cm,width=7cm}
\pgfplotsset{
  table/search path={figures/wigPlot}}
\begin{tikzpicture}
  \begin{axis}
    [mywigaxis,legend pos=south west,
    y tick label style={/pgf/number format/.cd,
            fixed,fixed zerofill,precision=1,/tikz/.cd} ]
    \addplot+[mark=none,line width=1pt] table [col sep=comma,x index=0, y index=2]{wigOpt_e3_1d.csv};
    \addlegendentry{optimal state};
    \addplot+[mark=none,line width=1pt] table [col sep=comma,x index=0, y index=1]{wigOpt_e3_1d.csv};
    \addlegendentry{squeezed state};
  \end{axis}
\end{tikzpicture}
%%% Local Variables:
%%% mode: latex
%%% TeX-master: "../../cv_thermal_encoding.tex"
%%% End:
\pgfplotsset{
  table/search path={figures/wigPlot}}
\begin{tikzpicture}
  \begin{axis}[mywigaxis,ybar,restrict x to domain=0:16,
    bar width=3pt, % bar width
    ybar=0pt, % spacing between bars
         y tick label style={/pgf/number format/.cd,
          fixed,fixed zerofill,/tikz/.cd},
            enlarge y limits  = 0.02,enlarge x limits  = 0.05, tickwidth=0pt,
    xtick={0,2,...,20}]
	\addplot table [col sep=comma,x index=0,y index=1]{pnOpt_e3.csv};
        \addlegendentry{optimal state};
	\addplot table [col sep=comma,x index=0,y index=2]{pnOpt_e3.csv};
    \addlegendentry{squeezed state};
	\end{axis}
\end{tikzpicture}
%%% Local Variables:
%%% mode: latex
%%% TeX-master: "../../cv_thermal_encoding.tex"
%%% End:
\caption{\label{fig:wigOpt}(left) Wigner function of the averaged
  state with an optimal state in Eq.~(\ref{eq:optngs}) and a displaced
  squeezed state resource at $E=3$. Each displaced squeezed state
  codeword has $5.40$ dB of squeezing and an amplitude of
  $1.60$. (right) Photon number distribution for the average
  state. The optimal state's Fock state population follows a geometric
  distribution by construction. The fidelity between the optimal
  resource and the displaced squeezed state is 0.996.}
\end{figure*}
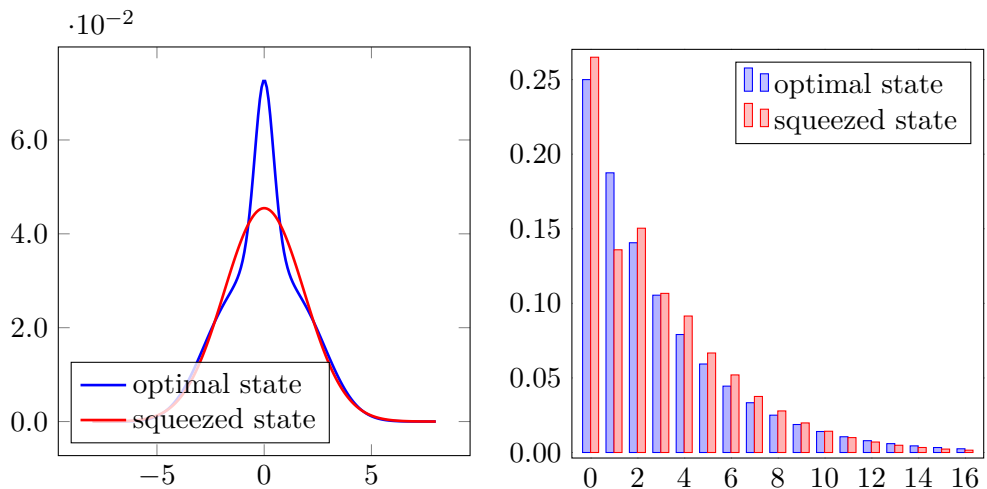

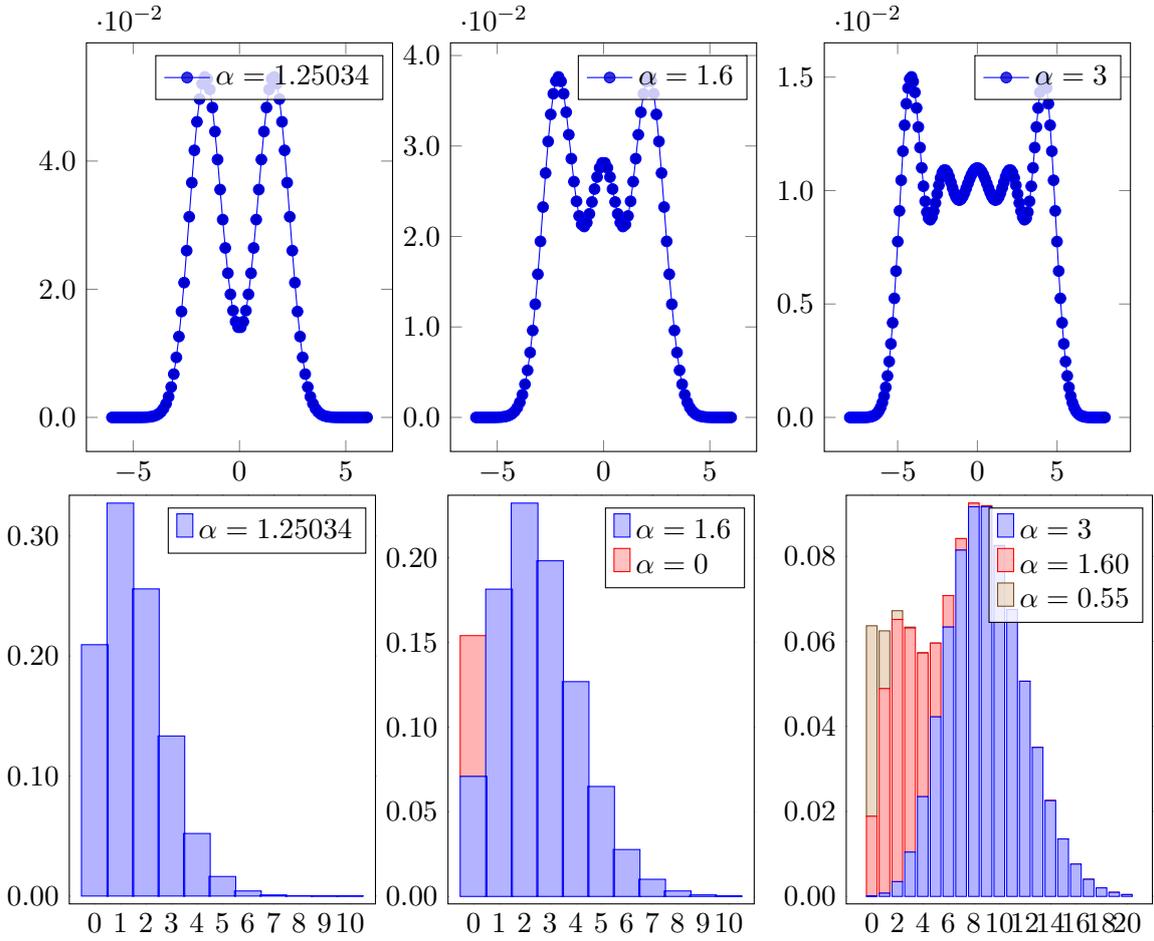
\begin{figure}[!h]
    \centering
    \pgfplotsset{compat=newest,height=7cm,width=7cm}
    \pgfplotsset{
  table/search path={figures/wigPlot}}
%   \begin{tikzpicture}
%   \begin{axis}
%     [y tick label style={/pgf/number format/.cd,
%             fixed,fixed zerofill,precision=1,/tikz/.cd} ]
%     \addplot+[] table [col sep=comma]{wig125_1d.csv};
%     \addlegendentry{$\alpha=1.25034$};
%     \addplot+[] table [col sep=comma]{wig16_1d.csv};
%     \addlegendentry{$\alpha=1.6$};
%     \addplot+[] table [col sep=comma]{wig30_1d.csv};
%     \addlegendentry{$\alpha=3$};
%   \end{axis}
% \end{tikzpicture}
\begin{tikzpicture}
  \begin{axis}
        [mywigaxis,width=0.33\columnwidth,y tick label style={/pgf/number format/.cd,
            fixed,fixed zerofill,precision=1,/tikz/.cd} ]
    \addplot+[] table [col sep=comma]{wig125_1d.csv};
    \addlegendentry{$\alpha=1.25034$};
  \end{axis}
\end{tikzpicture}%%no space
\begin{tikzpicture}
  \begin{axis}
        [mywigaxis,width=0.33\columnwidth,y tick label style={/pgf/number format/.cd,
            fixed,fixed zerofill,precision=1,/tikz/.cd} ]
    \addplot+[] table [col sep=comma]{wig16_1d.csv};
    \addlegendentry{$\alpha=1.6$};
  \end{axis}
\end{tikzpicture}
\begin{tikzpicture}
  \begin{axis}
    [mywigaxis,width=0.33\columnwidth,y tick label style={/pgf/number format/.cd,
            fixed,fixed zerofill,precision=1,/tikz/.cd} ]
    \addplot+[] table [col sep=comma]{wig30_1d.csv};
        \addlegendentry{$\alpha=3$};
  \end{axis}
\end{tikzpicture}
%%% Local Variables:
%%% mode: latex
%%% TeX-master: "../../cv_thermal_encoding.tex"
%%% End:
\\
    \pgfplotsset{
  table/search path={figures/wigPlot}}
\begin{tikzpicture}
  \begin{axis}[mywigaxis,width=0.33\columnwidth,ybar stacked,restrict x to domain=0:10,
        y tick label style={/pgf/number format/.cd,
          fixed,fixed zerofill,/tikz/.cd},
    enlarge y limits  = 0.02,enlarge x limits  = 0.1, tickwidth=0pt,
    xtick={0,...,10}]
    \addplot table [col sep=comma,x index=0,y index=1]{pn125.csv};
    \addlegendentry{$\alpha=1.25034$};
	\end{axis}
\end{tikzpicture}%no spaces
\begin{tikzpicture}
  \begin{axis}[mywigaxis,width=0.33\columnwidth,ybar stacked,restrict x to domain=0:10,
    y tick label style={/pgf/number format/.cd,
      fixed,fixed zerofill,/tikz/.cd},
        enlarge y limits  = 0.02,enlarge x limits  = 0.1, tickwidth=0pt,
    xtick={0,...,10}]
    \addplot table [col sep=comma,x index=0,y index=1]{pn16.csv};
    \addlegendentry{$\alpha=1.6$};
    \addplot table [col sep=comma,x index=0,y index=2]{pn16.csv};
    \addlegendentry{$\alpha=0$};
	\end{axis}
\end{tikzpicture}
\pgfplotsset{scaled y ticks=false}
\begin{tikzpicture}
  \begin{axis}[mywigaxis,width=0.33\columnwidth,ybar stacked,restrict x to domain=0:20,
    bar width=4pt,
        y tick label style={/pgf/number format/.cd,
          fixed,fixed zerofill,/tikz/.cd},
            enlarge y limits  = 0.02,enlarge x limits  = 0.1, tickwidth=0pt,
    xtick={0,2,...,20}]
	\addplot table [col sep=comma,x index=0,y index=1]{pn30.csv};
    \addlegendentry{$\alpha=3$};
	\addplot table [col sep=comma,x index=0,y index=2]{pn30.csv};
    \addlegendentry{$\alpha=1.60$};
	\addplot table [col sep=comma,x index=0,y index=3]{pn30.csv};
    \addlegendentry{$\alpha=0.55$};
	\end{axis}
\end{tikzpicture}
%%% Local Variables:
%%% mode: latex
%%% TeX-master: t
%%% End:
    \caption{(top) Wigner function for the average state of optimal encoding at $T=0$. (bottom) Photon number distribution for the average state of optimal encoding. \label{fig:wigPlot}}
\end{figure}
We plot out the cross-section (at $p=0$) of the Wigner function
$W(x,p)$ for some of the average states in the main text in
Fig~\ref{fig:wigPlot} (top).
When $|\alpha_\mathrm{max}|$ is less than 1.25034 the optimal encoding
is to use one ring at $\eta=1$. When $|\alpha_\mathrm{max}|$ is larger
than 1.25034 (but only up to a certain point), the optimal encoding now includes
adding a little bit of the vacuum state. As $|\alpha_\mathrm{max}|$
increases further, the optimal encoding will include more rings and
tends to become flatter. Figure~\ref{fig:wigPlot} (bottom) shows the
corresponding photon number distribution of the average
state. However, the flat distribution (Figure~\ref{fig:bigAlpHat}) is not optimal. It performs
worse than the 3-ring distribution.

\begin{figure*}[!ht]
  \centering
\pgfplotsset{compat=newest,height=7cm,width=7cm}
\pgfplotsset{
  table/search path={figures/wigPlot}}
\begin{tikzpicture}
  \begin{axis}
        [mywigaxis,y tick label style={/pgf/number format/.cd,
            fixed,fixed zerofill,precision=1,/tikz/.cd} ]
    \addplot+[] table [col sep=comma]{wig30Flat_1d.csv};
    \addlegendentry{$\alpha=3$};
  \end{axis}
\end{tikzpicture}
%%% Local Variables:
%%% mode: latex
%%% TeX-master: t
%%% End:
\pgfplotsset{
  table/search path={figures/wigPlot}}
\begin{tikzpicture}
  \begin{axis}[mywigaxis,ybar stacked,restrict x to domain=0:20,
    bar width=4pt,
        y tick label style={/pgf/number format/.cd,
          fixed,fixed zerofill,/tikz/.cd},
            enlarge y limits  = 0.02,enlarge x limits  = 0.1, tickwidth=0pt,
    xtick={0,2,...,20}]
	\addplot table [col sep=comma,x index=0,y index=1]{pn30Flat.csv};
    \addlegendentry{$\alpha=3$};
	\end{axis}
\end{tikzpicture}
%%% Local Variables:
%%% mode: latex
%%% TeX-master: "../cv_thermal_encoding.tex"
%%% End:
\caption{Wigner function (left) and photon number distribution of a flat distribution (with the resource state as the coherent state $\ket{3}$). But this is not optimal and does not perform as well as the 3-ring distribution in Figure~\ref{fig:wigPlot}. \label{fig:bigAlpHat}}
\end{figure*}
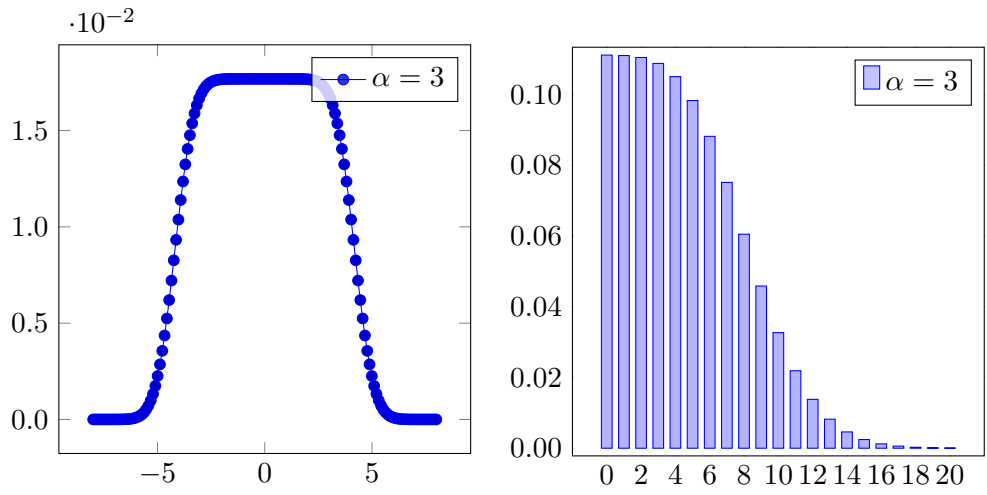

\end{document}